\documentclass[10pt,onecolumn]{article}
\usepackage[a4paper, hmargin={2.8cm, 2.8cm}, vmargin={2.5cm, 2.5cm}]{geometry}

\usepackage[utf8]{inputenc}

\usepackage{amsmath}
\usepackage{amssymb}  
\usepackage{amsthm}
\usepackage{amsfonts}
\usepackage{graphicx}
\usepackage{amscd,dsfont,stmaryrd}
\usepackage{cancel}
\usepackage{mathrsfs}

\usepackage{bbm}
\usepackage{url}
\usepackage{enumerate} 
\usepackage{ upgreek }
\usepackage{dsfont}
\usepackage{color}
\usepackage{subcaption}
\usepackage{comment}
\usepackage{makecell}
\usepackage{diagbox}
\usepackage{wrapfig}
\usepackage{rotating}
\usepackage{tabularx}

\usepackage{hyperref}
\hypersetup{
    colorlinks=true,
    linkcolor=blue,
    citecolor=red,
    filecolor=magenta,      
    urlcolor=cyan,
}

\usepackage{mathrsfs}

\newcommand{\tr}{\textnormal{Tr}}
\newcommand{\sC}{\mathscr{C}}
\newcommand{\sH}{\mathscr{H}}

\newcommand{\sD}{\mathscr{D}}

\newcommand{\id}{\ensuremath{\mathds{1}}}

\newcommand{\cB}{\mathcal{B}}
\newcommand{\cC}{\mathcal{C}}
\newcommand{\cD}{\mathcal{D}}

\newcommand{\cG}{\mathcal{G}}
\newcommand{\cH}{\mathcal{H}}

\newcommand{\cK}{\mathcal{K}}
\newcommand{\cL}{\mathcal{L}}
\newcommand{\cM}{\mathcal{M}}

\newcommand{\cP}{\mathcal{P}}
\newcommand{\cQ}{\mathcal{Q}}

\newcommand{\cS}{\mathcal{S}}
\newcommand{\cT}{\mathcal{T}}

\newcommand{\cV}{\mathcal{V}}
\newcommand{\cW}{\mathcal{W}}
\newcommand{\cX}{\mathcal{X}}

\def\one{\id}

\def\bx{{\bf x}}
\def\by{{\bf y}}

\newcommand\C{\mathbbm{C}}

\newcommand\R{\mathbbm{R}}
\newcommand\N{\mathbbm{N}}

\theoremstyle{plain}
\newtheorem{thm}{Theorem}[section]
\newtheorem{lem}[thm]{Lemma}

\newtheorem{defn}[thm]{Definition}

\theoremstyle{definition}
\newtheorem{example}{Example}[section]

\newtheorem{rem}{Remark}[section]

\renewcommand{\leq}{\leqslant}
\renewcommand{\geq}{\geqslant}

\begin{document}

\title{Boundary-driven quantum systems near the Zeno limit: steady states and long-time behavior}

\author{Eric A. Carlen\footnote{Department of Mathematics, Rutgers University, Piscataway, NJ 08854-8019 USA}, \ 
David A. Huse\footnote{Department of Physics, Princeton University, Princeton, NJ 08544, USA}\ \  and\   Joel L. Lebowitz
\footnote{Departments of Mathematics and Physics, Rutgers University, Piscataway, NJ 08854-8019 USA} 
}

\maketitle

\begin{abstract}
We study composite open quantum systems with a finite-dimensional state space $\cH_{AB} = \cH_A\otimes \cH_B$ 
governed by a Lindblad equation 
$\displaystyle{\frac{{\rm d}}{{\rm d}t}\rho(t) = \cL_\gamma \rho(t)}$ 
where $\cL_\gamma\rho = -i[H,\rho] + \gamma \cD \rho$.
Here,  $H$ is a Hamiltonian on $\cH_{AB}$ while $\cD$ is a dissipator $\cD_A\otimes \one$  acting non-trivially only on part $A$ of the system, which can be thought of as the boundary, and $\gamma$ is a parameter.   It is known that the dynamics may  simplify as the Zeno limit, $\gamma \to \infty$, is approached, so that after a initial time of order 
$\gamma^{-1}$, $\rho(t)$ is well approximated by $\pi_A\otimes R(t)$ where $\pi_A$ is a density matrix on $\cH_A$ such that $\cD_A\pi_A =0$, and $R(t)$  is an approximate solution of 
${\displaystyle \frac{{\rm d}}{{\rm d}t}R(t) = \cL_{P,\gamma}R(t)}$ 
where $\cL_{P,\gamma} R :=  -i[H_P,R] + \gamma^{-1} \cD_P R$ 
with $H_P$ being a Hamiltonian on $\cH_B$ and $\cD_P$ being a Lindblad generator acting on density matrices on $\cH_B$.  We give a rigorous proof of this holding in greater generality than in previous work; we assume only that $\cD_A$ is ergodic and gapped. 
Moreover, we precisely control the error terms, and use this to show
that the mixing times of $\cL_\gamma$ and $\cL_{P,\gamma}$ are tightly 
related near the Zeno limit. Despite this connection, 
the errors in the approximate
description of the evolution accumulate on times of order $\gamma^2$,
so it is difficult to directly access steady states $\bar\rho_\gamma$
of $\cL_\gamma$ 
through study of $\cL_{P,\gamma}$. In order to better control the long
time behavior, and in particular the steady states $\bar\rho_\gamma$, we introduce a third Lindblad generator
$\cD_P^\sharp$ that does not involve $\gamma$, but is still closely
related to $\cL_\gamma$ and $\cL_{P,\gamma}$. We show that if 
$\cD_P^\sharp$ is ergodic and gapped, then so are $\cL_\gamma$ and
$\cL_{P,\gamma}$ 
for all large $\gamma$, and in this case, if $\bar\rho_\gamma$ denotes the
unique steady state for $\cL_\gamma$, then 
$\lim_{\gamma\to\infty}\bar\rho_\gamma = \pi_A\otimes \bar R$ where $\bar R$ is the unique steady state for
$\cD_P^\sharp$. We further show that there is a trace norm convergent expansion 
${\displaystyle
\bar\rho_\gamma = \pi_A\otimes\bar R +\gamma^{-1} \sum_{k=0}^\infty \gamma^{-k} \bar n_k}$
where, defining $\bar n_{-1} :=  \pi_A\otimes\bar R$, $\cD \bar n_k = -i[H,\bar n_{k-1}]$ for all $k\geq 0$.   
Using properties of $\cD_P$ and $\cD_P^\sharp$, we show that this system of equations has a unique solution, 
and prove convergence. This is illustrated in a simple example for which one can compute $\bar\rho_\gamma$, and can carry out the expansion explicitly.

\end{abstract}

\tableofcontents

\section{Introduction}

In this paper we describe new exact results about the time evolution and steady states of a quantum system strongly coupled to a Markovian external reservoir.  The reservoir acts only on part of the system, e.g. on its surface.  More generally we consider a system consisting of two parts, $A$ and $B$ with finite-dimensional Hilbert space $\cH_{AB} = \cH_A\otimes\cH_B$, whose density matrix $\rho$ evolves according to the Lindblad equation
\begin{equation}\label{LindEq11}
\frac{{\rm d}}{{\rm d}t}\rho(t) = \cK\rho + \gamma \cD \rho(t) =: \cL_\gamma \rho(t)~.
\end{equation}
Here,
\begin{equation}\label{unitary}
\cK\rho := -i[H,\rho] 
\end{equation}
 describes the unitary evolution of the isolated system due to its time-independent Hamiltonian $H$ on $\cH_{AB}$, while $\cD$ is a dissipative Lindblad  generator of the form
\begin{equation*}
\cD := \cD_A\otimes \one_B~,
\end{equation*}
where $\cD_A$ acts only on $\cH_A$, here $\one_B$ is the identity super-operator on $\cH_B$, and $\gamma$ is a coupling constant.  
By Lindblad's Theorem \cite{Lin76}, a super-operator $\cD_A$ is such that $e^{t\cD_A}$ is a completely positive trace preserving (CPTP) operation if and only if $\cD_A$ can be written in the form
\begin{equation*}
\cD_A\rho = \sum_{j=1}^r \left( 2 L_j\rho L_j^\dagger - L_j^\dagger L_j\rho -
\rho L_j^\dagger L_j \right)- i[K_L,\rho]
\end{equation*}
for some operators $L_1,\dots,L_r$ on $\cH_A$, often called {\em jump operators} and some self-adjoint operator $K_L$ on $\cH_A$. We refer to such super-operators as {\em Lindblad generators}. 

We always  assume that 
 $\cD_A$ is ergodic so that there is a unique density matrix $\pi_A$ on $\cH_A$ satisfying $\cD_A\pi_A = 0$.
We also assume that $\cD_A$ is {\em gapped}, meaning that the real parts of all of its nonzero eigenvalues $\lambda_j$ satisfy $\Re(\lambda_j) \leq -a$, $a>0$.
Since $\cD_A$ is a Lindblad generator, the two assumptions together are equivalent to $0$ being a non-degenerate eigenvalue of $\cD_A$, and $\cD_A$ not having any eigenvalues that are purely imaginary.  

Without loss of generality, a constant has been added to the Hamiltonian so that $\tr H=0$~, $H=H_A\otimes\one_B+H_{AB}+\one_A\otimes H_B$~; here  $\one_A$ and $\one_B$ are the identity operators on those subspaces, $H_A=(1/d_B)\tr_B H$, $H_B=(1/d_A)\tr_A H$, and $d_A$ and $d_B$ are the dimensions of $\cH_A$ and $\cH_B$.

In a recent paper \cite[Theorem 1]{CHL}, we proved that for such systems, if $\bar \rho$ is a steady state solution of \eqref{LindEq11} and $[\bar\rho, H] =0$, then necessarily $\bar \rho$ is a product state of the form 
$\bar\rho = \pi_A \otimes\pi_B$ for some density matrix $\pi_B$ on $\cH_B$. It then follows \cite[Corollary 2]{CHL} that a steady state $\bar \rho$ can be a Gibbs state $\frac{1}{Z_\beta } e^{-\beta H}$ for the system Hamiltonian
only if $\beta H_{AB}= 0$.  In particular, when the system Hamiltonian does couple the two parts of the system so $H_{AB}\neq 0$, no steady state of \eqref{LindEq11} will be a Gibbs state with $\beta\neq 0$.  The question then arises: What is the nature of the steady-state solutions of \eqref{LindEq11}, and what can be said about the approach to them in the limit $t\to\infty$?

We study  these questions for large $\gamma$.  The limit  $\gamma\rightarrow\infty$ is known as the Zeno limit.  In this limit, there is a reduced description of the dynamics that has been investigated in \cite{PEPS,PP21,PPS25,ZC}.

  Because $\cD_A$ is ergodic and gapped, every density matrix $\rho$ on $\cH_{AB}$ satisfying $\cD \rho=0$ has the form
$\rho = \pi_A\otimes R$ for some density matrix $R$ on $\cH_B$.   Define $\cM$, the {\em steady-state manifold}, by
\begin{equation}\label{SSM}
\cM := \{\pi_A\otimes R\ :\  R\ {\rm is \ a \ density\ matrix\ on}\ \cH_B\}\ .
\end{equation}
Moreover, for all density matrices $\rho_0$ on  $\cH_{AB}$, ${\displaystyle   \lim_{t\to\infty}e^{t\cD}\rho_0}$ exists and belongs to $\cM$.  The super-operator
\begin{equation}\label{QCP}
\cP := \lim_{t\to\infty}e^{t\cD}\ 
\end{equation}
is a projection, though not necessarily an orthogonal projection. Its range consists of all  operators on $\cH_{AB}$ of the form $\pi_A\otimes Y$ for some operator $Y$ on $\cH_B$.  

Under these assumptions, it was shown in \cite{ZC}  that after a relaxation time $t_\gamma$ of order $\gamma^{-1}$, $\rho(t) = e^{t\cL_\gamma}\rho_0$
has the form
\begin{equation}\label{INTR1}
\rho(t) = \pi_A\otimes R(t) + \mathcal{O}(\gamma^{-1})\   
\end{equation}
uniformly in $t \geq t_\gamma$~,
where $R(t) = \tr_A \rho(t)$ is a density matrix on $\cH_B$.  That is, solutions of \eqref{LindEq11} for a general initial state $\rho_0$ and large $\gamma$ rapidly approach $\cM$ and then stay 
close  to $\cM$, uniformly in time.  In Section~\ref{PROJMO} we give simple  proofs of results of this genre in which ``closeness'' is measured in terms of the trace norm, denoted by $\|\cdot \|_1$; see Section~\ref{PRELIM}. 
In particular, we prove in Theorem~\ref{TZCVS} that  for all $\rho_0 = \pi_A \otimes R_0 \in \cM$, there is a constant $C>0$ independent of $\gamma$ such that for all $t>0$, 
\begin{equation}\label{INTR1C}
\| e^{t \cL_\gamma}\pi_A\otimes R_0  - \cP e^{t \cL_{\gamma}} \pi_A\otimes R_0\|_1 \leq \frac{C}{\gamma}\ ,
\end{equation}
which quantifies the sense in which solutions of \eqref{LindEq11} with initial data in $\cM$ stay close to $\cM$. A result of this type, in a more general framework but with a more complicated proof may be found in \cite{ZC}.

To deal with general initial data, we must take the initial relaxation time into account. 
Define  $t_\gamma := \frac{2\log(\gamma)}{a\gamma}$ where $a$ is the spectral gap of $\cD_A$.
We prove in Theorem~\ref{EULLIM}  that there is a constant $C>0$ independent of $\gamma$ such that 
for all $t \geq t_\gamma$ and all $\rho_0$, 
\begin{equation}\label{INTR1B}
\| e^{t \cL_\gamma}\rho_0 - e^{t \cL_{\gamma}} \cP \rho_0\|_1 \leq \frac{\log(1+\gamma)}{\gamma}C\ .
\end{equation}
Note that in \eqref{INTR1B}, the projector $\cP$ is applied before $e^{t\cL_\gamma}$ while in \eqref{INTR1C}, $\cP$ is applied after $e^{t\cL_\gamma}$.
These complementary results from Section~\ref{PROJMO} are important for the study of approximate equations of motion, to which we now turn.

It was shown in \cite{ZC} that, after the initial relaxation time, $R(t)$ satisfies
\begin{equation}\label{INTR2}
\frac{{\rm d}}{{\rm d}t}R(t) = 
- i[H_P,R(t)] +{\mathcal O}(\gamma^{-1})\ ,
\end{equation}
where
\begin{equation}\label{INTR3}
H_P = \tr_A[(\pi_A\otimes \one_B) H]
\end{equation}
is a self-adjoint operator called the {\em projected Hamiltonian}. The coherent evolution equation obtained by neglecting the ${\mathcal O}(\gamma^{-1})$ term in  \eqref{INTR2} cannot be expected to give an accurate
description of the evolution of $R(t)$ on long time scales of order
$\gamma$ since by then the error terms may accumulate to an effect of order one, and we expect to see dissipation on long times scales. 

A description of the evolution of $R(t)$  on time scales of order $\gamma$ was obtained in \cite{PEPS}.  The approximate equation obtained there is
\begin{equation}\label{INTR4}
\frac{{\rm d}}{{\rm d}t}R(t) = 
- i[H_P,R(t)] - \gamma^{-1}\tr_A[\cK \cS \cK(\pi_A\otimes R(t))] 
+{\mathcal O}(\gamma^{-2})
\end{equation}
where $\cS$ denotes a certain generalized inverse of $\cD$ that will be precisely defined in Section~\ref{PRELIM} below.  Define the super-operators $\cK_P$ and $\cD_P$  acting on operators on $\cH_B$ by
\begin{equation}\label{KPDPDEF}
\cK_P R := -i[H_P,R] \quad{\rm and}\quad  \cD_P R  := -\tr_A[\cK \cS \cK(\pi_A\otimes R]\ .
\end{equation}
These operators are studied in Section~\ref{SUPKPDP}.
It is evident that $\cK_P$ generates a CPTP evolution. 
Under  additional assumptions, the most important of which is a certain positivity condition, it was  shown in \cite{PEPS} that $\cD_P$ is a Lindblad generator. 
In this case, $\cL_{P,\gamma}$  defined by
\begin{equation}\label{INTR6}
\cL_{P,\gamma} := \cK_P + \gamma^{-1}\cD_P\ 
\end{equation}
is a Lindblad generator, and hence $e^{t\cL_{P,\gamma}}$ is also CPTP for all $t\geq 0$.

We show here that under only  the assumption that $\cD_A$ is ergodic 
and gapped, the positivity condition required in \cite{PEPS} is 
always satisfied, so that $\cD_P$  is always a Lindblad generator; see Theorem~\ref{ALWAYS}. Thus we 
may re-write \eqref{INTR4} as
\begin{equation}\label{INTR7}
\frac{{\rm d}}{{\rm d}t}R(t) = 
\cL_{P,\gamma}R(t) +{\mathcal O}(\gamma^{-2})\ ,
\end{equation}
and the equation ${\displaystyle \frac{{\rm d}}{{\rm d}t}\widetilde{R}(t) =  \cL_{P,\gamma}\widetilde{R}(t)}$ describes a CPTP evolution that closely tracks the projection onto $\cM$ of the evolution generated by $\cL_\gamma$.  We prove the following precise expression of this: Let $R_0$ be any density matrix on $\cH_B$ and define  $R(t) := \tr_A  e^{t\cL_\gamma}\pi_A\otimes R_0$ so that $\cP (e^{t\cL_\gamma}\pi_A\otimes R_0) = \pi_A\otimes R(t)$.  
Making essential use of the fact that $\cL_\gamma$ generates a CPTP evolution, we  prove in  Theorem~\ref{COHERENTSC}  that  there exists a constant $C$ depending only on $\cD_A$  and  the operator norm of $H$ such for all $T>0$, 
\begin{equation}\label{COHERENTSC1X}
 \| R(t) - e^{t\cL_{P,\gamma}}R_0\|_1 \leq  \frac{1}{\gamma^2} CT  \ 
 \end{equation}
for all $0 \leq t \leq T$.  

The generator in \eqref{INTR6} contains only the first order corrections to the coherent evolution. At the end of Section~\ref{APPROXEQ}, we derive the explicit form of the next order corrections, yielding the approximate equation
\begin{equation}\label{INTR51}
\frac{{\rm d}}{{\rm d}t} R = \cK_P R(t) +\gamma^{-1} \cD_P R(t) + \gamma^{-2}\cB_P R(t)\ .
\end{equation}
(See equation \eqref{BURDEF} for the explicit form of $\cB_P$). However, a simple 2 qubit example (See Example~\ref{EXAMP1})  shows that it can be the case that \eqref{INTR51} does not describe a CPTP evolution for any $\gamma>0$. 
That is, $\cK_P  +\gamma^{-1} \cD_P  + \gamma^{-2}\cB_P$ cannot be put in 
Lindblad form for any $\gamma>0$. 

We are particularly interested in stationary states of 
\eqref{LindEq11}; that is 
density matrices $\bar\rho_\gamma $ on $\cH_{AB}$ that satisfy
$\cL_\gamma \bar \rho_\gamma =0$.  These may be realized in the 
long time limit: $\bar\rho_\gamma=\lim_{t\to\infty}\rho(t)$, 
when $\cL_\gamma$ has no purely imaginary eigenvalues, and always by
${\displaystyle \bar\rho_\gamma = \lim_{T\to\infty}T^{-1}\int_0^T\rho(t){\rm d}t}$.

We may apply \eqref{COHERENTSC1X} to study the long time behavior of solutions of \eqref{LindEq11} on account of \eqref{INTR1C} and \eqref{INTR1B}: 
Let $\rho_0$ be any density matrix on $\cH_{AB}$, and define $R_0 := \tr_A[\rho_0]$ so that
$\cP \rho_0 = \pi_A\otimes R_0$.  Then by the triangle inequality, 
\begin{eqnarray*}
\| e^{t\cL_\gamma}\rho_0 - \pi_A\otimes e^{t \cL_{P,\gamma}}R_0\|_1 &\leq& \| e^{t\cL_\gamma}\rho_0 - e^{t \cL_{\gamma}}\cP\rho_0\|_1\\
  &+& \|  e^{t \cL_{\gamma}}\pi_A\otimes R_0 -  \cP(e^{t \cL_{\gamma}}\pi_A\otimes R_0)\|_1 \\
   &+& \|  \cP(e^{t \cL_{\gamma}}\pi_A\otimes R_0) - \pi_A\otimes e^{t \cL_{P,\gamma}}R_0\|_1 \ .
\end{eqnarray*}
Using  \eqref{INTR1B},  \eqref{INTR1C} and \eqref{COHERENTSC1X} to bound the three terms on the right leads to the proof of Theorem~\ref{MTILRM} which gives a tight relation between the evolutions governed by $\cL_\gamma$ and  by $\cL_{P,\gamma}$. 
\begin{equation}\label{BTM4INTR}
 \| e^{t\cL_\gamma}\rho_0 - \pi_A\otimes e^{t \cL_{P,\gamma}}R_0\|_1 \leq C \frac{\log(1+\gamma) +1+T}{\gamma}  ,
\end{equation}
uniformly on $[\epsilon, \gamma T]$ for any $0 < \epsilon < T$ for all $\gamma$ such that $ \frac{2\log(\gamma)}{a\gamma} \leq  \epsilon$. 
In particular, for all $t>0$,
\begin{equation}\label{BTM9}
\lim_{\gamma\to\infty} \| e^{t\cL_\gamma}\rho_0 - \pi_A\otimes e^{t \cL_{P,\gamma}}R_0\|_1 = 0\ .
\end{equation}
Then since for large $\gamma$, $\cL_{P,\gamma}$ does not differ much from $\cK_P$, it  follows that (see Theorem~\ref{MTILRMEUL})
\begin{equation}\label{BTMB}
\lim_{\gamma\to\infty}  e^{t\cL_\gamma}\rho_0 = \pi_A\otimes e^{t \cK_P} R_0\
\end{equation}
where the convergence is uniform on $\epsilon \leq t \leq T$ for any fixed $0 < \epsilon < T$. 
In other words, in the Zeno limit, {\em all of the dissipation takes place instantly} during the initial passage from $\rho_0$ to $\cP \rho_0 = \pi_A\otimes R_0$, and from then on the evolution is coherent.

The situation near, but not at, the Zeno limit is 
more interesting, and then there is dissipation on time scales of order $\gamma$, and the more precise bound  \eqref{BTM4INTR} can be used to relate the rates of approach to stationarity
for $\cL_\gamma$ and $\cL_{P,\gamma}$. However, the analysis of  both of these
Lindblad generators is complicated because  in both the  coherent and dissipative effects operate on different times scales.  For instance, with  $\cL_{P,\gamma}$ the coherent evolution due to $\cK_P$ acts on time scales of order one, while the dissipative evolution due to $\cD_P$ acts on time scale $\mathcal{O}(\gamma)$.  

To overcome this we introduce a new Lindblad generator $\cD_P^\sharp$,  which is independent of $\gamma$.
The new Lindblad generator 
 $\cD_P^\sharp$  is derived from  $\cD_P$ and $\cK_P$ through an averaging
operation introduced by Davies \cite{Da74} in his derivation of Lindblad
evolution equations in the weak coupling limit. 
The $\sharp$ notation is Davies'. More specifically, for 
any super-operator $\cT$ on $\cH_B$, define
\begin{equation}\label{INTR11}
\cT^\sharp(X) = \lim_{T\to\infty}\frac{1}{2T}
\int_{-T}^T e^{-t\cK_P}\cT(
e^{t\cK_P}X){\rm d}t\ .
\end{equation}
This limit always exists, and the rate at which convergence 
takes place in \eqref{INTR11} can be bounded in terms of the
spectral properties $H_P$.  Davies' averaging operation is studied in Section~\ref{DAOP},
 and we apply these results in Section~\ref{InteractionPicture}. 

In Section~\ref{InteractionPicture} we study $e^{-t\cK_P} e^{t\cL_{P,\gamma}}R_0$  in the Zeno limit, $\gamma \to \infty$, for general initial data $R_0$.  We run the initial data 
forward under 
the dissipative time evolution governed by $\cL_{P,\gamma}$~, and then backwards under the coherent evolution generated by $\cK_P$~.  In the language of scattering theory, we are passing to 
the ``interaction picture'' to compare the coherent evolution and its perturbation by a weak dissipation.  This is what Davies \cite{Da74} did in his derivation of Lindblad equations  in the weak coupling limit, 
and hence it is no surprise that his averaging operation introduced there shows up again here.  

Since
 $\cK_P$ and $\cL_{P,\gamma}$ differ by a term of order $\gamma^{-1}$, 
 we expect $e^{-t\cK_P} e^{t\cL_{P,\gamma}}$ to differ significantly
 from the identity only for times of order $\gamma$. It is therefore
 natural to introduce the rescaled time $\tau$ given by 
\begin{equation}\label{RESCALE}
t = \gamma\tau\ .
\end{equation}
 We shall then show that
 \begin{equation}\label{DALIMP8}
\lim_{\gamma\to\infty} e^{-\gamma\tau \cK_P} e^{\gamma\tau\cL_{P,\gamma}} = e^{\tau \cD_P^\sharp}\ .
 \end{equation}
 The rescaling of time \eqref{RESCALE} is different than in Davies' work because the Zeno limit is not the weak coupling limit, but as in 
 \cite{Da74} we use the elementary theory of perturbation of Volterra integral equations to 
compare the evolutions. 

The precise version of \eqref{DALIMP8}, Theorem~\ref{PROJMOZLTH},   says that
when $\cD_A$ is ergodic and gapped,  then for all $T>0$ there is a constant $C$ independent of $T$ and $\gamma$ such that for all density matrices $R_0$ on $\cH_B$, 
\begin{equation}\label{PROJMOZLTH10INTR} 
\| e^{-\gamma\tau \cK_P} e^{\gamma\tau   \cL_{P,\gamma}}R_0 - e^{\tau \cD_P^\sharp}R_0\|_{1} \leq \frac{1}{\gamma} CT e_{\phantom{0}}^{CT} 
\end{equation}
for all $0 \leq \tau \leq T$.  

This does not directly refer to the evolution governed by $\cL_\gamma$, but we
use this and  \eqref{BTM4INTR} to prove Theorem~\ref{PROJMOZLTHA}  which says that there are finite constants $C_0$ and $C_1$ independent of $\gamma$ such that for all $T>0$, and all density matrices $R_0$ on $\cH_B$, 
\begin{equation}\label{PROJMOZLTHA10INTR} 
\| e^{-\tau \gamma \cK_P} \tr_A e^{\gamma\tau   \cL_{\gamma}}\pi_A\otimes R_0   - e^{\tau \cD_P^\sharp}R_0\|_{1} \leq \frac{1}{\gamma} C_0T e_{\phantom{0}}^{C_1T} \ .
\end{equation}

The results \eqref{PROJMOZLTH10INTR} and \eqref{PROJMOZLTHA10INTR}  allow us to study the evolutions governed by $\cL_\gamma$ and $\cL_{P,\gamma}$ in terms of the evolution governed by $\cD_P^\sharp$. 
In particular, we use \eqref{PROJMOZLTH10INTR} and \eqref{PROJMOZLTHA10INTR} to show that if $\cD_P^\sharp$ is ergodic and gapped, then so are $\cL_\gamma$ and $\cL_{P,\gamma}$ for all sufficiently large $\gamma$, and moreover, we prove a tight relation among the rates of approach to stationarity for all three generators. 

\medskip

This is best done in terms of the {\em mixing time} which plays an important role in several recent investigations of open quantum systems \cite{FLT25,TKRWV,TZ25,V25}.  We recall the relevant  definitions:

For any two density matrices $\rho_0$, $\rho_1$ on 
$\cH$, define their {\em total variation distance} 
\begin{equation}\label{TVDDEF}
d_{{\rm TV}}(\rho_0,\rho_1) = \frac12 \|\rho_0- \rho_1\|_1\ ,
\end{equation} 
and hence $d_{{\rm TV}}(\rho_0,\rho_1) \leq 1$ 
with equality if and only if $\rho_0$ and $\rho_1$ are supported on mutually orthogonal subspaces of $\cH$. 

Let $\cL$ be a Lindblad generator acting on operators on $\cH$ a finite-dimensional Hilbert space, and let $0 < \epsilon < \tfrac12$.  
The $\epsilon$-{\em mixing time} of $\cL$ is defined here by
\begin{equation}\label{MIXTIMEDEF}
t_{{\rm mix}}(\cL,\epsilon) := \inf\left\{ \ t > 0\ :\ 
d_{{\rm TV}}(e^{t\cL}\rho_0,e^{t\cL}\rho_1) <  \epsilon\ \quad{\rm for\ all\ density\ matrices}\ \rho_0\ ,\ \rho_1 \right\}~,
\end{equation}
with the convention that $t_{{\rm mix}}(\cL,\epsilon) =\infty$ if there is no finite time $t$ such that  
$d_{{\rm TV}}(e^{t\cL}\rho_0,e^{t\cL}\rho_1) <  \epsilon$ for all 
pairs of density matrices $\rho_0,\rho_1$. 
It is known that if $t_{{\rm mix}}(\cL,\epsilon) <\infty$ for some $0 < \epsilon < \tfrac12$, then this is true for all such $\epsilon$; see 
Theorem~\ref{TMIXK}.

This definition does not
require $\cL$ to be ergodic, but $t_{{\rm mix}}(\cL,\tfrac14)<\infty$  if 
and only if $\cL$ is ergodic. We will use comparison of mixing times
to prove ergodicity, hence this definition, which does not refer to a
putative unique steady state $\pi$, is the one adapted to our applications. 
  In our finite dimensional setting, a spectral gap implies  finite mixing times, but the mixing times do not depend on the spectral gap $a$ and $\epsilon$  alone . For these facts and more on mixing times, see Appendix~\ref{MTAPP}.

Returning to $\cL_\gamma$, $\cL_{P,\gamma}$ and $\cD_P^\sharp$, we prove in 
Theorem~\ref{TMIXZL3} that for all but at most countably many values of $\epsilon_0\in (0,\tfrac12)$,
$$
\lim_{\gamma\to\infty} \frac{t_{{\rm mix}}(\cL_{\gamma},\epsilon_0)}{\gamma} = 
\lim_{\gamma\to\infty} \frac{t_{{\rm mix}}(\cL_{P,\gamma},\epsilon_0)}{\gamma} = t_{{\rm mix}}(\cD_P^\sharp,\epsilon_0)\ .
$$
This has the consequence that if $\cD_P^\sharp$ is ergodic and gapped, then so are $\cL_\gamma$ and $\cL_{P,\gamma}$ for all sufficiently large $\gamma$. See Theorem~\ref{TMIXZL3}, which complements spectral results of \cite{PP21}, for a precise statement. 

 This  shows one way in which $\cD_P^\sharp$ governs the long-time behavior of the evolution described by  both
$\cL_{P,\gamma}$~, and  by $\cL_\gamma$~, near the Zeno limit. For another, suppose  that $\cD_P^\sharp$ is ergodic and gapped, and let $\bar R$ be the unique density matrix on $\cH_B$ such that $\cD_P^\sharp \bar R =0$. Let $\bar\rho_\gamma$ be the unique (for $\gamma$ sufficiently large)
density matrix on $\cH_{AB}$ satisfying $\cL_\gamma \bar\rho_\gamma =0$. Then we prove that 
\begin{equation}\label{INTR8}
\lim_{\gamma\to \infty}\bar\rho_\gamma = \pi_A\otimes \bar R\ .
\end{equation}
Note the difference with \eqref{BTM9}: here we first let $t\to\infty$
to obtain $\bar\rho_\gamma$ and then take $\gamma\to \infty$, rather than
the other way around. By \eqref{BTMB}, if we first take 
$\gamma\to \infty$, after an initial dissipative projection onto $\cM$, 
the motion is coherent, and 
$\lim_{\gamma\to \infty}e^{t\cL_\gamma \rho_0}
= \pi_A\otimes e^{t\cK_P}\tr_A\rho_0$ has no limit as $t\to \infty$.

In Section~\ref{STSTEXP} we turn to the computation of non-equilibrium  steady states (NESS) for
\eqref{LindEq11}. Even in the general setting of open quantum
systems, boundary driven or not, there are few cases in which NESS can be computed for all but the smallest of systems. Extremely valuable examples in which this has be done are given in the work of Prosen \cite{P11A,P11B} and Popkov and Prosen \cite{PP15}. However, in most cases, to understand the natures of currents in NESS, one must use perturbation theory. 

There is a large literature on the perturbative approach; see \cite[Section III.D]{LPS02}. The simplest cases concern a Lindblad equation of the type
\begin{equation}\label{perturb}
\frac{{\rm d}}{{\rm d}t}\rho(t) = (\cL_0 + \epsilon \cL_1)\rho(t)
\end{equation}
in which $\cL_0$ is ergodic so that there is a unique density matrix
$\bar\rho_0$ satisfying $\cL_0\bar \rho_0 = 0$. One then {\em assumes} 
that for $\epsilon$ sufficiently small, there is a unique density matrix
$\bar\rho_\epsilon$ satisfying $(\cL_0 + \epsilon
\cL_1)\bar\rho_\epsilon =0$ and moreover (the key assumption), that  $\bar\rho_\epsilon$ has an 
expansion 
\begin{equation}\label{INTR12A}
\bar\rho_\gamma = \bar\rho_0 + 
\epsilon\sum_{k=0}^\infty \epsilon^{k} \bar n_k\ .
\end{equation}
Inserting this expansion into the equation \eqref{perturb}, and equating terms with like powers of $\epsilon$ leads to the system of equations
\begin{equation}\label{INTR12B}
\cL_0 \bar n_0 = -\cL_1\rho_0 \quad{\rm and, \ for}\  k\geq 1,\quad 
\cL_0 \bar n_k = \cL_1 \bar n_{k-1}\ .
\end{equation}
When $\cL_0$ is ergodic, its range consists of all traceless operators, and hence if $\cL_1$ is also a Lindblad generator, so that everything in its range is traceless, each of the equations in \eqref{INTR12B}
is solvable, and there will be a unique traceless solution. If one can
then prove that the series converges in the trace norm, and can prove that $\cL_0+\epsilon\cL_1$ is ergodic for all sufficiently small $\epsilon$, one will have justified the expansion \eqref{INTR12A}, from which $\bar\rho_\epsilon$ can now be confidently computed. 

Our problem is more complex. Here, $\cD$ plays the role of $\cL_0$ and it is far from ergodic. Therefore, there is no obvious candidate for $\bar\rho_0$. A related problem arises in \cite{LP15} for a boundary driven spin chain model (but not near the Zeno limit). In that case, a conjectured ansatz \cite[eq. (6)]{LP15}  based on exact calculation for short chains provides starting point $\rho_0$, but there is no proof that is the correct starting point of the expansion. 

In our setting, the Lindblad generator $\cD_P^\sharp$ provides an equation for $\bar\rho_0$, at least when $\cD_P^\sharp$ is ergodic and gapped.  Then by \eqref{INTR8}, $\lim_{\gamma\to\infty}\bar\rho_\gamma = \pi_A\otimes \bar R$, where $\cD_P^\sharp \bar R =0$,
and thus we take $\pi_A\otimes \bar R$ as the starting point of an  expansion 

\begin{equation*}
\bar\rho_\gamma = \pi_A \otimes \bar R + 
\gamma^{-1}\sum_{k=0}^\infty \gamma^{-k} \bar n_k \ .
\end{equation*} 
Now the equations to be solved are
\begin{equation}\label{INTR12BX}
\cD \bar n_0 = \cK(\pi_A\otimes \bar R)
\end{equation}
and then  for each $k\geq 1$, 
\begin{equation}\label{INTR12C}
\cD \bar n_k = \cK \bar n_{k-1}\ .
\end{equation}
On account of the large null space of $\cD$, and correspondingly small range, the solvability conditions are not so simple as they were for
\eqref{INTR12B}, and even when these equations are individually solvable, they will have infinitely many solutions. Nonetheless, under the conditions described above, we show the existence of a unique solution of the hierarchy \eqref{INTR12BX}
and \eqref{INTR12C}. Moreover, for finite constants $C_0$ and $C_1$, $\|\bar n_k\|_1 \leq C_0 C_1^k$  for all $k$ so that 
the series converges in trace norm for all  $\gamma > C_1$, and one can easily estimate the remainder if the sum is truncated.
These results are the content of Theorem~\ref{GOODSOLS} which presents them in a more explicit form.

The final section of the paper discusses an analogy with the theory of hydrodynamic limits that provides an interesting perspective on the relations between the evolution equations specified by $\cL_\gamma$, $\cK_P$, $\cL_{P,\gamma}$ and $\cD_P^\sharp$.

\section{Notation and some mathematical preliminaries}\label{PRELIM}

Let $\cH$ be any finite dimensional Hilbert space. Then $\widehat{\cH}$ will always denote  the Hilbert space consisting of operators on $\cH$ equipped with the Hilbert-Schmidt inner 
 product $\langle X,Y\rangle_{\widehat{\cH}}^{\phantom{x}} := \tr[X^*Y]$. For $X\in \widehat{\cH}$, $\|X\|_\infty$ denotes the {\em operator norm} of $X$ which is the largest singular value of $X$, 
 $\|X\|_1$ denotes the {trace norm} of $X$, which is the sum of the singular values of $X$, repeated according to their multiplicity.  The Hilbert-Schmidt norm of $X$ is denoted by $\|X\|_2$.

 Let $\cH$ and $\cK$ be two finite dimensional Hilbert spaces. Let $\cT$ be a linear operator from $\widehat{\cH}$ to $\widehat{\cK}$; that is, a super-operator from $\cH$ to $\cK$.  
 The adjoint operator from $\widehat{\cK}$ to $\widehat{\cH}$ (a super operator from $\cK$ to $\cH$) will always be denoted by $\cT^\dagger$.    

 Let $Y$ be an operator on $\cK$ and $X$ be an operator on $\cH$. 
 Then we may regard $Y$ and $X$ as vectors in the Hilbert spaces $\widehat{\cK}$ and 
 $\widehat{\cH}$ respectively, and define the 
 super-operator $|Y\rangle\langle X|$ from $\widehat{\cH}$ 
 to $\widehat{\cK}$ by
 \begin{equation}\label{VCTZ}
|Y\rangle\langle X|Z = \langle X,Z\rangle_{\widehat{\cH}}^{\phantom{x}} Y
 \end{equation}
 for all $Z$ in $\widehat{\cH}$.  In our analysis, this plays the 
 role of the ``vectorization'' discussed in \cite[Section III.B]{LPS02},
 but avoids identifying $\widehat{\cH}$ with $\cH\otimes \cH$ and 
 $\widehat{\cK}$ with $\cK\otimes \cK$ which requires a choice of basis.

 The symbol $\one$ will be used both for the identity operator on $\cH$ and the identity super-operator 
 on $\widehat{\cH}$; which is intended will always be clear from the context. 
 
 Define
 \begin{equation}\label{SOPOPNO1}
 \|\cT\|_{1\to 1} = \sup\{ \|\cT(X)\|_1\ : \|X\|_1 = 1\ \}\ .
 \end{equation}
 This is a norm on the space of super-operators from $\cH$ to $\cK$; in particular, it satisfies the triangle inequality.  We refer to it as the {\em super-operator trace norm}
For example, with $\cH = \cH_{AB}$, let $\cK$ be the super-operator defined by \eqref{unitary}. Then for all $X\in \widehat{\cH}_{AB}$,
$$\|\cK X\|_1 = \| HX -XH\|_1 \leq \|HX\|_1 + \|XH\|_1  \leq 2\|H\|_\infty\|X\|_1$$
where the last inequality is H\"older's inequality for traces. (See e.g. \cite[Section 6.3] {C25}.) Therefore,
\begin{equation}\label{KBND}
\|\cK\|_{1\to 1} \leq 2\|H\|_\infty\ .
\end{equation}

  It is known that if $\cT$ is completely positive and trace preserving (CPTP), then 
 $\|\cT\|_{1\to 1} =1$; see e.g. \cite[Theorem 6.26]{C25}. In other words, CPTP maps are contractive (non-increasing) in the trace norm, and this figures in the next lemma which will be used repeatedly in what follows. 
 
 \begin{lem}\label{COMPLEM} Let $\cH$ be a Hilbert space and let $\cL$ be a Lindblad generator acting on the trace class operator on $\cH$, so  that for all $t\geq 0$, $e^{t\cL}$ is CPTP. Let  $\rho(t)$ solve
 \begin{equation}\label{COMPLEM1}
 \frac{{\rm d}}{{\rm d}t}\rho(t) = \cL \rho(t) + F(t)
 \end{equation}
 with $\rho(0) = \rho_0$, and 
 where for all $t \leq T$, $\int_0^t \|F(s)\|_1{\rm d}s  \leq \epsilon$. Then
$\|\rho(t) - e^{t\cL}\rho_0\|_1 \leq \epsilon$ for all $t\leq T$. 
 \end{lem}
 
 \begin{proof} Writing \eqref{COMPLEM1} as a Volterra integral equation and using a Picard iteration
 shows  that this initial value problem has a unique solution given by 
 ${\displaystyle
 \rho(t) = e^{t\cL}\rho_0 + \int_0^t e^{(t-s)\cL} F(s) {\rm d} s}$.
 By the triangle inequality, the fact that $\|e^{(t-s)\cL}\|_{1\to 1} \leq 1$,  and \eqref{SOPOPNO1},
 ${\displaystyle
 \|\rho(t) - e^{t\cL}\rho_0\|_1 \leq  \int_0^t \|F(s)\|_1  {\rm d} s}$.
 \end{proof}
 
 Now let $\cD_A$ be an ergodic gapped Lindblad generator acting on $\widehat{\cH}_A$, and let $\pi_A$ be the unique steady state. Because of the gap, as in the definition \eqref{QCP} of $\cP$, 
\begin{equation}\label{PADEF}
\cP_A :=  \lim_{t\to\infty}e^{t\cD_A} 
\end{equation}
exists in the super-operator trace norm -- or any other norm in this finite dimensional setting. As a limit of CPTP maps, it is CPTP.  In fact, it is easy to see that for all $X\in \widehat{\cH}_A$, 
$\cP_A X = \tr[X]\pi_A$, and that $\cP_A$ is a projection onto the subspace of  $\widehat{\cH}_A$ that is spanned by $\pi_A$.  One readily computes that for all $Y\in \widehat{\cH}_A$,
$\cP_A^\dagger Y = \tr[\pi_A Y]\one$, so that while $\cP$ is a projection, it is not an orthogonal projection on $\widehat{\cH}_A$~.
In the notation \eqref{VCTZ}, $\cP_A = |\pi_A\rangle\langle \one|$ and 
$\cP_A^\dagger = |\pi_A\rangle\langle \one|$. 

Define the complementary projection $\cQ_A$ by
\begin{equation}\label{QADEF}
\cQ_A := \one - \cP_A\ .
\end{equation}
Because of the spectral gap,
\begin{equation}\label{SADEF}
\cS_A := -\int_0^\infty e^{t\cD_A} \cQ_A {\rm d}t
\end{equation}
exists and satisfies
\begin{equation}\label{SAGENINV}
\cS_A \cD_A = \cD_A\cS_A = \cQ_A\ ,
\end{equation}
so that $\cS_A$ is a generalized inverse of $\cD_A$, though it is not the Moore-Penrose \cite{M20,P55} generalized inverse since $\cQ_A$ is not an orthogonal projection. 

Because $\cD = \cD_A\otimes \one$, $e^{t\cD} = e^{t\cD_A}\otimes \one$. Therefore, the projection $\cP$ defined by \eqref{QCP} satisfies

\begin{equation*}
\cP  = \cP_A\otimes \one\ .
\end{equation*}
Define
\begin{equation}\label{QSDEF}
\cQ = \cQ_A\otimes \one \quad{\rm and}\quad \cS := \cS_A\otimes\one\ .
\end{equation}
Then $\cQ$ and $\cP$ are complimentary projections, and 
 it follows from \eqref{SAGENINV} that
\begin{equation}\label{SGENINV}
\cS \cD = \cD\cS = \cQ\ ,
\end{equation}
so that $\cS$ is a generalized inverse of $\cD$.  Since $\cP_A X = \tr[X]\pi_A$, if follows that for all $Z\in \widehat{H}_{AB}$, $\cP X = \pi_A\otimes \tr_A Z$ where $\tr_A$ denotes the super-operator from 
$\widehat{H}_{AB}$ to $\widehat{\cH}_B$ sending $Z$ to its partial trace over $\cH_A$, $\tr_A[Z]$. As is well-known, $\tr_A$ is  CPTP. 
Define the super-operator $\cV$ from $\widehat{\cH}_B$ to $\widehat{H}_{AB}$ by
\begin{equation}\label{CVVDEF}
\cV(X) := \pi_A\otimes X\ .
\end{equation}
This is also CPTP; if $X$ is a density matrix, it corresponds to adding ancilla in the steady state of $\cD_A$. Therefore, we have the factorization $\cP = \cV \tr_A$.

 \section{Bounds on the projected evolution of $\rho(t)$}\label{PROJMO}
   
Recall that $\cM$ denotes the steady state manifold \eqref{SSM}. As noted in the introduction,  it has been shown in \cite{ZC} under very general conditions that
solutions of \eqref{LindEq11} with general  initial data $\rho_0$ approach $\cM$ in a time of order $1/\gamma$, and then stay close to $\cM$, within a distance of order $1/\gamma$, uniformly in time.     
In this section we derive results of this type under the   assumption that $\cD_A$ is a Lindblad generator  that is ergodic and gapped with stationary state $\pi_A$. 
Using this additional information on the generator, we derive more precise bounds in a simple manner.

 \begin{thm}\label{TZCVS} Let $\cD_A$ be ergodic with spectral gap $a>0$.  There are positive finite and calculable constants $C_0$ and  and $C_1$  depending only on $\cD_A$ and $\|H\|_\infty$ (but not on the dimension of 
 $\cH_B$) such that for all $t>0$, 
\begin{equation}\label{SOMS}
\|\cQ e^{t\cL_\gamma}\|_{1\to 1} \leq  C_0 e^{-\gamma a t/2} + \frac{C_1}{\gamma}\ .
\end{equation}
and 
\begin{equation}\label{SOM2S}
\|\cQ e^{t\cL_\gamma}\cP\|_{1\to 1} \leq   \frac{C_1}{\gamma}\ .
\end{equation}
\end{thm}

 \begin{rem}\label{INITLAY} Let $\rho_0$ be any density matrix on $\cH_{AB}$, and define $\rho(t) := e^{t\cL_\gamma}\rho_0$.  Then as a consequence of \eqref{SOMS} and  $\cQ = \one - \cP$, 
 ${\displaystyle
 \|\rho(t) - \cP \rho(t)\|_1 \leq C_0 e^{-\gamma a t/2} + \frac{C_1}{\gamma}}$.
 For $t > \frac{2\log(\gamma)}{a \gamma}$,  this simplifies to 
  \begin{equation*}
 \|\rho(t) - \cP \rho(t)\|_1 = \|\cQ \rho(t)\|_1 \leq  \frac{C_0 + C_1}{\gamma }\ .
\end{equation*}
This is a precise trace norm version of  the result \eqref{INTR1} of \cite{ZC}.  The proof below, which makes crucial use of the fact that the CPTP super-operator $e^{s\cL_\gamma}$ is a contraction in the trace norm,  is much simpler than the analysis in \cite{ZC} which however is not restricted to CPTP evolution equations. 
 \end{rem}

\begin{proof}        By Duhammel's formula, 
 \begin{equation}\label{DuHam1}
 e^{t\cL_\gamma}  = e^{t\gamma \cD} + \int_0^t e^{(t-s)\gamma \cD}\cK e^{s\cL_\gamma}{\rm d}s \ .
 \end{equation}
 By the triangle inequality and the fact that $[\cQ,\cD] =0$ (see \eqref{SAGENINV} and \eqref{QSDEF}), 
 \begin{equation}\label{DuHam1A}
 \| \cQ e^{t\cL_\gamma} \|_{1\to 1} \leq \|  e^{t\gamma \cD}\cQ \|_{1\to1} +  \int_0^t \| e^{(t-s)\gamma \cD}\cQ\cK e^{s\cL_\gamma}\|_{1\to1}{\rm d}s\ .
 \end{equation}
 By the definition of the super-operator trace norm,
 $$
 \|e^{(t-s)\gamma \cD}\cQ \cK e^{s\cL_\gamma}\|_{1\to1} \leq  \|\cQ e^{(t-s)\gamma \cD}\|_{1\to1} \|\cK \|_{1\to1} \|e^{s\cL_\gamma}\|_{1\to1} 
$$

 Since $e^{s\cL_\gamma}$ is CPTP, $\|e^{s\cL_\gamma}\|_{1 \to 1} =1$, and by \eqref{KBND}, $ \|\cK \|_{1\to1} \leq 2\|H\|_\infty$.    By Theorem~\ref{GENINVTNBLM},  for a finite constant $C$ depending only on $\cD_A$, 
 $\|\cQ e^{t\gamma \cD}\|_{1\to 1} \leq  Ce^{- t\gamma a/2}$.
     Using these estimates in \eqref{DuHam1A},
 $$
\| \cQ e^{t\cL_\gamma} \|_{1\to 1} \leq Ce^{- t\gamma a/2}+ 2\|H\|_\infty C \int_0^\infty e^{-\gamma (t-s)a/2} {\rm d}s
 $$
 which yields \eqref{SOMS} with $C_0 =C$ and $C_1 = \frac{4C\|H\|_\infty}{a}$. 
 The proof of \eqref{SOM2S} is even simpler: Since $\cP$ and $\cQ$ commute with $e^{(t-s)\cD}$, $\cQ e^{(t-s)\cD}\cP = 0$, and the contribution of the term involving $\cK$ is no larger than before. 
 \end{proof}

Crucially for what follows, the projected evolution is globally Lipschitz, and after a relaxation time of order $\gamma^{-1}$, so is the full evolution.

\begin{lem}\label{LipLEM} For any density matrix $\rho_0$ on $\cH$, define $\rho(t) := e^{t\cL_\gamma}\rho_0$ and $R(t) := \tr_A \rho(t)$. Then for all  $t>s>0$,
\begin{equation}\label{LIPBNDP}
\| R(t) - R(s)\| \leq \|\cK\|_{1\to1}(t-s)\ .
\end{equation}

Moreover, for all $t> s > \frac{2\log(\gamma)}{a\gamma}$,
\begin{equation}\label{LIPBND}
\|\rho(t) - \rho(s)\|_1 \leq  L(t-s)
\end{equation}
where $L =(\|\cK\|_{1\to1} + C_0+C_1)$, with $C_0$ and $C_1$ as in Theorem~\ref{TZCVS}. If $\rho_0 \in \cM$, then \eqref{LIPBND} is valid for all $t > s \geq 0$. 
\end{lem}

\begin{proof} Since $\cD = \cD_A\otimes \one$, and since $\tr_A \cD_A =0$, 
$$
\frac{{\rm d}}{{\rm d}t}R(t) = \tr_A(\cK + \gamma \cD)\rho(t) = \tr_A \cK \rho(t)\ .
$$
Then since $\|\tr_A \cK \rho(t)\|_1 \leq \|\tr_A\|_{1\to 1}\|\cK\|_{1\to1}\|\rho(t)\|_1 = \|\cK\|_{1\to1}$.  This proves \eqref{LIPBNDP}.

Next, using the fact that $\cD =\cD \cQ$, we compute
${\displaystyle \frac{{\rm d}}{{\rm d}t}e^{t\cL_\gamma}\rho = \cL_\gamma(e^{t\cL_\gamma}\rho) = \cK(e^{t\cL_\gamma}\rho) + \gamma \cD \cQ e^{t\cL_\gamma}\rho}$.
Therefore, by Remark~\ref{INITLAY}, for $t> t_\gamma =   \frac{2\log(\gamma)}{a\gamma}$
$$
\left\Vert \frac{{\rm d}}{{\rm d}t}\rho(t)\right\Vert_1 \leq \|\cK\|_{1\to1} + \gamma \frac{C}{\gamma a}  = \|\cK\|_{1\to1} + \frac{C}{a} \ .
$$
Now integration yields the bound. When $\rho_0\in \cM$, the condition $t>  \frac{2\log(\gamma)}{a\gamma}$ is superfluous. 
\end{proof}

The next theorem says that near the Zeno limit, we may replace general initial data $\rho_0$ by its projection onto $\cM$, $\cP\rho_0$, and after a time of order $\gamma^{-1}$, the trace norm 
difference between the two solution is (essentially) of order $\gamma^{-1}$, uniformly in time.  

\begin{thm}\label{EULLIM} Let $\cD_A$ be ergodic with spectral gap $a>0$. There is a constant $C>0$ independent of $\gamma$ such that 
for all $t \geq t_\gamma := \frac{2\log(\gamma)}{a\gamma}$, 
\begin{equation}\label{EULLIM1} 
\| e^{t \cL_\gamma}\rho - e^{t \cL_{\gamma}} \cP \rho\|_1 \leq \frac{\log(1+\gamma)}{\gamma}C\ .
\end{equation}
\end{thm} 

\begin{proof}
Fix $t>0$, and let $\gamma$ be sufficiently large that $t > t_\gamma$. Then 
$$e^{t\cL_\gamma }\rho = e^{(t- t_\gamma)\cL_\gamma }(\cP + \cQ)e^{t_\gamma \cL_\gamma }\rho =  e^{(t- t_\gamma)\cL_\gamma }\cP\rho +  e^{(t- t_\gamma)\cL_\gamma }\cQ e^{t_\gamma\cL_\gamma }\rho\ .$$
By the definition of $t_\gamma$, Theorem~\ref{TZCVS}, and the fact that $e^{(t- t_\gamma)\cL_\gamma }$ is a contraction in the trace norm, $\|e^{(t- t_\gamma)\cL_\gamma }\cQ e^{t_\gamma\cL_\gamma }\rho\|_1 \leq   \frac{C_0 + C_1}{\gamma }$. By Lemma~\ref{LipLEM},
${\displaystyle \|e^{(t- t_\gamma)\cL_\gamma }\cP \rho - e^{t\cL_\gamma }\cP \rho\|_1  \leq L t_\gamma}$,
and now \eqref{EULLIM1}  follows by the triangle inequality. 
\end{proof}

\section{The super-operators $\cK_P$ and $\cD_P$}\label{SUPKPDP}

Recall  the that the super-operator $\cV$ acting on $\widehat{\cH}_B$  has been defined in  \eqref{CVVDEF} by $\cV X = \pi_A\otimes X$. The super-operator $\cK_P$ defined in \eqref{KPDPDEF} can be expressed 
in terms of $\cV$ which is how it shall arise in expansions considered here. The following lemma 
is due to \cite{ZC}; its simple proof is included for convenience.  

\begin{lem}\label{CommLem} 
For all density matrices $R$ on $\cH_B$,  
 \begin{equation}\label{cKredform}
\tr_A \cK\cV  R =   -i   [H_P,R] =: \cK_P R
 \end{equation}
 where
 $H_P :=  \tr_A[(\pi_A\otimes \one_B)H]$ as in \eqref{INTR3}.
 \end{lem}
 
 \begin{proof} By the definition \eqref{CVVDEF} of $\cV$, 
 \begin{eqnarray*}\cK_P R = \cP\cK \pi_A\otimes R &=& -i\cP  \left(H(\pi_A \otimes R) -  (\pi_A\otimes R)H\right)\\
 &=& -i\cP  \left(H(\pi_A \otimes \one_B)(\one_A\otimes R) -  (\one_A\otimes R)(\pi_A\otimes \one_B)H\right)\\
 &=& -i\left(\pi_A\otimes \tr_A[ H(\pi_A \otimes \one_B)(\one_A\otimes R)] - \pi_A\otimes \tr_A[ (\one_A\otimes R)(\pi_A\otimes \one_B)H]\right)\\
 &=& -i\left(\pi_A\otimes \tr_A[ H(\pi_A \otimes \one_B)]R - \pi_A\otimes R \tr_A[ (\pi_A\otimes \one_B)H]\right)\ .
 \end{eqnarray*}
 By the partial cyclicity of the partial trace, 
 $ \tr_A[H(\pi_A\otimes \one_B)] =  \tr_A[(\pi_A\otimes \one_B)H] = H_P$,  which shows that $H_P$ is self-adjoint, and proves \eqref{cKredform}. 
 \end{proof}

For future use, 
note that since $\tr_A$ and $\cV$ are CPTP hence hence trace norm contractions, 
\begin{equation*}
\|\cK_P\|_{1\to1} \leq \|\cK \|_{1\to 1} \leq 2\|H\|_\infty
\end{equation*}
where the final inequality is from \eqref{KBND}. 

As an immediate consequence of Lemma~\ref{CommLem}, $\cK_P$ generates a coherent CPTP evolution. We now turn to $\cD_P$ as defined by
\begin{equation}\label{DPDEFX}
\cD_P := -\tr_A \cK \cS \cK \cV\ .
\end{equation}  
It was shown in \cite{PEPS} that under certain conditions, 
$\cD_P$ is a Lindblad generator so that $\cD_P$ generates a CPTP evolution,
and therefore so does
\begin{equation}\label{LPGDEFX}
\cL_{P,\gamma} := \cK_P +\frac{1}{\gamma}\cD_P\ .
\end{equation}.  

  We show that $\cD_P$ is a Lindblad generator  assuming only that $\cD_A$ is ergodic and gapped. This was shown in \cite{PEPS} under the additional assumptions that $\cD_A$ is diagonalizable, and more important, that a certain matrix determined by $\cD_A$  is positive semidefinite, but conditions for this matrix to be positive semidefinite are not discussed in \cite{PEPS}.

The next lemma presents the argument of \cite{PEPS} adapted to remove the assumption that $\cD_A$ is diagonalizable. Apart from the small changes needed to avoid assuming that $\cD_A$ is diagonalizable, we closely follow the argument in \cite{PEPS}.  For background on Jordan bases and dual bases, see Appendix~\ref{DBJBA}. 

\begin{lem}\label{LINDINGEN} Let $\cD_A$ be  ergodic and gapped  with the unique steady state $\pi_A$. Let $\{Y_0,\dots,Y_{n_A}\}$ be a Jordan basis for $\cD_A$ such that $Y_0 = \pi_A$, and let 
 $\{X_0,\dots,X_{n_A}\}$ be the dual basis. Define the operators $\{G_0,\dots,G_{n_A}\}$ on $\cH_B$ by 
 \begin{equation}\label{HEXPTHM}
 H = \sum_{j=0}^{n_{A}} X_j \otimes G_j \ .
 \end{equation}
 Define an $n_A\times n_A$ matrix $M$ by
 \begin{equation}\label{MMATDEF}
 M_{j,k} :=   -\langle \cS_A ^\dagger X_k, X_j\pi_A\rangle
 \end{equation}
 and define $A = \frac12(M + M^\dagger)$ and $B := \frac{1}{2i}(M - M^\dagger)$. 
Suppose that $A$ is positive semidefinite so that it has an spectral decomposition
\begin{equation}\label{JOADEF}
A_{j,k} = \sum_{\ell =1}^r \mu_\ell v^{(\ell)}_j  \overline{v^{(\ell)}_k}
\end{equation}
where $r$ is the rank of $A$, and each $\mu_\ell$ is strictly positive. 

Then $\cD_P$, defined in \eqref{DPDEFX}, is a Lindblad generator, and 
 for all operators $R$ on $\cH_B$, 
 \begin{equation}\label{PEPSLG2}
\cD_P R  = \sum_{\ell=1}^{r }\left(2V_\ell R V_\ell^\dagger - V_\ell^\dagger V_\ell  R - R V_\ell^\dagger V_\ell  \right)  -i[H_L,R]\ ,
\end{equation}
where 
 \begin{equation}\label{PEPSLG3}
  V_\ell :=  \sqrt{\mu_\ell}\sum_{k=1}^{n_A}\overline{v^{(\ell)}_k} G_k\quad{\rm and}\quad  H_L := \sum_{j,k=1}^{n_A} B_{j,k} G^\dagger_k G_j\ .
  \end{equation}
 \end{lem}
 
 \begin{proof} 
Since $H$ is self-adjoint, \eqref{HEXPTHM} can be extended to
${\displaystyle
 H = \sum_{j=0}^{n_{A}} X_j \otimes G_j = \sum_{j=0}^{n_{A}}  X^\dagger_j \otimes G^\dagger_j}$.  
Using both of these expression for $H$, 
$$ \cK \cV R = -i\left( H (\pi_A\otimes R)  - (\pi_A\otimes R)  H\right)
=  -i\sum_{j=0}^{n_A} \left( X_j \pi_A \otimes G_j R  - \pi_A   X^\dagger_j \otimes R G^\dagger_j \right)\ ,$$
 and
   since $\cS(Y_0\otimes \left(G_0 R - R G_0^\dagger\right) =0$,
 ${\displaystyle
 \cS \cK (\pi_A\otimes R)  = -i\sum_{j=1}^{n_A} \cS \left( X_j \pi_A \otimes G_j R  - \pi_A   X^\dagger_j \otimes R G^\dagger_j \right)}$.
 Using the expression for $\cS$ provided in Theorem~\ref{ALTS}, namely
 ${\displaystyle 
 \cS  = \sum_{k=1}^{n_A} |Y_k\rangle  \langle \cS^\dagger X_k|}$.
\begin{eqnarray*}
\cS  (X_j \pi_A \otimes G_j R) = \sum_{k=1}^{n_A}  \langle \cS ^\dagger X_k, X_j\pi_A\rangle Y_k \otimes G_jR
= \sum_{k=1}^{n_A} M_{k,j}  Y_k \otimes G_jR\ .
\end{eqnarray*}
Using the same expression for $\cS$, but now for the dual bases $\{Y_0^\dagger,\dots,Y_{n_A}^\dagger\}$ and  $\{X_0^\dagger,\dots,X_{n_A}^\dagger\}$,
\begin{eqnarray*}
\cS  ( \pi_A X_j^\dagger  \otimes R G_j^\dagger) = \sum_{k=1}^{n_A}  \langle \cS^\dagger (X_k^\dagger), (X_j\pi_A)^\dagger\rangle  Y_k^\dagger \otimes RG_j^\dagger\ .
\end{eqnarray*}
By cyclicity of the trace,
$
\langle e^{t\cD_A^\dagger}X_k^\dagger, (X_j\pi_A)^\dagger\rangle =   \tr[(X_j \pi_A )^\dagger e^{t\cD_A^\dagger}(X_k)] = \overline{\langle e^{t\cD_A^\dagger}X_k, X_j\pi_A\rangle}
$. (For $z\in \C$, $\bar z$ denotes its complex conjugate.) Integrating in $t$, 
${\displaystyle 
 \langle \cS^\dagger (X_k^\dagger), (X_j\pi_A)^\dagger\rangle  = \overline{ \langle \cS ^\dagger X_k, X_j\pi_A\rangle }}$.
Altogether,
$$
\cS \cK (\pi_A\otimes R)  = i\sum_{j,k=1}^{n_A} \left(M_{k,j}  Y_k \otimes G_jR - \overline{M_{k,j}} Y_k^\dagger \otimes RG_j^\dagger\right)\ .
$$
Therefore,
 \begin{eqnarray*}
 \cK \cS \cK (\pi_A\otimes R)  &=& \sum_{\ell=0}^{n_A}\sum_{j,k=1}^{n_A}   M_{k,j} \left(  X^\dagger _\ell Y_k  \otimes G^\dagger_\ell G_j R   -   Y_k X^\dagger_\ell \otimes  G_j R G^\dagger_\ell \right)\nonumber\\
 &-& \sum_{\ell=0}^{n_A}\sum_{j,k=1}^{n_A}   \overline{M_{k,j}}  \left( X_\ell  Y^\dagger_k\otimes   G_\ell  R G^\dagger_j   -  Y^\dagger_kX_\ell \otimes R G^\dagger_j G_\ell \right)\ .
 \end{eqnarray*}
Taking the partial trace over $\cH_A$, which produces a factor of $\delta_{k,\ell}$ in each term,
 \begin{equation}\label{KSKCALC}
 -\tr_A[ \cK \cS \cK (\pi_A\otimes R)] = \sum_{j,k=1}^{n_A} M_{k,j}  \left( G_j R G_k^\dagger -  G^\dagger_k G_j R  \right) 
 +\sum_{j,k=1}^{n_A}  \overline{M_{k,j}}   \left(    G_k R G^\dagger_j   -   R G^\dagger_j  G_k\right)\ .
 \end{equation}
 
 Note that 
 ${\displaystyle 
  \sum_{j,k=1}^{n_A}  \overline{M_{k,j}}   G_k R G^\dagger_j  = \sum_{j,k=1}^{n_A}  \overline{M_{j,k}}     G_k R G^\dagger_j  = \sum_{j,k=1}^{n_A}  M_{k,j}^\dagger   G_k R G^\dagger_j}$.
 In the same way, \hfill\break
 ${\displaystyle  \sum_{j,k=1}^{n_A}  \overline{M_{k,j}} R G^\dagger_j  G_k = \sum_{j,k=1}^{n_A}  M_{k,j}^\dagger R G^\dagger_k  G_j =
 \sum_{j,k=1}^{n_A}  (A -iB)_{k,j}^\dagger R G^\dagger_k  G_j}$. 
 Therefore, \eqref{KSKCALC} becomes 
 \begin{equation}\label{ALMLIND}
 -\tr_A[ \cK \cS \cK (\pi_A\otimes R)] = \sum_{j,k=1}^{n_A} A_{k,j} \left(2 G_j R G_k^\dagger - G^\dagger_k G_j R   -   R G^\dagger_k  G_j\right)
  -i[H_L,R]
\end{equation}
where
${\displaystyle
H_L := \sum_{j,k=1}^{n_A} B_{j,k} G^\dagger_k G_j}$.  Now assuming that $A$ is positive semidefinite, it has a spectral decomposition of 
the form \eqref{JOADEF} where $\mu_\ell > 0$ for $\ell = 1,\dots,r$. 
Substituting this into \eqref{ALMLIND} yields \eqref{PEPSLG2}. 
\end{proof}

It remains to prove that the self-adjoint part of the matrix $M$ defined in \eqref{MMATDEF} is positive semidefinite. We use the Gelfand-Naimark-Segal (GNS) inner product associated to $\pi_A$, defined as follows: 

Let $\cH$ be a finite dimensional Hilbert space, and let $\mu$ ve a density matrix on $\cH$. For operators $X$ and $Y$ on $\cH$, their GNS inner product $\langle X,Y\rangle_\mu$ is given by
\begin{equation}\label{GNSDEF}
\langle X,Y\rangle_\mu := \tr[X^*Y\mu] = \langle X\mu^{1/2},Y\mu^{1/2}\rangle_{\widehat{\cH}}\ .
\end{equation}
If $\cL$ is a Lindblad generator acting on $\widehat{\cH}$ with invariant state $\pi$, it is 
said to satisfy the {\em GNS detailed balance condition} with respect 
to $\pi$ if its  Hilbert-Schmidt adjoint $\cL^\dagger$ is self-adjoint
with respect to the GNS inner product associated to $\pi$. In the
quantum setting, as in the classical, detailed balance corresponds to
microscopic reversibility \cite{Ar73}. In the
quantum setting, there are other way to weight the inner product, for
example $\langle
\mu^{1/4}X\mu^{1/4},\mu^{1/4}Y\mu^{1/4}\rangle_{\widehat{\cH}}$, which
defines the Kubo-Martin-Schwinger (KMS) inner product and another
notion of detailed balance \cite{FU10}. The GNS notion of detailed balance is the
most restrictive \cite{AC21}, and a theorem of Alicki \cite{A78} gives
the form of all Lindblad generators satisfying this condition, which is very helpful \cite{CM17} when studying the rate of approach to stationarity. 

It is easy to see that when $\cD_A^\dagger$ is 
self-adjoint with respect to the GNS inner product induced by $\pi_A$,
which is the case in the example considered in \cite{PEPS}, the matrix
$A$ specified  in Lemma~\ref{LINDINGEN} is 
diagonal with non-negative entries and $H_L$ in \eqref{PEPSLG2} will always be identically zero. 
The following general inequality will be shown to yield the
positivity of $A$  without any extra assumption on $\cD_A$.

\begin{thm}\label{POSINEQINT} Let $\cD_A$ be an  ergodic and gapped Lindblad generator on $\cH_A$ with steady state $\pi_A$. 
Define $\cP_A$, $\cQ_A$ and $\cS_A$ in  terms of $\cD_A$ as in \eqref{PADEF}, \eqref{QADEF} and \eqref{SADEF}. 
Then for all $Z\in \widehat{\cH}_A$, 
\begin{equation*}
\langle \cS_A^\dagger Z,\cQ_A^\dagger Z\rangle_{\pi_A} + \langle \cQ_A^\dagger Z, \cS_A^\dagger Z\rangle_{\pi_A} \leq 0\ 
\end{equation*}
where the inner products are taken in the GNS inner product associated to $\pi_A$, as in \eqref{GNSDEF}. 
\end{thm}

\begin{proof}Since $e^{t\cD_A^\dagger}$  completely positive and $e^{t\cD_A^\dagger}\one_A = \one_A$, an inequality of Lieb and Ruskai \cite{LR74} (a slightly more general form was proved by  Choi \cite[Corollary 2.8]{Choi74}, but we only need the original; see \cite[Section 5.6]{C25}) says that for all operators $W$ on $\cH_A$,
$\left(e^{t\cD_A^\dagger}W\right)^\dagger \left( e^{t\cD_A^\dagger}W\right)   \leq  e^{t\cD_A^\dagger}(W^\dagger W)$.
Differentiating in $t$ at $t=0$ yields
\begin{equation}\label{LRC}
\cD_A^\dagger(W)^\dagger W + W^\dagger \cD_A^\dagger(W) \leq \cD_A(W^\dagger W)\ .
\end{equation}
The inequality \eqref{LRC}, derived in this manner, also appears in Lindblad's work \cite[(3.1)]{Lin76} on the generators of CPTP semigroups. Now apply \eqref{LRC} as follows:

Define $W := \cS_A^\dagger Z$, so that $\cD_A^\dagger W = \cQ_A^\dagger Z$. Then 
\begin{eqnarray*}
\langle \cS_A^\dagger (Z),\cQ_A^\dagger Z\rangle_{\pi_A} + \langle \cQ_A^\dagger Z, \cS_A^\dagger(Z)\rangle_{\pi_A} &=& 
\langle W,\cD_A^\dagger W\rangle_{\pi_A} + \langle \cD_A^\dagger W, W\rangle_{\pi_A}\\
&=& \tr\left[\left(W^\dagger \cD_A^\dagger(W) + \cD_A^\dagger(W^\dagger)W\right)\pi_A\right]\\
&\leq& \tr\left[\left( \cD_A^\dagger(W^\dagger W) \right)\pi_A\right]\\
&=& \langle \pi_A, \cD_A^\dagger(W^\dagger W) \rangle = \langle \cD_A(\pi_A), W^\dagger W\rangle = 0\ ,
\end{eqnarray*}
where the only inequality is \eqref{LRC}. 
\end{proof}

\begin{thm}\label{ALWAYS} Let $\cD_A$ be gapped and ergodic, and let 
$\cD_P = -\tr_A \cK \cS \cK \cV$ as in \eqref{DPDEFX}. Then $\cD_P$ is the Lindblad generator given by  \eqref{PEPSLG2} and \eqref{PEPSLG3}. 
\end{thm}

\begin{proof} By Lemma~\ref{LINDINGEN} it suffices to prove that
for arbitrary
$(z_1,\dots,z_{n_A})\in \C^{n_A}$,
${\displaystyle 
\sum_{j,k=1}^{n_A} \overline{z_k} A_{k,j} z_j \geq 0}$.
By \eqref{MMATDEF} and cyclicity of the trace,
${\displaystyle
A_{k,j} = -\frac12\left(\langle \cS ^\dagger X_k, X_j\pi_A\rangle  + \langle X_j,\cS ^\dagger X_k\pi_A\rangle\right)}$.
Define ${\displaystyle Z := \sum_{j=1}^{n_A}z_j  X_j}$, and note that $Z = \cQ^\dagger Z$.  Then 
\begin{eqnarray}\label{poscrit3}
\sum_{j,k=1}^{n_A} \overline{z_k} A_{k,j} z_j &=&   -\frac12\sum_{j,k=1}^{n_A}  \left(  \left\langle z_j \cS^\dagger X_j,z_k X_k\pi_A\right\rangle^{}_{\text{\tiny HS}} +
\left\langle z_j X_j,z_k \cS^\dagger X_k\pi_A\right\rangle^{}_{\text{\tiny HS}}\right)\nonumber\\
&=&   -\frac12 \left(\left\langle \cS^\dagger(Z), \cQ^\dagger Z\pi_A\right\rangle^{}_{\text{\tiny HS}}  + \left\langle \cQ^\dagger Z,  \cS^\dagger(Z)\pi_A\right\rangle^{}_{\text{\tiny HS}}\right)\ .
 \end{eqnarray}
 By Theorem~\ref{POSINEQINT} the left side of \eqref{poscrit3} is non-negative.  \end{proof}

 \section{Approximate equations of motion}\label{APPROXEQ}
 
 In this section we present a simple and rigorous  derivation of the approximate equations of motion of \cite{ZC} and \cite{PEPS} that provides explicit error  estimates. We first consider initial data in $\cM$.

 \begin{lem}\label{COHERENTSC} Suppose that $\cD_A$ is ergodic with steady state $\pi_A$ and spectral gap $a>0$.  Let $R_0$ be any density matrix on $\cH_B$. 
 Let ${\displaystyle \cL_{P,\gamma} = \cK_P + \frac{1}{\gamma}\cD_P}$ as in \eqref{LPGDEFX}.  
 Then there exists a constant $C$ depending only on $\cD_A$  and the operator norm of $H$  such for all $T>0$, 
  \begin{equation}\label{COHERENTSC1}
 \| \tr_A  e^{t\cL_\gamma}\pi_A\otimes R_0 - e^{t\cL_{P,\gamma}}R_0\|_1 \leq  \frac{1}{\gamma^2} CT  \ 
 \end{equation}
 for all $0 \leq t \leq T$. 
 \end{lem}
 
\begin{proof}
 Let $R_0$ be any density matrix on $\cH_B$. Define $R(t) := \tr_A  e^{t\cL_\gamma}\pi_A\otimes R_0$ so that $\cP (e^{t\cL_\gamma}\pi_A\otimes R_0) = \pi_A\otimes R(t)$.
 Define $Y(t) = e^{t\cL_\gamma}\pi_A\otimes R_0$  so that  $R(t) := \tr_A Y(t)$. 
 By Duhammel's formula \eqref{DuHam1},  and $e^{t\gamma \cD}\pi_A\otimes R_0 = \pi_A\otimes R_0$,
 \begin{equation}\label{DuHamFoPr}
 Y(t) =   Y(0) + \int_0^t e^{(t-s)\gamma \cD}\cK Y(s) {\rm d}s\ .
 \end{equation}
Since $\cD = \cD_A\otimes \one$ and   $\tr_A \cD_A =0$,  taking the partial trace over $\cH_A$ throughout 
\eqref{DuHamFoPr} yields
 \begin{equation}\label{DuHamPr}
 R(t)   = R_0 + \int_0^t \tr_A \cK e^{s\cL_\gamma}\pi_A\otimes R_0 {\rm d}s\ .
 \end{equation}
 
 Since $\cP+ \cQ =\one$ and $\cP = \cV\tr_A$,
 ${\displaystyle 
 e^{s\cL_\gamma}\pi_A\otimes R_0  = \cV R(s) + \cQ e^{s\cL_\gamma}\pi_A\otimes R_0}$.
 Inserting this in \eqref{DuHamPr}, and using Lemma~\ref{CommLem},
 \begin{equation}\label{DuHamPr2}
 R(t)   = R_0 + \int_0^t \cK_P R(s) {\rm d}s + \int_0^t \tr_A \cK \cQ Y(s) {\rm d}s\ .
 \end{equation}

  Applying $\cQ$ to both sides of \eqref{DuHamFoPr}, and making a simple change of the integration variable,
 $$
 \cQ Y(t) =  \frac{1}{\gamma}  \int_0^{\gamma t} e^{r \cD}\cQ \cK Y(t- r/\gamma) {\rm d}r\ .
 $$
 Define $W(t,r) := Y(t- r/\gamma) - Y(t)$. Then
 ${\displaystyle 
 \cQ Y(t) =  \frac{1}{\gamma}  \int_0^{\gamma t} e^{r \cD}\cQ \cK Y(t) {\rm d}r +  \frac{1}{\gamma}  \int_0^{\gamma t} e^{r \cD}\cQ \cK W(t,r) {\rm d}r}$.
 Inserting this in \eqref{DuHamPr2} yields
 \begin{equation*}
 R(t)   = R_0 + \int_0^t \cK_P R(s) {\rm d}s + \int_0^t \tr_A \cK \left( \frac{1}{\gamma}  \int_0^{\gamma s} \left(e^{r \cD}\cQ \cK {\rm d}r \right) Y(s) +  \frac{1}{\gamma}  \int_0^{\gamma s} e^{r \cD}\cQ \cK W(t,r) {\rm d}r \right) {\rm d}s\ .
 \end{equation*}
 In the first double integral, we again use $Y(s) = \cP Y(s) + \cQ Y(s) = \cV R(s) + \cQ Y(s)$ to obtain
 \begin{equation}\label{DuHamPr2V1B}
 R(t)   = R_0 + \int_0^t  \left(\cK_P  +   \frac{1}{\gamma}  \int_0^{\gamma s} \tr_A \cK e^{r \cD}\cQ \cK \cV {\rm d}r \right)R(s){\rm d}s +  \int_0^t F_0(s){\rm d}s \ 
 \end{equation}
 where 
${\displaystyle  F_0(s) :=    \frac{1}{\gamma}    \int_0^{\gamma s} \tr_A \cK\left(  e^{r \cD}\cQ \cK(\cQ Y(s)+ W(t,r)\right)  {\rm d}r}$.

 By Lemma~\ref{LipLEM}, $\|W(t,r)\|_1 \leq  Lr/\gamma$. By Theorem~\ref{GENINVTNBLM} in the appendix, there is a constant $C$ depending only on $\cD_A$ such that $\|\cQ e^{t\cD}\|_{1\to 1}\leq Ce^{-ta/2}$. 
 Therefore,
 $$
 \left\Vert  \frac{1}{\gamma}  \int_0^{\gamma s} \cK  e^{r \cD}\cQ \cK W(t,r) {\rm d}r \right\Vert_1 \leq \|\cK\|_{1\to 1}^2 \frac{CL}{\gamma^2}\int_0^\infty r e^{-ta/2}{\rm d} r\ .
 $$
 and again by Theorem~\ref{TZCVS}, $\|\cK\ e^{r \cD}\cQ \cK(\cQ Y(s)\|_1 \leq 4\|\cK\|^2_{1\to1} \frac{C_1}{\gamma}$. Altogether there is a constant $C$ depending only on $\cD_A$ and $\|\cK\|_{1\to1}$ such that 
 for all $s \geq 0$, 
  \begin{equation}\label{DuHamPr2V3}
 \|F_0(s) \|_1  \leq \frac{C}{\gamma^2}\ .
  \end{equation}

 Define ${\displaystyle F_1(s) :=  \left(\frac{1}{\gamma} \cD_P - \frac{1}{\gamma}  \int_0^{\gamma s} \tr_A \cK e^{r \cD}\cQ \cK \cV {\rm d}r\right)R(s)}$, so that with $F(t) = F_0(t) + F_1(t)$, 
 \eqref{DuHamPr2V1B} becomes
  \begin{equation}\label{DuHamPr2VJ}
 R(t)   = R_0 + \int_0^t  \left(\cK_P  +   \frac{1}{\gamma}  \cD_P \right)R(s){\rm d}s +  \int_0^t F(s){\rm d}s \ .
 \end{equation}

 Since ${\displaystyle   -\cS\cQ\cK \cV  = \int_0^\infty e^{r \cD}\cQ \cK \cV {\rm d}r}$, 
 $$
\left\Vert  \cD_P   -  \int_0^{\gamma s} \tr_A \cK e^{r \cD}\cQ \cK \cV {\rm d}r \right\Vert _{1\to 1}   =\left\Vert \int_{s\gamma}^\infty \tr_A \cK e^{r \cD}\cQ \cK \cV {\rm d}r  \right\Vert_{1\to1}  \leq 
2\|\cK\|^2_{1\to 1} \int_{\gamma s}^\infty \| e^{r \cD}\cQ \|_{1\to 1}{\rm d}r 
 $$
 Again by Theorem~\ref{GENINVTNBLM} proved in the appendix,  for a finite constant $C$ depending only on $\cD_A$ $
    \|e^{r\gamma \cD}\cQ \|_{1\to 1} \leq  Ce^{- r\gamma a/2}$.
 Therefore,
 ${\displaystyle
 \int_0^t \| F_1(s)\|_{1} {\rm d}s \leq  \frac{4\|\cK\|^2_{1\to 1} C}{\gamma^2 a}}$. Combining this with \eqref{DuHamPr2V3} yields 
 ${\displaystyle \int_0^t \|F(s)\|_1{\rm d}s \leq \frac{C}{\gamma^2}}$ for a constant $C$ depending only on $\cD_A$ and $\|\cK\|_{1\to 1}$. 
  Now differentiate \eqref{DuHamPr2VJ} and apply Lemma~\ref{COMPLEM} and \eqref{KBND}.
\end{proof}

\begin{rem} Since differentiating \eqref{DuHamPr2VJ}
 yields ${\displaystyle \frac{{\rm d}}{{\rm d}t}R(t) = \cL_{P,\gamma} R(t) + F(t)}$. 
The estimates obtained above show that  $F(t) = {\mathcal O}(\gamma^{-2})$. Therefore,
$
 {\displaystyle \frac{{\rm d}}{{\rm d}t}R(t)  =  (\cK_P + \gamma^{-1}\cD_P)R(t) + {\mathcal O}(\gamma^{-2})}$, 
which by \eqref{INTR6} is the same as \eqref{INTR7}.     
\end{rem}

While Lemma~\ref{COHERENTSC} applies to initial data in $\cM$, we may extend its domain of application to general initial data using Theorem~\ref{SOMS} and Theorem~\ref{EULLIM}, yielding the following theorem: 

\begin{thm}\label{MTILRM} Suppose that $\cD_A$ is ergodic with steady state $\pi_A$ and spectral gap $a>0$.
Let $\rho_0$ be any density matrix on $\cH_{AB}$, and define $R_0 := \tr_A[\rho_0]$ so that
$\cP \rho_0 = \pi_A\otimes R_0$. Let $a$ denote the spectral gap of $\cD_A$.  For any $\epsilon > 0$ and any $T>0$,
\begin{equation}\label{BTM4}
 \| e^{t\cL_\gamma}\rho_0 - \pi_A\otimes e^{t \cL_{P,\gamma}}R_0\|_1 \leq C \frac{\log(1+\gamma) +1+T}{\gamma}  ,
\end{equation}
uniformly on $[\epsilon, \gamma T]$ for any $0 < \epsilon < T$ for all $\gamma$ such that $ \frac{2\log(\gamma)}{a\gamma} \leq  \epsilon$.
\end{thm}

\begin{proof}
By the triangle inequality, 
\begin{eqnarray}
\| e^{t\cL_\gamma}\rho_0 - \pi_A\otimes e^{t \cL_{P,\gamma}}R_0\|_1 &\leq& \| e^{t\cL_\gamma}\rho_0 - e^{t \cL_{\gamma}}\pi_A\otimes R_0\|_1\label{BTM1}\\
  &+& \|  e^{t \cL_{\gamma}}\pi_A\otimes R_0 - \pi_A\otimes \tr_A  e^{t\cL_\gamma}\pi_A\otimes R_0\|_1 \label{BTM2}\\
   &+& \|  \pi_A\otimes \tr_A  e^{t\cL_\gamma}\pi_A\otimes R_0 - \pi_A\otimes e^{t \cL_{P,\gamma}}R_0\|_1 \label{BTM3}\ .
\end{eqnarray}
Fix $0 < \epsilon < T$, and consider $t$ such that $0 \leq t \leq \gamma T$. 
Using \eqref{EULLIM1} of Theorem~\ref{EULLIM}  the term on the right in \eqref{BTM1} is bounded by $\frac{\log(1+\gamma)}{\gamma}C$ for all $\gamma$ such that $ \frac{2\log(\gamma)}{a\gamma} \leq  \epsilon$.
  The
   term in \eqref{BTM2} equals
$\|  e^{t \cL_{\gamma}}\pi_A\otimes R_0 - \cP(e^{t \cL_{\gamma}}\pi_A\otimes R_0)\|_1\leq \frac{C}{\gamma}$  where the inequality is  \eqref{SOM2S} of Theorem~\ref{TZCVS}. The term in \eqref{BTM3} equals 
$$
\left\|  \pi_A\otimes \left(\tr_A  e^{t\cL_\gamma}\pi_A\otimes R_0 -  e^{t \cL_{P,\gamma}}R_0\right)\right\|_1  = \left\|  \tr_A  e^{t\cL_\gamma}\pi_A\otimes R_0 -  e^{t \cL_{P,\gamma}}R_0\right\|_1 \leq \frac{CT}{\gamma^2}\ , 
$$
where the inequality is \eqref{COHERENTSC1} of Lemma~\ref{COHERENTSC}. This last inequality holds for any $T>0$, and then all $0\leq t \leq T$.  Now replacing $T$ by $\gamma T$, 
and putting all of the estimates together, yields \eqref{BTM4}
\end{proof}

The following simple consequence of Theorem~\ref{MTILRM} relating the evolution generated by $\cL_\gamma$  to the approximate coherent evolution generated by $\cK_P$.

\begin{thm}\label{MTILRMEUL}
Let $\rho_0$ be any density matrix on $\cH_{AB}$, and define $R_0 := \tr_A[\rho_0]$ so that
$\cP \rho_0 = \pi_A\otimes R_0$. Let $a$ denote the spectral gap of $\cD_A$.  For any $\epsilon > 0$ and any $T>0$,
\begin{equation}\label{MTILRMEUL1}
 \| e^{t\cL_\gamma}\rho_0 - \pi_A\otimes e^{t \cK_P}R_0\|_1 \leq C \frac{\log(1+\gamma) +1+T/\gamma}{\gamma}  ,
\end{equation}
uniformly on $[\epsilon, T]$ for any $0 < \epsilon < T$ for all $\gamma$ such that $ \frac{2\log(\gamma)}{a\gamma} \leq  \epsilon$.
In particular,
 \begin{equation*}
 \lim_{\gamma\to\infty} \tr_A e^{t \cL_\gamma}\rho = e^{t\cK_P}\tr_A \rho\ ,
  \end{equation*}
 where the convergence is in the trace norm, and for any $0 < \epsilon < T$, it is uniform on  $\epsilon < t \leq \gamma T$.
\end{thm}

\begin{proof} We first show that 
for all $T>0$ and all $\gamma >0$, there is a finite constant $C$ independent of $T$ and $\gamma$ such that for all $0 < t < T$,
\begin{equation}\label{QEULER1}
\|e^{t\cL_{P,\gamma}}R_0 - e^{t\cK_P}R_0\|_1 \leq \frac{1}{\gamma}T\|\cD_P\|_{1\to 1}\ .
\end{equation}
To see this, define $R(t) := e^{t\cL_{P,\gamma}}R_0$. Then evidently, ${\displaystyle \frac{{\rm d}}{{\rm d}t}R(t) =  \cK_P R(r) + F(t)}$ where $F(t) = \gamma^{-1}\cD_P R(t)$. 
Then $\|F(t)\|_1 \leq  \gamma^{-1}\|\cD_P\|_{1\to 1}$, and then \eqref{QEULER1} follows from  Lemma~\ref{COMPLEM}.  

By the triangle inequality and the fact that $\| \pi_A\otimes (e^{t \cL_{P,\gamma}}R_0 - e^{t \cK_P}\R_0)\|_1 = \| e^{t \cL_{P,\gamma}}R_0 - e^{t \cK_P}\R_0\|_1$.
\begin{equation}\label{QEULER2}
\| e^{t\cL_\gamma}\rho_0 - \pi_A\otimes e^{t \cK_P}R_0\|_1  \leq  \| e^{t\cL_\gamma}\rho_0 - \pi_A\otimes e^{t \cL_{P,\gamma}}R_0\|_1 + \| e^{t \cL_{P,\gamma}}R_0 - e^{t \cK_P}\R_0\|_1
\end{equation}

By Theorem~\ref{MTILRMEUL},
${\displaystyle 
\| e^{t\cL_\gamma}\rho_0 - \pi_A\otimes e^{t \cL_{P,\gamma}}R_0\|_1 \leq C \frac{\log(1+\gamma) +1+T/\gamma}{\gamma} }$
uniformly on $[\epsilon, T]$ for any $0 < \epsilon < T$ for all $\gamma$ such that $ \frac{2\log(\gamma)}{a\gamma} \leq  \epsilon$. Using this to bound the first term on the right in \eqref{QEULER2}, and using 
\eqref{QEULER1} to bound the second, and then harmonizing constants proves \eqref{MTILRMEUL1}. 
\end{proof}

While the proof of Theorem~\ref{MTILRM} contains a rigorous derivation of the equation 
${\displaystyle \frac{{\rm d}}{{\rm d}t}R(t) = \cL_{P,\gamma}R(t)}$, 
there is another simple and natural but  formal approach that not only produces this equation, but provides an infinite sequences of higher order ``corrections''.  The approach is based on the Hilbert expansion from kinetic theory \cite{DH01,DH12}. 

 For a solution $\rho(t)$ of \eqref{LindEq11},  define  $R(t)$ by $\pi_A\otimes R(t) = \cP \rho(t)$ as before, and then
define $n(t)$ by 
\begin{equation}\label{ndef}
 \rho(t) = \pi_A\otimes R(t) + \frac1\gamma n(t) \ .
\end{equation}
Equivalently, $n(t) = \gamma \cQ \rho(t)$, and hence $\cP n(t) =0$ and $\cQ(\pi_A \otimes Z) =0$ for all operators $Z$ on $\cH_A$. 
Inserting \eqref{ndef} into \eqref{LindEq11},
\begin{equation}\label{neq1}
\pi_A\otimes \frac{{\rm d}}{{\rm d}t}R(t) = \cK( \pi_A\otimes R(t)) + \cD n(t) + \frac{1}{\gamma}\cK n(t) - \frac1\gamma \frac{{\rm d}}{{\rm d}t} n(t)\ .
\end{equation}
Since $\cD$ and $\cP$ commute, applying $\cP$ to both sides of \eqref{neq1} yields
\begin{equation}\label{neq2}
\pi_A\otimes \frac{{\rm d}}{{\rm d}t}R(t) = \cP\cK( \pi_A\otimes R(t)) +  \frac{1}{\gamma}\cP \cK n(t)\ .
\end{equation}
Since $\cD$ and $\cQ$ commute applying $\cQ$ to both sides of  \eqref{neq1} yields
\begin{equation}\label{neq3}
\frac{1}{\gamma}\frac{{\rm d}}{{\rm d}t}n(t) = \cQ \cK( \pi_A\otimes R(t)) + \cD n(t) + \frac{1}{\gamma}\cQ \cK n(t)\ .
\end{equation}
Apply $\cS$ to both sides of \eqref{neq3} and use  the fact \eqref{SGENINV} that $\cS\cD = \cQ$ to obtain
\begin{equation}\label{neq4}
n(t) = -\cS \cK( \pi_A\otimes R(t)) -  \frac{1}{\gamma}\cS \cK n(t) + \frac{1}{\gamma}\frac{{\rm d}}{{\rm d}t}\cS n(t) \ .
\end{equation}

Following Hilbert \cite{DH12}, assume that $n(t)$ has an expansion 
 in inverse powers of $\gamma$:
\begin{equation}\label{ngexpan}
n(t) = \sum_{j=0}^\infty \gamma^{-j}n_j(t)\ .
\end{equation}
Inserting  the expansion into  the equation and then applying $\cP$ and $\cQ$ to both sides, yields two coupled equations. 

Inserting this into \eqref{neq4} and equating like powers of $\gamma$ yields
\begin{eqnarray}
n_0(t) &=& -\cS \cK( \pi_A\otimes R(t)) \label{neq5}\\
n_j(t) &=& -  \cS \cK n_{j-1}(t) +  \cS\frac{{\rm d}}{{\rm d}t} n_{j-1}(t) \quad{\rm for}\quad j\geq 1\ .\label{neq6}
\end{eqnarray}

Using \eqref{neq5}, and replacing $n(t)$ by $n_0(t)$ in \eqref{neq2} formally yields
$$
\pi_A\otimes \frac{{\rm d}}{{\rm d}t}R(t) = \cP \cK( \pi_A\otimes R(t))  - \frac{1}{\gamma}\cP\cK \cS \cK( \pi_A\otimes R(t)) + \mathcal{O}\left(\frac{1}{\gamma^2}\right)\ ,
$$
and then taking the trace of this over $\cH_A$  yields
\begin{equation}\label{Hilbert1}
\frac{{\rm d}}{{\rm d}t}R(t) = \tr_A[\cK( \pi_A\otimes R(t))] - \frac{1}{\gamma}\tr_A[\cK \cS \cK( \pi_A\otimes R(t))] + \mathcal{O}\left(\frac{1}{\gamma^2}\right)\ .
\end{equation}

The Hilbert expansion provides a systematic  method for computing corrections to \eqref{Hilbert1} at any order. To go to the next order, all one needs to do is to use $n_0(t) +\frac1\gamma n_1(t)$ in place of $n(t)$ in  \eqref{neq2}.
 Using \eqref{neq6} to determine $n_1(t)$ then yields
\begin{eqnarray}
n_1(t) &=&  -\cS \cK n_{0}(t) + \frac{{\rm d}}{{\rm d}t}n_{0}(t)\nonumber\\
&=& \cS \cK \cS \cK( \pi_A\otimes R(t))  -   \cS^2 \cK( \pi_A\otimes  \frac{{\rm d}}{{\rm d}t} R(t))\nonumber\\
&=& \cS \cK \cS \cK( \pi_A\otimes R(t))  -   \cS^2 \cK\left( \cP\cK( \pi_A\otimes R(t)) +  \frac{1}{\gamma}\cP \cK n(t)\right)\nonumber\\
&=&   (\cS \cK \cS \cK - \cS^2\cK\cP\cK) ( \pi_A\otimes R(t))  -    \frac{1}{\gamma}\cS^2 \cK \cP \cK n(t)\label{neq7}\ .
\end{eqnarray}

Now combine \eqref{neq2}, \eqref{ngexpan} and \eqref{neq7} and then take the trace over $\cH_A$  to obtain
\begin{eqnarray*}
 \frac{{\rm d}}{{\rm d}t}R(t) &=& \tr_A[\cK( \pi_A\otimes R(t))] -  \frac{1}{\gamma}\tr_A[ \cK \cS \cK( \pi_A\otimes R(t)] \nonumber \\
 &+&  \frac{1}{\gamma^2}\tr_A\left[ \cK \left(  \cS \cK \cS \cK - \cS^2\cK\cP\cK) ( \pi_A\otimes R(t))\right)\right] + \mathcal{O}\left(\frac{1}{\gamma^3}\right)\ .
\end{eqnarray*}

It is clear at this point that the corrections of every order can be expressed in terms of $\cK$, $\cS$ and $\cP$.   Define the super-operator $\cB_P$ by
\begin{equation}\label{BURDEF}
\cB_P(R) := \tr_A\left[ \cK \left(  \cS \cK \cS \cK - \cS^2\cK\cP\cK) ( \pi_A\otimes )\right)\right] \ .
\end{equation}
We have then formally derived the approximate equation
$$
\frac{{\rm d}}{{\rm d}t} R = \cK_P R(t) +\gamma^{-1} \cD_P R(t) + \gamma^{-2}\cB_P R(t)\ .
$$
The utility of this equation is in doubt. 
As mentioned in the introduction, an 
exact calculation for the model in Example~\ref{EXAMP1} shows that there are simple cases in which $\cK_P  +\gamma^{-1} \cD_P  + \gamma^{-2}\cB_P$ cannot be put in Lindblad form for any $\gamma>0$. 
 Nonetheless, it still may be that the assumption of an expansion as in \eqref{ngexpan} could be rigorously justified. Later, we shall prove convergence of a stationary form of the Hilbert expansion for steady state solutions of \eqref{LindEq11}.

\section{Davies' oscillation averaging operation}\label{DAOP}

As in \cite{Da74}, for  any super-operator $\cT$ on $\cH_B$, define $\cT^\sharp$ by
 \begin{equation}\label{DAVLIM56}
 \cT^\sharp(X) := \lim_{T\to\infty}\frac{1}{2T}\int_{-T}^T e^{t\cK_P}\cT( e^{-t\cK_P} X)  {\rm d}t\ .
 \end{equation}  
 We are mainly interested in $\cD_P^\sharp$ where $\cD_P$ is a projected Lindbladian generator on $\widehat{\cH}_B$. However, many properties of Davies' coherent averaging operation are valid for all super-operators $\cT$ on $\cH_B$.

 It is also useful to define the $\sharp$ operation on operators as well as super-operators. For an operator $X$ on $\cH_B$ define
 \begin{equation}\label{OPSHARP}
 X^\sharp := \lim_{T\to\infty}\frac{1}{2T} \int_{-T}^T e^{t\cK_P} X {\rm d} t\ .
 \end{equation}
 
 \begin{lem}
 For all super-operators $\cT$ on $\cH$, 
$\|\cT^\sharp\|_{1\to1} \leq \|\cT\|_{1\to1}$.
\end{lem}

\begin{proof}
This is immediate from the definition of $\cT^\sharp$ and the unitary invariance of the trace norm. 
\end{proof}

 An alternate expression for $\cT^\sharp$, also figuring in Davies' work \cite{Da74}, will be useful. 
 The operator $i \cK_P$  is self-adjoint on $\widehat{\cH}_B$. If $H_P\phi_\mu = \mu\phi_\mu$ and $H_P\phi_\nu = \nu \phi_\nu$, then
 $i\cK_P|\phi_\mu \rangle\langle \phi_\nu| = (\mu - \nu)|\phi_\mu \rangle\langle \phi_\nu|$. 
 Therefore, the spectrum of $i\cK_P$, $\sigma(i\cK_P)$, is related to the spectrum of $H_P$, $\sigma(H_P)$, by
 \begin{equation}\label{SpeciKP}
 \sigma(i\cK_P) = \{\mu -\nu\ :\ \mu,\nu\in \sigma(H_P)\}\ .
 \end{equation}

 Let $H_P = \sum_{\mu\in \sigma(H_P)}\mu P_\mu$ be the spectral decomposition of $H_P$.  For $\mu,\nu\in \sigma(H_P)$ define the orthogonal projections $Q_{\mu,\nu}$ on $\widehat{\cH}$ by
 \begin{equation}\label{DAVLIM54}
 Q_{\mu,\nu}(X) = P_\mu X P_\nu\ .
 \end{equation} 
 (One readily checks that  $Q^2_{\mu,\nu} =  Q_{\mu,\nu}$ and $Q^\dagger _{\mu,\nu} =  Q_{\mu,\nu}$, so that $Q_{\mu,\nu}$ is indeed an orthogonal projection.)
 Then 
 \begin{equation}\label{DAVLIM55}
 i\cK_P  = \sum_{\mu,\nu\in \sigma(H_P)} (\nu - \mu) Q_{\mu,\nu} \quad{\rm so\ that}\quad e^{t\cK_P} = \sum_{\mu,\nu\in \sigma(H_P)} e^{it(\nu - \mu)} Q_{\mu,\nu}\ .
 \end{equation} 
 The first equation in 
\eqref{DAVLIM55} is not exactly  the spectral decomposition of $\cK_P$ since  for $\omega\in \sigma(\cK_P)$, 
the corresponding spectral projection of $\cK_P$ is $\sum_{\mu,\nu\in \sigma(H_P)\ :\ \mu-\nu = \omega}Q_{\mu,\nu}$.
However, this somewhat finer decomposition  will be useful.

 By  \eqref{DAVLIM56}, ${\displaystyle 
 \cT^\sharp =  
 \sum_{\mu,\mu'\nu,\nu'\in \sigma(H_P)} Q_{\mu,\nu}\cT Q_{\mu',\nu'} \lim_{T\to\infty}\frac{1}{2T}\int_{-T}^T e^{it((\nu - \mu) -(\nu' - \mu'))}  {\rm d}t}$ and hence an alternate expression for 
 $\cT^\sharp$ is
 \begin{equation}\label{DAVLIM58}
 \cT^\sharp =  
 \sum_{\mu,\mu',\nu,\nu'\in \sigma(H_P)}\delta_{(\nu - \mu) -(\nu' - \mu')} Q_{\mu,\nu}\cT Q_{\mu',\nu'}\ .
 \end{equation}

The next theorem bounds the rate at which $\frac{1}{2T} \int_{-T}^T e^{-t\cK_P}\cT(e^{t \cK_P}X){\rm d}t$ converges to $\cT^\sharp$. This rate depends on
 \begin{equation}\label{Kgap}
 b := \min\{ |(\nu - \mu) -(\nu' - \mu')| \ :\ \mu,\mu'\nu,\nu'\in \sigma(H_P)\ ,\  (\nu - \mu) \neq(\nu' - \mu')\ \}\ .
 \end{equation}
 Evidently $b>0$, and in fact $b$ is the spectral gap of $|\cK_P|$.    In fact, we need a slightly more general bound that depends only on the length of the time interval and does not require it to be centered. 

\begin{thm}
For all $s\in \R$ and all $T>0$, 
\begin{equation}\label{OscEst} 
 \left\Vert 
 \cT^\sharp(X) - \frac{1}{2T}\int_{-T-s}^{T+s} e^{-t\cK_P}\cT(e^{t \cK_P}X){\rm d}t \right\Vert_{1\to 1} 
 \leq   |\sigma(H_P)|^4 \frac{3\pi}{2bT} \ 
 \end{equation}
 where $|\sigma(H_P)|$ denote the number of distinct eigenvalues of $H_P$. 
\end{thm}

\begin{proof}
By \eqref{DAVLIM58},  for all $X\in \widehat{H}_B$,
\begin{equation*} \cT^\sharp(X)  - \frac{1}{2T}\int_{-T}^T e^{-t\cK_P}\cT(e^{t \cK_P}X){\rm d}t =
 \sum_{\mu,\mu'\nu,\nu'\in \sigma(H_P)}  C_{\mu,\mu'\nu,\nu'}(T)Q_{\mu,\nu}\cT Q_{\mu',\nu'}(X)
 \end{equation*}
 where
 $$
 C_{\mu,\mu'\nu,\nu'}(T) :=  \frac{1}{2T}\int_{-T}^T e^{it((\nu - \mu) -(\nu' - \mu'))}  {\rm d}t - \delta_{(\nu - \mu) -(\nu' - \mu')} 
 $$
  Then $C_{\mu,\mu'\nu,\nu'}(T)  = 0$  for all $T$ if $(\nu - \mu) = (\nu' - \mu')$, and for  $(\nu - \mu) \neq(\nu' - \mu')$,
 $$
 |C_{\mu,\mu'\nu,\nu'}(T)| \leq \frac{2\pi/ b}{2T} = \frac{\pi}{2bT}\ .
 $$
 Moreover, for any $s\in \R$, and $(\nu - \mu) \neq(\nu' - \mu')$,
 $$
\left| \frac{1}{2T}\int_{-T}^T e^{it((\nu - \mu) -(\nu' - \mu'))} - \frac{1}{2T}\int_{-T-s}^{T+s} e^{it((\nu - \mu) -(\nu' - \mu'))}\right| \leq  \frac{\pi}{bT}\ .
 $$
 By \eqref{DAVLIM54}, for each $\mu,\nu$, $\|Q_{\mu,\nu}\|_{1\to 1} \leq 1$.  This proves \eqref{OscEst}. 
\end{proof} 

The next theorem says that if $\sD$ is a Lindblad generator acting on $\widehat{\cH}_B$, then so is $\sD^\sharp$, and yields an explicit formula of $\sD^\sharp$ in terms of the spectral decomposition of $H_P$
and the Lindblad form of $\sD$.  We are mainly interested in the case that $\sD = \cD_P$, but this spacial choice of $\sD$ does not enter the proof.

 \begin{thm}\label{INTMAIN} Let ${\displaystyle \sD R = \sum_{\ell=1}^{r }\left(2V_\ell R V_\ell^\dagger - V_\ell^\dagger V_\ell  R - R V_\ell^\dagger V_\ell  \right)  -i[H_L,R]}$ be any 
 Lindblad generator on $\widehat{\cH}_B$.
 Let $\sD^\sharp$ be defined in terms of $\sD$ through  \eqref{DAVLIM58}.  
 Let $H_P = \sum_{\lambda\in \sigma(H_P)}\lambda P_\lambda$ be the spectral decomposition of $H_P$ so that $P_\lambda$ is the orthogonal projection on the the 
eigenspace with eigenvalue $\lambda$, and $\sigma(H_P)$ is the spectrum of $H_P$. 

Then  $\sD^\sharp$  has the form 
 \begin{equation}\label{DPSHFO}
 \sD^\sharp R = \sum_{\omega\in \sigma(i\cK_P)} \sum_{j=1}^{n_A} \left(2 V_{j,\omega} R V_{j,\omega}^\dagger - V^\dagger_{j,\omega}V_{j,\omega} R - RV^\dagger_{j,\omega}V_{j,\omega} \right)  - i [(H_L)_0,R]
 \end{equation}
 where for any operator $Z$ on $\cH_B$ and any $\omega\in  \sigma(i\cK_P)$,
  \begin{equation}\label{DPSHFO2}
 Z_\omega := \sum_{\mu,\mu'\in \sigma(H_P)\ :\ \mu-\mu' = \omega}P_{\mu'}ZP_{\mu}\ .
  \end{equation}
 \end{thm}
 
 \begin{proof} For $V\in \widehat{\cH}_B$,  consider the  Lindblad generator 
${\displaystyle \sD R := 2VRV^\dagger  - V^\dagger VR - RV^\dagger V}$.
We now compute $\sD^\sharp$ term by term. First consider $\cT(X) = VXV^\dagger$. Then by \eqref{DAVLIM54} and \eqref{DAVLIM58},
$$
\cT^\sharp(X) = \sum_{\mu,\mu',\nu,\nu'\in \sigma(H_P)}\delta_{(\nu-\mu)- (\nu'-\mu')} P_{\mu}VP_{\mu'}XP_{\nu'}V^\dagger P_{\nu} \ ,
$$
For $\omega \in \sigma(\cK_P)$,  define $V_\omega$ by \eqref{DPSHFO2}.
  Then 
  ${\displaystyle 
\cT^\sharp(X) = \sum_{\omega \in \sigma(\cK_P)}V_\omega XV_\omega^\dagger}$. If for the same $V$, $\cT(X) = V^\dagger V X$, 
\begin{eqnarray*}
\cT^\sharp(X) &=& \sum_{\mu,\mu'\nu,\nu'\in \sigma(H_P)}\delta_{(\nu-\mu)- (\nu'-\mu')} P_{\mu}V^\dagger VP_{\mu'}XP_{\nu'} P_{\nu} \\
&=& \sum_{\mu,\mu'\nu,\nu'\in \sigma(H_P)}\delta_{(\nu-\mu)- (\nu'-\mu')}\delta_{\nu-\nu'} P_{\mu}V^\dagger V P_{\mu'}XP_{\nu} 
=\sum_{\mu\in \sigma(H_P)} P_\mu V^\dagger V P_\mu X\ . 
\end{eqnarray*}

We must show that 
${\displaystyle
\sum_{\omega\in \sigma(i\cK_P)} V^\dagger _\omega V_\omega = \sum_{\mu\in \sigma(H_P)}P_\mu V^\dagger V P_\mu}$.
 Define $P_{\lambda+\omega} = 0$ if $\lambda+\omega \notin \sigma(H_P)$. Then 
\begin{eqnarray*}
V^\dagger _\omega V_\omega   &=& \sum_{\mu,\mu'\nu,\nu'\in \sigma(H_P)\ :\ \mu-\mu' = \nu-\nu' = \omega}P_{\mu}V^\dagger P_{\mu'}P_{\nu'}VP_{\nu}\\
 &=& \sum_{\lambda\in \sigma(H_P)}P_{\lambda+\omega}V^\dagger P_{\lambda}V P_{\lambda+\omega}\\
\end{eqnarray*}
If $P_{\lambda+\omega} \neq 0$, $\omega = \mu - \lambda$ for some  $\mu\in \sigma(H_P)$.  
 For all $\lambda,\mu\in \sigma(H_P)$, there is exactly one $\omega\in \sigma(i\cK_P)$ such that
$\lambda+\omega = \mu$. To see this, note that $\omega_0 := \mu - \lambda\in \sigma(i\cK_P)$, so that there is at least one such $\omega$. But every other $\omega\in  \sigma(i\cK)$ satisfies either $\omega > \omega_0$ or 
$\omega < \omega_0$. Therefore, for every other $\omega$, $\lambda+\omega > \mu$ or $\lambda+\omega < \mu$,
and hence
$$
\sum_{\omega\in \sigma(i\cK_P)} V^\dagger _\omega V_\omega = \sum_{\omega\in \sigma(i\cK_P)} \sum_{\lambda\in \sigma(H_P)}P_{\lambda+\omega}V^\dagger P_{\lambda}VP_{\lambda+\omega} = 
\sum_{\lambda,\mu\in \sigma(H_P)}P_\mu V^\dagger P_\lambda V P_\mu = \sum_{\mu\in \sigma(H_P)}P_\mu V^\dagger V P_\mu\ .
$$
The same reasoning applies to $\cT(X) = XV^\dagger V$.

Finally, define $\cT(X) :=  -i[H_L,X]$.  Then
\begin{eqnarray*}
\cT^\sharp(X) &=& -i\sum_{\mu,\mu'\nu,\nu'\in \sigma(H_P)}\delta_{(\nu-\mu)- (\nu'-\mu')} \left( P_{\mu}H_L P_{\mu'}XP_{\nu'} P_{\nu} - P_{\mu}P_{\mu'}XP_{\nu'} H_L P_{\nu} \right)\\
&=& \sum_{\mu,\nu\in \sigma(H_P)} \left(P_\mu H_L P_\mu X P_\nu - P_\mu X P_\nu H_L P_\nu\right)\\
&=& \left(\textstyle{\sum_{\mu\in \sigma(H_P)}P_\mu H_L P_\mu}\right) X - X \left(\textstyle{\sum_{\mu\in \sigma(H_P)}P_\mu H_L P_\mu}\right)\ .
\end{eqnarray*}
Putting together the pieces yields the claim for the general Lindblad generator as stated. 
\end{proof}

 \begin{lem}\label{DHARPLEM} Let $\cT$ be any super operator on $\cH_B$.  For all $X\in {\rm ker}(\cT^\sharp)$, $X^\sharp \in {\rm ker}(\cT^\sharp)$ where $X^\sharp$ is defined in 
 \eqref{OPSHARP}.  In particular, if $\cD_P^\sharp$ is ergodic, then the unique density matrix $\bar R$ satisfying $\cD_P^\sharp \bar R$ commutes with $H_P$. 
 \end{lem}
 
 \begin{proof} For any real $s$, and any $X\in {\rm ker} (\cK_P)$,
 \begin{eqnarray*}
 e^{s\cK_P}\cT^\sharp(X) &=& \lim_{T\to\infty}\frac{1}{2T}\int_{-T}^T e^{(t+s)\cK_P}\cT( e^{-t\cK_P} X)  {\rm d}t\\
 &=& \lim_{T\to\infty}\frac{1}{2T}\int_{-T}^T e^{(t+s)\cK_P}\cT( e^{-(t+s)\cK_P} e^{s\cK_P}X)  {\rm d}t  = \cT^\sharp(e^{s \cK_P} X)\\
 \end{eqnarray*}
 Therefore, if $\cT^\sharp(X) =0$, then  $\cT^\sharp(e^{s \cK_P} X)= 0$ for all $s$, and hence $\cT^\sharp(X^\sharp) =0$.  Taking $\cT = \cD_P$, we see that the null space of $\cD_P^\sharp$ is invariant under that map $R\ \mapsto R^\sharp$, and hence if $\cD_P^\sharp$ is ergodic, the unique density matrix in its null space commutes with $H_P$.   \end{proof}

\begin{example}\label{EXAMP1} We consider a very simple 2 qubit system. Let $\cH_A$ and $\cH_B$  both be $\C^2$ with the usual inner product. Let 
  $|0\rangle := {\scriptscriptstyle \left(\begin{array}{c} 0\cr 1\end{array}\right)}$ and $|1\rangle := {\scriptscriptstyle \left(\begin{array}{c} 1\cr 0\end{array}\right)}$.
  Let $\sigma_1$, $\sigma_2$ and $\sigma_3$ be the Pauli matrices, and define 
  $$
  \sigma_- := \frac12(\sigma_1 -i\sigma_2) = |0\rangle\langle 1| \quad{\rm and }\quad  \sigma_+ := \sigma_-^\dagger = \frac12(\sigma_1+i\sigma_2) =  |1\rangle\langle 0|\ .
  $$
  For $\beta \in \R$, define the jump operators
   $V_1$ and $V_2$ by 
  \begin{equation}\label{ProjJump}
  V_1 := \sqrt{ \frac{ e^{-\beta/2} }{e^{\beta/2}+e^{-\beta/2}}}\sigma_+ \quad{\rm and}\quad V_2 :=  \sqrt{ \frac{ e^{\beta/2} }{e^{\beta/2}+e^{-\beta/2}}}\sigma_-\ .
  \end{equation}
  and then define  the Lindblad generator $\cD_A$ by
  \begin{equation}\label{DAEX1DEF}
 \cD_A(R) :=  \sum_{j=1}^{2 }\left(2V_j R V_j^\dagger - V_j^\dagger V_j  R - R V_j^\dagger V_j  \right)\ .
 \end{equation}
 As we now show, $\cD_A$ is ergodic and gapped, and the unique steady state is a Gibbs state for $H_A =\sigma_3$.

Define
   ${\displaystyle \pi_A :=  \frac12 \one_A - \frac12 \tanh(\beta/2) \sigma_3}$.  Then with $H_A := \sigma_3$, $\pi_A = \frac{1}{Z_{\beta/2}}e^{-(\beta/2) H_A}$. 
Simple calculations show that 
\begin{equation}\label{DASPECEX1}
\cD_A(\pi_A) =0\ ,\quad  \cD_A(\sigma_-) =  -\sigma_-\ ,\quad  
\cD_A \sigma_+ = -\sigma_+\quad{\rm and} \quad\cD_A(\sigma_3) = -2(\sigma_3)\ .
\end{equation}
   Therefore, $\cD_A$ is ergodic with unique steady state $\pi_A$, and its spectral gap is $1$.   Let $\{X_0,X_1,X_2,X_3\}$ be the dual basis to $\{Y_0,Y_1,Y_2,Y_3\} := \{\pi_A,\sigma_+,\sigma_-,\sigma_3\}$. 
   Using properties of the Pauli matrices, it is easy to see that
  \begin{equation}\label{DUBAEX1}
  \{X_0,X_1,X_2,X_3\} = \{\one_A, \sigma_+,\sigma_-, \tfrac12(\tanh(\beta/2)\one_A + \sigma_3)\}\ .
  \end{equation}
  By \eqref{DASPECEX1} and Theorem~\ref{DUBAJF2}, 
  \begin{equation}\label{DASPECEX2}
  \cD_A^\dagger(X_0) =0\ ,\quad\cD_A^\dagger(X_1)=-X_1\ ,\quad
\cD_A^\dagger(X_2)=-X_2\quad{\rm and}\quad \cD_A^\dagger(X_3)=-2X_3\ .
\end{equation}

   Take $H := H_A\otimes \one_B + H_{AB} + \one_A\otimes H_B$ where $H_A = \sigma_3$ as above, 
 $$
  H_{AB}  = \sigma_-\otimes \sigma_+ + \sigma_+\otimes \sigma_- \ ,
$$
  and  $H_B$ is a self-adjoint operator to be specified.  To  compute the projected Hamiltonian and the projected dissipator,  the first step is to write the Hamiltonian in the form
  $H = \sum_{j=0}^3 X_j\otimes G_j$ as in \eqref{HEXPTHM} of Lemma~\ref{LINDINGEN}. By \eqref{DUBAEX1} and simple computations,
  $H_A = 2X_3 -\tanh(\beta/2)X_0$ and then
  $$
 H  = 2 X_3 \otimes\one + X_1\otimes \sigma_+ + X_2\otimes \sigma_- + X_0\otimes \left(H_B  -\tanh(\beta/2)\one \right)\ .
  $$
  This gives the expansion $H = \sum_{j=0}^3 X_j\otimes G_j$
 where 
\begin{equation}\label{GVALEX1}
G_0 =H_B  -\tanh(\beta/2)\one\ ,\quad  G_1 = \sigma_+\ ,\quad
G_2 = \sigma_-\quad{\rm and}\quad G_3 = 2\one\ .
\end{equation}

By Lemma~\ref{CommLem}, 
\begin{equation}\label{HBHPT}
  H_P := \tr_A[(\pi_A\otimes \one_B)H] = G_0 =H_B  -\tanh(\beta/2)\one\ .
  \end{equation}
  
 To complute $\cD_P$, we use Lemma~\ref{LINDINGEN} (or the corresponding formulas in \cite{PEPS}). First,
  we work out the  matrix 
  $M_{j,k} :=   -\langle \cS_A ^\dagger X_k, X_j\pi_A\rangle$ in \eqref{MMATDEF} where $1 \leq j,k\leq 3$. 
  Then defining $\lambda_1 = \lambda_2 =-1$ and $\lambda_3 =-2$, it follows from \eqref{DASPECEX2} and $\cS_A^\dagger\cD_A^\dagger = \cQ_A^\dagger$  that 
$$
M_{j,k} = -\frac{1}{\lambda_k}\langle X_k,X_j\rangle_{\pi_A}
$$
where the inner product is the GNS inner product determined by $\pi_A$ as in \eqref{GNSDEF}. 
Simple computations show that $\langle X_k,X_j\rangle_{\pi} = \delta_{j,k}\tr[X_j^\dagger X_j\pi_A]$.  Therefore, $M$ is diagonal and positive definite. The decomposition  $M = A+iB$ is then trivial with $B=0$, so that $H_L=0$ by \eqref{PEPSLG3}. Since $A$ is already diagonal, it follows also from \eqref{PEPSLG3} that the jump operators are multiples of 
$G_1$, $G_2$ and $G_3$. But since $G_3$ is a multiple of the identity,
which would make a trivial contribution to $\cD_P$, we need only consider the jump operators that are multiples of $G_1= \sigma_+$ and $G_2 = \sigma_-$. Since 
$$
 \langle X_1,X_1\rangle_{\pi_A} =  \frac{ e^{-\beta/2} }{e^{\beta/2}+e^{-\beta/2}}\quad{\rm and}\quad   \langle X_2,X_2\rangle_{\pi_A} =  \frac{ e^{\beta/2} }{e^{\beta/2}+e^{-\beta/2}}\ ,
  $$
  these are exactly the jump operators $V_1$ and $V_2$ given by
  \eqref{ProjJump}. 
  That is,
\begin{equation}\label{EX1DP}
    \cD_P(R) :=  \sum_{j=1}^{2 }\left(2V_j R V_j^\dagger - V_j^\dagger V_j  R - R V_j^\dagger V_j  \right)\ .
 \end{equation}
 This is true independent of the choice of $H_B$; by \eqref{GVALEX1}, $H_B$ only figures in $G_0$, and $\cD_P$ does not depend on $G_0$.
 
Comparing this with \eqref{DAEX1DEF}, we see that  $\cD_P$  has the exact same form as $\cD_A$,  which we already know to be  ergodic and gapped with unique ground state $\pi_A$.
 
 Moreover, simple computations show that if $H_B =0$, or more generally if $H_B$ commutes with $\sigma_3$, then 
  $[H, \pi_A\otimes \pi_A] = [H_{AB}, \pi_A\otimes \pi_A] = 0$
  for all $\beta$.  Therefore, with such a choice of $H_B$,  $\pi_A\otimes \pi_A$ is the exact steady state for $\cK + \gamma \cD_A$ for all values of  $\gamma$.  It will be no surprise that in this case, simple computations show that $\cD_P^\sharp = \cD_P$. (The may be checked by following the procedure implemented just below.)
  
The situation is more interesting when we choose $H_B$  so that $H_P$ does not commute with $\sigma_3$.  By \eqref{HBHPT}, for an appropriate choice of $H_B$ we can arrange that $H_P$ is any self-adjoint operator on $\cH_B$, and we have seen that $\cD_P$ does not depend on the choice of $H_B$. 

In the rest of this example, we 
consider the case in which $H_B := \sigma_2 +\tanh(\beta/2)\one$ so that 
$H_P = \sigma_2$, and $\cD_P$ is given by \eqref{EX1DP}. 

  We now turn to the computation of $\cD_P^\sharp$ for this model, illustrating the application of Theorem~\ref{INTMAIN}. 
   The spectrum of $H_B$ is  $\{-1,1\}$, and the spectral projections are given by  $P_1= \frac12(\one + \sigma_2)$
  and $P_{-1} =  \frac12(\one - \sigma_2)$. 
  By \eqref{SpeciKP}, the spectrum of $i\cK_P$ consists of differences of eigenvalues of $H_p$, so that it is $\{-2,0,2\}$.  From \eqref{DPSHFO2}, for any $Z\in \widehat{\cH}_B$, we have the following expressions for $Z_\omega$, $\omega = -2,0,2$. 
  $$
Z_{-2} = P_{1}ZP_{-1}\ ,\quad Z_{0} = P_{-1}ZP_{-1} + P_{1}ZP_{1}\quad{\rm and}\quad Z_{2} = P_{-1}ZP_{1}\ .
  $$
We now appply this for $Z = V_1$ and $Z = V_2$ with $V_1,V_2$ given by 
\eqref{ProjJump}. Define 
$c := \sqrt{ \frac{ e^{-\beta/2} }{e^{\beta/2}+e^{-\beta/2}}}$ and $s := \sqrt{ \frac{ e^{\beta/2} }{e^{\beta/2}+e^{-\beta/2}}}$ so that
$V_1 = c\sigma_+$ and $V_2 = s\sigma_-$. 

 Simple computations yield $P_1\sigma_1P_{-1} =\frac12(\sigma_1 -i \sigma_3)$ and $P_1\sigma_1P_1 = P_{-1}\sigma_1 P_{-1} =0$. Likewise, $P_1\sigma_2P_{-1} =0$, $P_1\sigma_2P_1 = P_1$ and 
 $P_{-1}\sigma_2P_{-1} = -P_{-1}$. 
 Therefore, define
 \begin{equation}\label{aDEF}
 a := \frac12(\sigma_1 - i\sigma_3)\ ,
 \end{equation}
 and note that $a^2 =0$  and $a^\dagger a + a a^\dagger = \one$. 
 We compute  
 \begin{eqnarray*}
 V_{1,2} &=& \frac{c}{2} a\ , \quad V_{1,-2} = \frac{c}{2}a^\dagger \quad{\rm and}\quad V_{1,0} = \frac{ic}{2}\sigma_2\\
  V_{2,2} &=& \frac{s}{2} a\ , \quad V_{2,-2} = \frac{s}{2}a^\dagger \quad{\rm and}\quad V_{2,0} = \frac{is}{2}\sigma_2\ ,
 \end{eqnarray*}
 and then by \eqref{DPSHFO}, 
 \begin{equation}\label{DPSHARP}
 \cD_P^\sharp(R) = \frac12 a R a^\dagger + \frac12  a^\dagger R a + \frac12 \sigma_2 R \sigma_2 -  R\ ,
 \end{equation}
 independent of $\beta$,
 where $a$ is given by \eqref{aDEF}.
It follows that
  \begin{equation}\label{DPSHARP2}
 \cD_P^\sharp(\one) = 0\ ,\quad \cD_P^\sharp(a) = -a\ ,\quad  \cD_P^\sharp(a^\dagger) = -a^\dagger \ ,\quad \cD_P^\sharp(\sigma_2) = -\frac32\sigma_2\ .
 \end{equation}
 Therefore, the normalized identity, $\frac12\one$, is the unique steady state for  $\cD_P^\sharp$, independent of $\beta$. Note that $\cD_P^\sharp$ is self-adjoint 
 with respect to the GNS inner product for its unique steady state, and therefore it satisfies the detailed balance condition, which makes it a much simpler generator to work with than $\cK_P + \gamma^{-1}\cD_P$ for any finite 
 $\gamma$.  
 
Finally, simple but tedious computations with the Pauli basis  yield the explicit form of the operator $\cB_P$ defined in \eqref{BURDEF}, and hence the explicit form of the operator $\cK_P + \gamma^{-1}\cD_P + \gamma^{-2}\cB_P$.
 A convenient necessary and sufficient condition for a super-operator to generate a CPTP evolution due to Gorini, Kossakowski and Sudarshan \cite{GKS76} may be applied to show that this operator does not generate a CPTP evolution for any $\gamma>0$.  
 \end{example}

  \section{Oscillation averaging and the projected evolution}\label{InteractionPicture}

The main result of this section relates the semigroups generated by $\cL_{P,\gamma}$ and by $\cD_P^\sharp$.

\begin{thm}\label{PROJMOZLTH}   There are finite constants $C_0$ and $C_1$ independent of $\gamma$ such that for all $T>0$, 
\begin{equation}\label{PROJMOZLTH10} 
\| e^{-\tau \gamma \cK_P} e^{\gamma\tau   \cL_{P,\gamma}} - e^{\tau \cD_P^\sharp}\|_{1\to 1} \leq \frac{1}{\gamma} C_0T e_{\phantom{0}}^{C_1T} \ .
\end{equation}
\end{thm}

\begin{lem}\label{PROJMOZLLEM} 
For any  $\tau> 0$, define
$Y_P^{(\gamma)}(\tau) := e^{-\tau \gamma \cK_P} e^{\tau \gamma \cL_{P,\gamma}}$.  Then for all $\tau$, 
\begin{equation}\label{PROJMOZLLEM1} 
\|Y_P^{(\gamma)}(\tau) \|_{1\to1} \leq 1
\end{equation}
and as a function of $\tau$, $Y_P^{(\gamma)}$ is globally Lipschitz, uniformly in $\gamma$. More specifically, for all $\tau'> \tau \geq 0$. 
\begin{equation}\label{PROJMOZLLEM2} 
\|Y_P^{(\gamma)}(\tau') - Y_P^{(\gamma)}(\tau)\|_{1\to 1}    \leq   \|\cD_P\|_{1\to 1}(\tau' - \tau)
\end{equation}
\end{lem}

\begin{proof} The bound \eqref{PROJMOZLLEM1}  is a consequence of the fact that 
$Y_P^{(\gamma)}(\tau)$ is CPTP for all $\tau$.
Next, by Duhammel's formula,
${\displaystyle
e^{\tau \gamma \cL_{P,\gamma}}  = e^{\tau \gamma \cK_P} + \int_0^\tau e^{(\tau- \sigma) \gamma \cK_P} \cD_P e^{\sigma \gamma \cL_{P,\gamma}}  {\rm d}\sigma}$
and hence
\begin{equation}
Y_P^{(\gamma)}(\tau)  = \one +  \int_0^\tau e^{- \sigma \gamma \cK_P} \cD_Pe^{\sigma \gamma \cK_P} e^{-\sigma \cK_P} e^{\sigma \gamma \cL_{P,\gamma}}  {\rm d}\sigma
= \one +  \int_0^\tau e^{- \sigma \gamma \cK_P} \cD_Pe^{\gamma \sigma \cK_P} Y_P^{(\gamma)}(\sigma) {\rm d}\sigma\ .\label{PROJMOZLLEM3} 
\end{equation}
Therefore, for $\tau' > \tau > 0$, using the fact that $\cK_P$ generates a group of unitary conjugations and the trace norm is unitarily invariant, 
\begin{eqnarray*}
\|Y_P^{(\gamma)}(\tau') - Y_P^{(\gamma)}(\tau)\|_{1\to 1}   &\leq & \int_{\tau}^{\tau'} \left\Vert  e^{- \sigma \gamma \cK_P} \cD_Pe^{\gamma \sigma \cK_P} Y_P^{(\gamma)}(\tau)\right\Vert_{1\to 1}  {\rm d}\sigma\\
&\leq & \int_{\tau}^{\tau'} \left\Vert  \cD_Pe^{\gamma \sigma \cK_P} Y_P^{(\gamma)}(\tau)\right\Vert_{1\to 1}  {\rm d}\sigma\\
&\leq& \|\cD_P\|_{1\to1}  \int_{\tau}^{\tau'} \left\Vert  e^{\gamma \sigma \cK_P} Y_P^{(\gamma)}(\tau)\right\Vert_1 {\rm d}\sigma = \|\cD_P\|_{1\to 1}(\tau' - \tau)\ ,
\end{eqnarray*}
which proves \eqref{PROJMOZLLEM2}. 
\end{proof}

The proof of Theorem~\ref{PROJMOZLTH} takes advantage of the perturbation theory for Volterra integral equations. In fact, in \eqref{PROJMOZLLEM3} we have already written 
$Y_P^{(\gamma)}(\tau)$ is the solution to a Volterra integral equation. The relevant definition and results are recalled in Appendix~\ref{VOLTERRA}.

\begin{proof}[Proof of Theorem~\ref{PROJMOZLTH}] With $Y_P^{(\gamma)}$ defined as in Lemma~\ref{PROJMOZLLEM}, we must bound $\|Y_P^{(\gamma)}(\tau) - e^{\tau \cD_P^\sharp}\|_{1\to 1}$. We do this by showing that
$Y_P^{(\gamma)}$ arises as the solution of two Volterra integral equations, which we then compare using Theorem~\ref{VOLTIWTHM}. 

Define $\sC$ to be the space of continuous functions form $[0,T]$ into $\widehat{\cH}_B$ equipped with the norm $\|Y(\tau)\|_\sC := \max_{0\leq \tau \leq T}\{\|Y(\tau)\|_{1\to 1}\|\}$.

First, define  $\sH^{(\gamma)}$ to act on $\sC$ by
\begin{eqnarray}
\sH^{(\gamma)} Y(\tau) &:=&  \int_0^\tau e^{- \sigma \gamma \cK_P} \cD_Pe^{\gamma \sigma \cK_P} Y(\sigma) {\rm d}\sigma\nonumber\\
&=& \sum_{\mu,\nu,\mu',\nu'\in \sigma(H_P)} Q_{\mu,\nu} \cD_PQ_{\mu',\nu'}  \int_0^\tau e^{i\gamma(\omega - \omega')\sigma} Y(\sigma) {\rm d}\sigma\ ,\label{PROJMOZLTH5} 
\end{eqnarray}
where $\omega = \mu -\nu$ and $\omega' = \mu'-\nu'$.  Define $Z(\tau) = \one$ for all $\tau$.  Then by \eqref{PROJMOZLLEM3},  $Y_P^{(\gamma)}$ solves $Y = \sH Y + Z$. 

Second, define $\widetilde{\sH}^{(\gamma)}$ to act on $\sC$ by
$$
\widetilde{\sH}^{(\gamma)} Y(\tau) :=  \int_0^\tau  \cD_P^\sharp Y(\sigma) {\rm d}\sigma\ ,
$$
and define
$$
Z^{(\gamma)}(\tau) := \one + \sH^{(\gamma)} Y_P^{(\gamma)}(\tau) - \widetilde{\sH}^{(\gamma)} Y_P^{(\gamma)}(\tau)\ .
$$
Then 
$Y_P^{(\gamma)}$ solves the Volterra integral equation
$Y  = \widetilde{\sH}^{(\gamma)} Y  + Z^{(\gamma)}$.
Since $\widetilde{Y}_P^{(\gamma)} := e^{\tau \cD_P^\sharp}$ satisfies $\widetilde{Y}_P^{(\gamma)} = \widetilde{\sH}^{(\gamma)} \widetilde{Y}_P^{(\gamma)}  +Z$ where again
$Z(\tau) =\one$ for all $\tau$,
it follows from the second part of Theorem~\ref{VOLTIWTHM} that
$$
\| Y_P^{(\gamma)}  - \widetilde{Y}_P^{(\gamma)} \|_{\sC} \leq e^{2GT} \sup_{0 \leq \tau \leq T}\left\{\  \| \sH^{(\gamma)} Y_P^{(\gamma)}(\tau) - \widetilde{\sH}^{(\gamma)} Y_P^{(\gamma)}(\tau)\|_{1\to 1} \ \right\}\ .
$$
By \eqref{PROJMOZLTH5}, 
$$
\sH^{(\gamma)} Y_P^{(\gamma)}(\tau) - \widetilde{\sH}^{(\gamma)} Y_P^{(\gamma)}(\tau) =
\sum_{\mu,\nu,\mu',\nu'\in \sigma(H_P)\ ,\ \omega \neq \omega'} Q_{\mu,\nu} \cD_PQ_{\mu',\nu'}  \int_0^\tau e^{i\gamma(\omega - \omega')\sigma}  Y_P^{(\gamma)}(\sigma) {\rm d}\sigma
$$

Then since for $\omega \neq \omega'$,
\begin{eqnarray}
\int_0^\tau e^{i\gamma(\omega - \omega')\sigma}  Y_P^{(\gamma)}(\sigma) {\rm d}\sigma &=& 
\frac12\int_0^\tau \left(e^{i\gamma(\omega - \omega')\sigma} - e^{i\gamma(\omega - \omega')(\sigma+ \frac{\pi}{\gamma(\omega - \omega')}} \right)Y_P^{(\gamma)}(\sigma) {\rm d}\sigma\nonumber\\
&=& \frac12\int_{\frac{\pi}{\gamma(\omega - \omega')}}^\tau e^{i\gamma(\omega - \omega')\sigma}\left(Y_P^{(\gamma)}(\sigma) -Y_P^{(\gamma)}\left(\sigma- \tfrac{\pi}{\gamma(\omega - \omega')}\right)\right){\rm d}\sigma\label{LIPTERM}\\
&+& \frac12 \int_0^{\frac{\pi}{\gamma(\omega - \omega')}} e^{i\gamma(\omega - \omega')\sigma}Y_P^{(\gamma)}(\sigma) {\rm d}\sigma\label{BTERM1}\\
&-& \frac12 \int_\tau^{\tau +\frac{\pi}{\gamma(\omega - \omega')}} e^{i\gamma(\omega - \omega')\sigma}Y_P^{(\gamma)}\left(\sigma- \tfrac{\pi}{\gamma(\omega - \omega')}\right) {\rm d}\sigma\ .\label{BTERM2}
\end{eqnarray}
By \eqref{PROJMOZLLEM2}  of Lemma~\ref{PROJMOZLLEM}, 
$$\left\|Y_P^{(\gamma)}(\sigma) -Y_P^{(\gamma)}\left(\sigma- \tfrac{\pi}{\gamma(\omega - \omega')}\right)\right\|  \leq \|\cD_P\|_{1\to 1}\frac{\pi}{\gamma|\omega - \omega'|} \leq  \|\cD_P\|_{1\to 1}
\frac{\pi}{\gamma b}$$
where $b$ is the gap for $\cK_P$ as defined in \eqref{Kgap}. Therefore, the super-operator trace norm of the integral in \eqref{LIPTERM} is no more than $ \frac12 \|\cD_P\|_{1\to 1}
\frac{T\pi}{\gamma b}$. 
By \eqref{PROJMOZLLEM1}, $\|Y_P^{(\gamma)}(\sigma)\|_{1\to1} \leq 1$ for all $\sigma$, and hence for all $0 \leq \tau \leq T$, the super-operator trace norm of each of the integrals in \eqref{BTERM1} and \eqref{BTERM2} is bounded by $\frac12 \frac{\pi}{\gamma b}$.  Altogether, 
$$
\left\Vert \int_0^\tau e^{i\gamma(\omega - \omega')\sigma}  Y_P^{(\gamma)}(\sigma) {\rm d}\sigma\right\Vert_{1\to1} \leq \left(\frac{1}{2} \|\cD_P\|_{1\to 1} + 1\right)\frac{T \pi}{\gamma |\omega - \omega'|}\ .
$$
By \eqref{DAVLIM54}, $\|Q_{\mu,\nu}\|_{1\to 1} \leq 1$ for all $\mu,\nu$, and hence
for each $\mu,\nu,\mu',\nu'$, $\|Q_{\mu,\nu} \cD_PQ_{\mu',\nu'} \|_{1\to1} \leq  \|\cD_P\|_{1\to1}$.  The cardinality of the set $\{\mu,\nu,\mu',\nu'\in \sigma(H_P)\ ,\ \omega \neq \omega'\}$ is evidently no more that 
$|\sigma(H_P)|^4$, where $|\sigma(H_P)|$ is the number of distinct eigenvalues of $H_P$. Therefore,
$$
\|\sH^{(\gamma)} Y_P^{(\gamma)}(\tau) - \widetilde{\sH}^{(\gamma)} Y_P^{(\gamma)}(\tau) \|_{1\to1} \leq  |\sigma(H_P)|^4 \|\cD_P\|_{1\to 1} \left(\frac{1}{2} \|\cD_P\|_{1\to 1} + 1\right)\frac{T \pi}{\gamma b}
$$
where $b$ is the gap for $\cK_P$ as defined in \eqref{Kgap}. 
This proves \eqref{PROJMOZLTH10}. 
\end{proof}

Combing Theorem~\ref{PROJMOZLTH} with Theorem~\ref{COHERENTSC}  yields tight intertwining relation between the semigroups $e^{\gamma \tau \cL_\gamma}$ and $e^{\tau \cD_P^\sharp}$.

\begin{thm}\label{PROJMOZLTHA}   There are finite constants $C_0$ and $C_1$ independent of $\gamma$ such that for all $T>0$, and all density matrices $R_0$ on $\cH_B$, 
\begin{equation}\label{PROJMOZLTHA10} 
\| e^{-\tau \gamma \cK_P} \tr_A e^{\gamma\tau   \cL_{\gamma}}\cV R_0   - e^{\tau \cD_P^\sharp}R_0\|_{1} \leq \frac{1}{\gamma} C_0T e_{\phantom{0}}^{C_1T} \ .
\end{equation}
\end{thm}

\begin{proof} Replacing $t$ by $\gamma\tau$ and $T$ by $\gamma T$ in \eqref{COHERENTSC1} of Theorem~\ref{COHERENTSC} 
  \begin{equation}\label{COHERENTSC1Z}
 \| \tr_A  e^{\gamma\tau\cL_\gamma}\cV R_0 - e^{\gamma\tau\cL_{P,\gamma}}R_0\|_1 \leq  \frac{1}{\gamma} CT  \ 
 \end{equation}
 for all $\tau \leq \gamma T$. By the unitary invariance of the trace norm, 
  \begin{equation}\label{COHERENTSC1Z1}
   \| \tr_A  e^{\gamma\tau\cL_\gamma}\cV R_0 - e^{\gamma\tau\cL_{P,\gamma}}R_0\|_1 = \| e^{-\tau \gamma \cK_P}\tr_A  e^{\gamma\tau\cL_\gamma}\cV R_0 - e^{-\tau \gamma \cK_P}e^{\gamma\tau\cL_{P,\gamma}}R_0\|_1\ .
    \end{equation}
By Theorem~\ref{PROJMOZLTH},
 \begin{equation}\label{COHERENTSC1Z2}
\| e^{-\tau \gamma \cK_P} e^{\gamma\tau R_0    \cL_{P,\gamma}} - e^{\tau \cD_P^\sharp}R_0 \|_{1} \leq \frac{1}{\gamma} C_0T e_{\phantom{0}}^{C_1T} \ .
 \end{equation}
By \eqref{COHERENTSC1Z},  \eqref{COHERENTSC1Z1} and \eqref{COHERENTSC1Z2} together with the triangle inequality, we obtain
\eqref{PROJMOZLTHA10} after harmonizing the constants. 
\end{proof}

We now apply Theorem~\ref{PROJMOZLTH} and Theorem~\ref{PROJMOZLTHA}  to mixing times. 
The next theorem relates $t_{{\rm mix}}(\cD_P^\sharp)$ to ${\displaystyle \frac{t_{{\rm mix}}(\cL_{P,\gamma},\epsilon')}{\gamma}}$ and ${\displaystyle \frac{t_{{\rm mix}}(\cL_{\gamma},\epsilon')}{\gamma}}$.

\begin{lem}\label{MIXRELSH} {\it (1)} For all $0 < \epsilon < \epsilon' <\tfrac12$, 
\begin{equation}\label{MIXRELSH1}
\varlimsup_{\gamma\to\infty}\frac{t_{{\rm mix}}(\cL_{P,\gamma},\epsilon')}{\gamma} \leq  
t_{{\rm mix}}(\cD_P^\sharp,\epsilon) \quad {\rm and}\quad  \varlimsup_{\gamma\to\infty}\frac{t_{{\rm mix}}(\cL_{\gamma},\epsilon')}{\gamma} \leq  t_{{\rm mix}}(\cD_P^\sharp,\epsilon)\\ .
\end{equation}
 In particular,  when $\cD_P^\sharp$ is ergodic and gapped, for all sufficiently large  $\gamma$, $\cL_{P,\gamma}$ and $\cL_\gamma$ are ergodic and gapped. 

\smallskip
\noindent{\it (2)} For all $0 < \epsilon < \epsilon' <\tfrac12$,
$$
t_{{\rm mix}}(\cD_P^\sharp,\epsilon') \leq \varliminf_{\gamma\to \infty}\frac{t_{{\rm mix}}(\cL_{P,\gamma},\epsilon)}{\gamma}
\quad{\rm and}\quad 
t_{{\rm mix}}(\cD_P^\sharp,\epsilon') \leq \varliminf_{\gamma\to \infty}\frac{t_{{\rm mix}}(\cL_{\gamma},\epsilon)}{\gamma}
$$
In particular, if for some $0 <\epsilon <\tfrac12$, either ${\displaystyle
\varlimsup_{\gamma\to\infty}\frac{t_{{\rm mix}}(\cL_{P,\gamma},\epsilon)}{\gamma}<\infty}$ or 
${\displaystyle
\varlimsup_{\gamma\to\infty}\frac{t_{{\rm mix}}(\cL_{P,\gamma},\epsilon)}{\gamma}<\infty}$, then $\cD_P^\sharp$ is ergodic and gapped. 
\end{lem}

\begin{proof}
Let $R_0$ and $R_1$ be any two density matrices on $\cH_B$. Then since
$$
 e^{\gamma \tau \cK_P}(e^{\gamma\tau \cL_{P,\gamma}} R_0 - e^{\gamma\tau \cL_{P,\gamma}}R_1) - (e^{\tau \cD_P^\sharp}R_0 - e^{\tau \cD_P^\sharp}R_1)  = 
( e^{\gamma \tau \cK_P}e^{\gamma\tau \cL_{P,\gamma}}- e^{\tau \cD_P^\sharp})(R_0- R_1)\ ,
$$
it follows from Theorem~\ref{PROJMOZLTH} that 
${\displaystyle 
\left| \|e^{\gamma\tau \cL_{P,\gamma}} R_0 - e^{\gamma\tau \cL_{P,\gamma}}R_1 \|_1 - \|e^{\tau \cD_P^\sharp}R_0 - e^{\tau \cD_P^\sharp}R_1\|_1\right| \leq \frac{2}{\gamma} C_0T e_{\phantom{0}}^{C_1T} }$.
Therefore
\begin{equation}\label{MIXRELSH2}
\left|d_{{\rm TV}}(e^{\gamma\tau \cL_{P,\gamma}} R_0, e^{\gamma\tau \cL_{P,\gamma}} R_1) - d_{{\rm TV}}(e^{\tau \cD_P^\sharp}R_0,e^{\tau \cD_P^\sharp}R_1)\right| \leq
\frac{1}{\gamma} C_0T e_{\phantom{0}}^{C_1T}\ .
\end{equation}
It follows that for $\tau \geq t_{{\rm mix}}(\cD_P^\sharp,\epsilon)$, so that $d_{{\rm TV}}(e^{\tau \cD_P^\sharp}R_0,e^{\tau \cD_P^\sharp}R_1) \leq \epsilon$, 
\begin{equation*}
d_{{\rm TV}}(e^{\gamma\tau \cL_{P,\gamma}} R_0, e^{\gamma\tau \cL_{P,\gamma}} R_1) \leq \epsilon + \frac{1}{\gamma} C_0(t_{{\rm mix}}(\cD_P^\sharp,\epsilon)) e_{\phantom{0}}^{C_1 t_{{\rm mix}}(\cD_P^\sharp,\epsilon)}
\end{equation*}
Hence for all $\gamma$ such that $(\epsilon' -\epsilon)\gamma \geq  C_0(t_{{\rm mix}}(\cD_P^\sharp)) e_{\phantom{0}}^{C_1 t_{{\rm mix}}(\cD_P^\sharp)}$,
$d_{{\rm TV}}(e^{\gamma\tau \cL_{P,\gamma}} R_0, e^{\gamma\tau \cL_{P,\gamma}} R_1) \leq \epsilon'$. This proves the first inequality in \eqref{MIXRELSH1}. The second is proved in exactly the same way
using Theorem~\ref{PROJMOZLTHA} in place of Theorem~\ref{PROJMOZLTH}.
The final statement now follows from Theorem~\ref{POSLEM}. This proves  {\it (1)}.

To prove {\it (2)}, suppose that ${\displaystyle \tau_* := \varliminf_{\gamma\to \infty}\frac{t_{{\rm mix}}(\cL_{P,\gamma},\epsilon)}{\gamma} < \infty}$. The for $\tau \geq \tau_*$, there exist arbitrarily large $\gamma$ such that 
$t_{{\rm mix}}(\cL_{P,\gamma})\leq \gamma \epsilon$. 
By \eqref{MIXRELSH2}, for such $\tau$ and $\gamma$,
$$
d_{{\rm TV}}(e^{\tau \cD_P^\sharp}R_0,e^{\tau \cD_P^\sharp}R_1) \leq \epsilon +  \frac{1}{\gamma} C_0T e_{\phantom{0}}^{C_1T}\ .
$$
For sufficiently large $\gamma$, the right side is less than $\epsilon'$, while the left side is independent of $\gamma$. This proves the first statement, and the second is proved in exactly the same way
using Theorem~\ref{PROJMOZLTHA} in place of Theorem~\ref{PROJMOZLTH}.
The final statement now follows from Theorem~\ref{POSLEM}.
\end{proof}

A point of continuity of 
$t_{{\rm mix}}(\cD_P^\sharp(\epsilon))$ is a value 
$0 < \epsilon_0<\tfrac12$ such that ${\displaystyle \lim_{\epsilon\to \epsilon_0}t_{{\rm mix}}(\cD_P^\sharp,\epsilon) =
t_{{\rm mix}}(\cD_P^\sharp,\epsilon_0)}$. Since $t_{{\rm mix}}(\cD_P^\sharp(\epsilon)$ increases monotonically on $(0,\tfrac12)$,
all but countably many $\epsilon_0\in (0,\tfrac12)$ are points of continuity. 

\begin{thm}\label{TMIXZL3} Suppose that $t_{{\rm mix}}(\cD_P^\sharp,\epsilon) <\infty$ for some, and hence all $0 < \epsilon < \infty$. Then at all points of continuity $\epsilon_0$ of 
$t_{{\rm mix}}(\cD_P^\sharp,\epsilon)$,
${\displaystyle \lim_{\gamma\to\infty} \frac{t_{{\rm mix}}(\cL_{\gamma},\epsilon_0)}{\gamma}}$ and  
${\displaystyle \lim_{\gamma\to\infty} \frac{t_{{\rm mix}}(\cL_{P,\gamma},\epsilon_0)}{\gamma}}$ and satisfy
$$
\lim_{\gamma\to\infty} \frac{t_{{\rm mix}}(\cL_{\gamma},\epsilon_0)}{\gamma} = 
\lim_{\gamma\to\infty} \frac{t_{{\rm mix}}(\cL_{P,\gamma},\epsilon_0)}{\gamma} = t_{{\rm mix}}(\cD_P^\sharp,\epsilon_0)\ .
$$
\end{thm}

\begin{proof}
For any point of continuity $\epsilon_0$ of 
$t_{{\rm mix}}(\cD_P^\sharp(\epsilon))$, and all $\delta>0$ such that
$0 < \epsilon_0-\delta < \epsilon+\delta < \tfrac12$,
$$
t_{{\rm mix}}(\cD_P^\sharp,\epsilon_0-\delta) \leq \varliminf_{\gamma\to \infty}\frac{t_{{\rm mix}}(\cL_{\gamma},\epsilon_0)}{\gamma}
\leq 
\varlimsup_{\gamma\to \infty}\frac{t_{{\rm mix}}(\cL_{\gamma},\epsilon_0)}{\gamma} \leq  t_{{\rm mix}}(\cD_P^\sharp,\epsilon_0+\delta)\ 
$$
where we have used the bounds from Lemma~\ref{MIXRELSH}.
Now take the limit $\delta \to 0$ to conclude that 
${\displaystyle \lim_{\gamma\to\infty} \frac{t_{{\rm mix}}(\cL_{\gamma},\epsilon_0)}{\gamma}}$ exists and equals $t_{{\rm mix}}(\cD_P^\sharp,\epsilon_0)$. The statement for $\cL_{P,\gamma}$ is proved in the same way. 
\end{proof}

 \section{Application to  the computation of steady states}\label{STSTEXP}
 
 In this section, we make the additional assumption that $\cD_P^\sharp$ is ergodic and gapped so that it has a finite mixing time. Let $\bar R$ 
 denote the  unique steady state for $\cD_P^\sharp$.  By  part {\it (1)} of Lemma~\ref{DHARPLEM}, $\bar R = \bar R^\sharp$, or what is the same thing, $\cK_P \bar R =0$. 
Moreover,  $\cL_\gamma$ is also ergodic with a finite mixing time. Let $\bar \rho_\gamma$ denote 
  the  unique steady state for $\cL_\gamma$.

As we have seen,  for large enough $\gamma$,  $\bar\rho_\gamma \approx \pi_A \otimes \bar R$.  To get a better approximation to $\bar\rho_\gamma$, we  assume an expansion of the form
\begin{equation}\label{SSTEX0}
 \bar\rho_\gamma =  \pi_A\otimes \bar R + \frac{1}{\gamma}\bar n \quad{\rm where}\quad \bar n = \sum_{k=0}^\infty \gamma^{-k}\bar n_k\ .
\end{equation}
We shall show that is formal assumption leads to a system of equations for the $\bar n_k$, $k\geq 0$, that has a unique solution. Moreover, we then show with this choice of the $\bar n_k$, the sum in \eqref{SSTEX0} converges for sufficiently large $\gamma$, producing the unique steady state of 
$\cL_\gamma$. This finally proves that $\bar \rho_\gamma$ does indeed have the expansion in \eqref{SSTEX0}.

Inserting \eqref{SSTEX0}  into $\cL_\gamma \bar\rho_\gamma =0$ and equating like powers of $\gamma$ leads to the equations
\begin{equation}\label{SSTEX1}
\cD \bar n_0 = - \cK(\pi\otimes \bar R)
\end{equation}
and for $k\geq 1$,
\begin{equation}\label{SSTEX2}
\cD \bar n_k = - \cK(\bar n_{k-1})
\end{equation}

Throughout this section,  $\cX$ denotes the space of traceless operators on $\cH_B$,  
$\cV$ denotes ${\rm ran}(\cK_P)$ and $\cW$ denotes ${\rm ker}(\cK_P)\cap \cX$.  Since $\cK_P$ is skew adjoint, its kernel and range are orthogonal complements in $\widehat{\cH}_B$, 
and evidently every operator in the range of $\cK_P$ is traceless. That is, $\cV \subset \cX$. It follows that $\cX$ is the orthogonal direct sum of $\cV$ and $\cW$:
 \begin{equation}\label{DHARPLEM1}
 \cX = \cV \oplus \cW \ ;
 \end{equation}
 Let $P_{\cW}$ denote the orthogonal projection in $\widehat{\cH}_B$ onto ${\rm ker}(\cK_P)$ and note that its restriction to $\cX$ is the orthogonal projection in $\cX$ onto $\cW$. By the von Neumann Ergodic Theorem,
  \begin{equation*}
 P_{\cW}  = \lim_{T\to\infty}\frac{1}{2T}\int_{-T}^T e^{t\cK_P}{\rm d} t\ . 
 \end{equation*}
 It follows that $P_{\cW}$ is a CPTP operation, and hence $\|P_{\cW}\|_{1\to 1} =1$.

In solving the equations \eqref{SSTEX1} and \eqref{SSTEX2}, the following lemma will be essential.

 \begin{lem}\label{DHARPLEME} 
 
  \smallskip
 
 \noindent{\it (1)} For all $X\in {\rm ker} (\cK_P)$, $X\in {\rm ker}(\cD_P^\sharp)$ if and only if $\cD_P(X) \in {\rm ran}(\cK_P)$.
 
  \smallskip
 
 \noindent{\it (2)} for all $V'\in \cV$, there is a unique $V\in \cV$ such that $\cK_P V = V'$, and 
 $\|V\|_1 \leq b^{-1}\|V'\|_1$
 where $b>0$ is the spectral gap of $|\cK_P|$; see \eqref{Kgap}. 
 
   \smallskip
 
 \noindent{\it (3)}  
 Under the further hypothesis  that  $\cD_P^\sharp$ is ergodic,  for every $X\in \cX$ there exist uniquely determined $V\in \cV$ and $W\in \cW$ such that
 \begin{equation}\label{DHARPLEM2}
 X =  \cK_P(V) + \cD_P(W)\ .
 \end{equation}
 Moreover, there are finite  constants $C_\cV$ and $C_\cW$ such that 
$\|V\|_1 \leq C_\cV \|X\|_1$ $\|W\|_1 \leq C_\cW \|X\|_1$.
 \end{lem}
 
 \begin{proof} Suppose that $\cK_P(X) =0$. Then 
 ${\displaystyle
 \cT^\sharp(X) =  \lim_{T\to\infty}\frac{1}{2T}\int_{-T}^T e^{t\cK_P}\cT(  X)  {\rm d}t}$.
 Since $\cK$ is skew-adjoint,  $\cT(X)$ has a unique decomposition $\cT(X) = Y+Z$ where $Y\in {\rm ker}(\cK_P)$ and $Z\in {\rm ran}(\cK_P)$, so that $Z = \cK_P(W)$ for some operator $W$.  Then
 ${\displaystyle 
  e^{t\cK_P}\cT(  X)  = Y + \frac{{\rm d}}{{\rm d}t}e^{t\cK_P} W}$
 and consequently
 $$
 \cT^\sharp(X) = Y +\lim_{T\to\infty}\frac{1}{2T}(e^{T\cK_P}W - e^{-T\cK_P}W) = Y\ .
 $$
 Therefore,  $\cT^\sharp(X) =0$ if and only if $Y=0$, which is the same as $\cT(X) \in {\rm ran}(\cK_P)$. This proves {\it (1)}. 
 
  For part {\it (2)}, since $\cK_P$ is skew-adjoint, it is invertible on its range, and hence there is a unique $V\in \cV$ such that $\cK_P V = V'$. By the Spectral Theorem, $\|V\|_1 \leq b^{-1}\|V'\|$.
   
  For part {\it (3)},
we first show that when $\cD_P^\sharp$ is ergodic, $\cD_P(\cW)$ has the same dimension as $\cW$.  
 Suppose that $W\in \cW$ satisfies $\cD_P(W)\in \cV$.  Then by part {\it (1)}, $\cD_P^\sharp(W) =0$. Since $\cD_P^\sharp$ is ergodic and $W$ is traceless, this implies that $W=0$.  Therefore, 
 $\cD_P(\cW) \cap \cV = 0$.
 
 By standard linear algebra,  and then \eqref{DHARPLEM1},
 \begin{eqnarray*}
 {\rm dim}(\cD_P(\cW) + \cV) &=&  {\rm dim}(\cD_P(\cW)) + {\rm dim}(\cV) - {\rm dim}(\cD_P(\cW) \cap \cV)\\
  &=& {\rm dim}(\cW) + {\rm dim}(\cV)  = {\rm dim}(\cX) \ .
 \end{eqnarray*}
 Since $\cD_P(\cW) + \cV \subseteq \cX$, and since ${\rm dim}(\cD_P(\cW) + \cV) = {\rm dim}(\cX)$,  so that every $X\in \cX$ has a decomposition $X = \cD_P(W) + V$. The decomposition is unique since
 $\cD_P(\cW)\cap \cV = 0$. This proves that every $X\in \cX$ has a unique decomposition $X = V' +  \cD_P(W)$ with $V'\in \cV$ and $W\in \cW$.  
 By part {\it (2)}, $V' = \cK_P(V)$ for some uniquely determined $V$.  At this point we have that \eqref{DHARPLEM2} is satisfied with uniquely determined $V$ and $W$. It remains to bound 
   $\|V\|_1$ and $\|W\|_1$ in terms of $\|X\|_1$. 
   
Applying $P_\cW$ to both sides of $X = \cK_P V +\cD_P W$
 yields $P_\cW X =  P_\cW \cD_P W$. Since $\cD_P(\cW)\cap \cV = 0$, ${\rm ker}\left(P_\cW \cD_P|_{\cW}\right) =0$,   the operator sending $W'\in \cW $ to $P_\cW \cD_P W'\in \cW$ is invertible. 
 Let $C_W$ denote the super-operator trace norm of its inverse. Then $\|W\|_1 \leq C_W \|P_\cW X\|_1 \leq C_W \|X\|_1$, with the final inequality coming from $\|P_\cW\|_{1 \to 1} =1$.  
 
 Let $P_\cV$ be the orthogonal projection onto $\cV$. Since $P_\cV = \one - P_\cW$, $\|P_\cV\|_{1\to 1} \leq 2$.  Applying $P_\cV$ to both sides of $X = \cK_P V +\cD_P W$ yields
 $P_\cV X = \cP_\cV \cD_P W + V$, and hence
 $$
 \|V\|_1 \leq \| P_\cV X\|_1 + \|\cP_\cV \cD_P W\|_1 \leq 2\|X\|_1 + 2\|\cD_P\|_{1\to 1}\|W\|_1\ .
 $$
 Then defining $C_\cV := 2(1+ C_\cW \|\cD_P\|_{1\to1})$, $\|V\|_1 \leq C_\cV\|X\|_1$. 
 \end{proof}

\begin{lem}\label{FIRSTIT} Under the assumption that $\cD_P^\sharp$ is ergodic, the equation $\cD \bar n_0 = - \cK(\pi\otimes \bar R)$ for $\bar n_0$ is solvable.  The general solution
 such that the equation $\cD \bar n_1 = -\cK \bar n_0$ is also solvable  has the form 
\begin{equation}\label{FIRSTIT1}
\bar n_0 = - \cS \cK(\pi\otimes \bar R) + \pi_A\otimes V_0 + \pi_A\otimes W_0\ 
\end{equation}
for a uniquely determined choice of  $V_0\in \cV$ and arbitrary  $W_0\in \cW$.    For this choice of $V_0$,  define
\begin{equation}\label{FIRSTIT1B}
\tilde{m}_0 := - \cS \cK(\pi\otimes \bar R) + \pi_A\otimes V_0 
\end{equation}
so that the general solution of $\cD \bar n_0 = - \cK(\pi\otimes \bar R)$ for which the  $\cD n_1 = -\cK n_0$ is solvable is
\begin{equation}\label{FIRSTIT1C}
\bar n_0 = \tilde{m}_0  + \pi_A\otimes W_0 
\end{equation}
for arbitrary $W_0\in \cW$. Then
\begin{equation}\label{FIRSTIT2}
\|\tilde{m}_0\|_1 \leq  \|\cS \cK\|_{1\to1} +  \frac{\|\cD_P\bar R\|_{1}}{b}\
\end{equation}
 $b$ is the spectral gap of $|\cK_P|$. 
\end{lem}

\begin{proof} Since the range of $\cD$ is precisely the space of operators $Z\in \widehat{\cH}_{AB}$ such that $\tr_A{Z} =0$,   for \eqref{SSTEX1} to be solvable, it is necessary and sufficient that $\tr_A[\cK(\pi\otimes \bar R)] =0$, and since
$\tr_A[\cK(\pi\otimes \bar R)] = \cK_P \bar R$, the condition for solvability of \eqref{SSTEX1} is that $\cK_P \bar R =0$. As noted above, when $\cD_P^\sharp$ is ergodic, this is a consequence of 
part {\it (1)} of Lemma~\ref{DHARPLEM}.  Then $- \cS \cK(\pi\otimes \bar R)$ is a particular solution of  \eqref{SSTEX1}, and the general solution is
\begin{equation}\label{SSTEX1B1}
\bar n_0 = - \cS \cK(\pi\otimes \bar R) + \pi_A\otimes X_0\ .
\end{equation}
Since $\tr [\bar n_0] =0$ and $\tr[\cS \cK(\pi\otimes \bar R)]= 0$, it must be that $\tr[X_0] =0$, so that $X_0\in \cX$.  By \eqref{DHARPLEM1}, $X_0$ has a 
unique decomposition $X_0 = V_0 + W_0$ where $V_0\in \cV$ and $W_0\in \cW$.

The equation \eqref{SSTEX2} is solvable for $k=1$ if and only if $\tr_A[\cK \bar n_0] =0$.  By \eqref{SSTEX1B1}
\begin{equation*}
\tr_A[\cK \bar n_0] = \tr_A[- \cK \cS \cK(\pi\otimes \bar R)] + \tr_A[\cK \pi_A\otimes X_0]  = \cD_P \bar R + \cK_P X_0  = \cD_P \bar R + \cK_P V_0\ . 
\end{equation*}
Therefore, there is no constraint at this point on $W_0$, but $V_0$ must be chosen so that $\cK_P V_0 = -\cD_P \bar R$,
and this proves the first claim concerning  the formula for $\bar n_0$ in \eqref{FIRSTIT1}.   

Then with $\tilde{m}_0$ defined by \eqref{FIRSTIT1B}, \eqref{FIRSTIT1} becomes \eqref{FIRSTIT1C}.
By part {\it (1)} of Lemma~\ref{DHARPLEME}, since $\bar R \in {\rm ker}(\cK_P)\cap {\rm ker}(\cD_P^\sharp)$, $\cD_P(\bar R)\in {\rm ran}(\cK_P)$, and 
then by part {\it (2)} of Lemma~\ref{DHARPLEME}, there is a unique $V_0\in \cV$ such that $\cK V_0 = -\cD_P \bar R$, and $\|V_0\|_1 \leq b^{-1}\|\cD_P \bar R\|_1$, from which \eqref{FIRSTIT2}
follows. \end{proof}

Now let $\bar n_0 = \tilde{m}_0 + \pi_A\otimes W_0$  as in \eqref{FIRSTIT1C} of Lemma~\ref{FIRSTIT}, so that  $\cD \bar n_1 = -\cK \bar n_0$ is solvable.  The general solution has the form  
$$
\bar n_1 = -\cS \cK \bar n_0 + \pi_A\otimes V_1 + \pi_A\otimes W_1 =  -\cS \cK \tilde{m}_0 - \cS \cK   \pi_A\otimes W_0 + \pi_A\otimes V_1 + \pi_A\otimes W_1\ .
$$
for arbitrary $V_1\in \cV$ and $W_0,W_1\in \cW$.  The solvability condition for $\cD \bar n_2 = -\cK \bar n_1$ becomes 
\begin{equation}\label{SSTEX1D2}
0 = \tr_A[\cK \bar n_1] = - \tr_A[\cK \cS \cK \tilde{m}_0] + \cD_P W_0 + \cK_P V_1\ ;
\end{equation}
the choice of $W_1$ does not enter here.   

We do not know much about  $\tr_A[\cK \cS \cK \tilde{m}_0]$ except that it is traceless, but this is all we need. 
Under the assumption that $\cD_P^\sharp$ is ergodic,  by part {\it (3)} of Lemma~\ref{DHARPLEME}, $\tr_A[\cK \cS \cK \tilde{m}_0]$
 can be written in the form $\cD_P W + \cK_PV$ for some unique choice of $W$ and $V$.  Hence there is a unique choice of $W_0$ and $V_1$ such that \eqref{SSTEX1D2} is satisfied. Moreover,
 by part {\it (3)} of Lemma~\ref{DHARPLEME},  since $\|  \tr_A[\cK \cS \cK \tilde{m}_0]\|_1 \leq \|\cK\cS \cK\|_{1\to 1}\|\tilde{m}_0\|_1$, 
\begin{equation}\label{SSTEX1D2R}
\|V_1\|_1 \leq C_\cV \|\cK\cS \cK\|_{1\to 1}\|\tilde{m}_0\|_1 \quad{\rm and}\quad   \|W_0\|_1 \leq  C_\cW \|\cK\cS \cK\|_{1\to 1}\|\tilde{m}_0\|_1 \ . 
 \end{equation}
 Define
 ${
 \tilde{m}_1 := -\cS \cK \bar n_0 + \pi_A\otimes V_1  = -\cS \cK \tilde{m}_0 + \pi_A \otimes W_0 + \pi_A\otimes V_1}$.
 For any choice of $W_1\in \cW$, $\cD \bar n_2 = -\cK(\tilde{m}_1 + W_1)$ is solvable, and 
\begin{equation}\label{SSTEX1D2RX}
 \|\tilde{m}_1\|_1 \leq ( \|\cS\cK\|_{1\to 1} + (C_\cV + C_\cW)\|\cK\cS\cK\|_{1\to 1})\|\tilde{m}_0\|_1\ .
 \end{equation}
where we have used    $\|\tr_A[\cK \cS \cK \tilde{m}_0]\|_1   \leq  \|\cK\cS\cK\|_{1\to 1}\|\tilde{m}_0\|_1$. Altogether, 
 $$
 \|n_0\|_1 \leq \|\tilde{m}_0\|_1 + \|W_0\|_1 \leq  (1 + C_\cW  \|\cK\cS\cK\|_{1\to 1}) \|\tilde{m}_0\|_1 
 $$
 where $ \|\tilde{m}_0\|_1 $ is bounded by \eqref{FIRSTIT2}. 
 From here the pattern repeats, and we may make a simple induction. 

\begin{thm}\label{GOODSOLS} Under the assumption the $\cD_P^\sharp$ is ergodic, there  is a uniquely determined sequence  $\{ \bar n_k\}_{k\geq 0}$  in $\widehat{\cH}_B$ such that the equations \eqref{SSTEX1} and \eqref{SSTEX2}
are all satisfied, and moreover for all $k\geq 0$,
\begin{equation}\label{GOODSOLS1} 
\|\bar  n_k\|_1 \leq 2( \|\cS\cK\|_{1\to 1} + (C_\cV + C_\cW)^k \left(\|\cS \cK\|_{1\to1} +  \frac{\|\cD_P\bar R\|_{1}}{b}\right)\ .
\end{equation}
where $b$ is the spectral gap of $|\cK_P|$ and the constants $C_\cV$ and $C_\cW$ are defined in Lemma~\ref{DHARPLEME}.

Consequently for $\gamma> ( \|\cS\cK\|_{1\to 1} + (C_\cV + C_\cW)$,  ${\displaystyle  \sum_{k=0}^\infty \gamma^{-k}n_k}$ converges in the trace norm, and hence
\begin{equation}\label{GOODSOLS2} 
\bar \rho_\gamma :=\pi_A\otimes  \bar R + \gamma^{-1}  \sum_{k=0}^\infty \gamma^{-k}\bar n_k 
\end{equation}
satisfies $\cD \bar\rho_\gamma =0$, and is the unique steady state for $\cD$. 
\end{thm}

\begin{proof}
Let $\bar n_0$ be the unique solution of $\cD n_0 = -\cK(\pi_A\otimes \bar R)$, such that $\cD \bar n_1 = - \cK \bar n_0$  is solvable, and has solutions such that
$\cD \bar n_2 = -\cK \bar n_1$ is solvable. 
Fix an integer $k\geq 1$. Suppose that for all $1 \leq \ell\leq k$, the system of equations
$\cD \bar n_\ell = - \cK \bar n_{\ell-1} \quad 1\leq \ell \leq k$,
is solvable, and for $\ell < k$ the solution is unique and has the form
\begin{equation*}
\bar n_\ell = -\cS  \cK \bar n_{\ell-1} + \pi_A\otimes V_{\ell} + \pi_A\otimes W_{\ell}\ .
\end{equation*}
for  uniquely determined $V_\ell\in \cV$ and  $W_\ell\in \cW$,
while for $\ell = k$, the general solution such that $\cD \bar n = - \cK \bar n_k$ is solvable has the same form for a uniquely determined $V_k\in \cV$ and arbitrary $W_k\in \cW$. 

For $1\leq \ell \leq k$, define
$\tilde{m}_{\ell} := -\cS \cK n_{\ell-1} + \pi_A\otimes V_{\ell}$,
so that $\bar n_\ell = \tilde{m}_\ell + \pi_A\otimes W_\ell$, and let $\tilde{m}_0$ be given as in Lemma~\ref{FIRSTIT}. Suppose that for all $1\leq \ell\leq k$, 
\begin{equation}\label{SSTEX1D3B1}
\|\tilde{m}_\ell \|_1 \leq ( \|\cS\cK\|_{1\to 1} + (C_\cV + C_\cW)\|\cK\cS\cK\|_{1\to 1})\|\tilde{m}_{\ell -1}\|_1
\end{equation} 
and  
\begin{equation}\label{SSTEX1D3B2}
\|W_\ell\|_1 \leq C_\cW \|\cK\cS\cK\|_{1\to 1}\|\tilde{m}_{\ell -1}\|_1\ .
\end{equation} 

All of this has been proved for $k=1$; see \eqref{SSTEX1D2R} and \eqref{SSTEX1D2RX} and the paragraphs leading up to them, and we already know by the inductive hypotheses that 
$\cD \bar n = -\cK \bar n_k = - \cK(\tilde{m}_k + \pi_A\otimes W_k)$  is solvable for all $W_k \in\cW$. It follows that the general solution has the form
$$ 
\bar n_{k+1} = -\cS \cK  \tilde{m}_k -\cS \cK \pi_A\otimes W_k +   \pi_A\otimes V_{k+1} +  \pi_A\otimes W_{k+1}\
$$
for arbitrary $V_{k+1}\in \cV$ and $W_{k},W_{k+1}\in \cW$.  The condition for solvability of $\cD \bar n = -\cK \bar n_{k+1}$ then is 
$$
-\tr_A[\cK \cS \cK  \tilde{m}_k] + \cD_P W_k + \cK_P V_{k+1}\ .
$$
Since $\tr_A[\cK \cS \cK  \tilde{n}_k]\in \cX$,  by part {\it (3)} of Lemma~\ref{DHARPLEME}, there is a unique choice of $W_k$ and $V_{k+1}$ that solves this equation and moreover, since 
$\|\tr_A[\cK \cS \cK  \tilde{m}_k]\|_1 \leq \|\cK \cS \cK\|_{1\to 1}\|\tilde{m}_k\|_1$, $\|W_{k-1}\|_1 \leq C_\cW \|\cK \cS \cK\|_{1\to 1}\|\tilde{n}_k\|_1$, $\|W_{k-1}\|_1$ and $V_{k+1} \leq  
C_\cV \|\cK \cS \cK\|_{1\to 1}\|\tilde{m}_k\|_1$, $\|W_{k-1}\|_1$. Then with 
$\tilde{m}_{k+1} = -\cS \cK  \tilde{m}_k -\cS \cK \pi_A\otimes W_k +   \pi_A\otimes V_{k+1}$,  \eqref{SSTEX1D3B1} and \eqref{SSTEX1D3B2} are satisfied also for $\ell = k+1$. 
This completes the proof of the inductive step and hence the inductive hypotheses are valid for all $k\geq 1$. 

Now using the validity of \eqref{SSTEX1D3B1} and \eqref{SSTEX1D3B2} for all $1\leq \ell \leq k$,
$$
\|\tilde{m}_k\|_1 \leq  ( \|\cS\cK\|_{1\to 1} + (C_\cV + C_\cW)^k\|\tilde{m}_0\|_1
$$
and 
$$
\|W_{k}\|_1 \leq   \|\cS\cK\|_{1\to 1} ( \|\cS\cK\|_{1\to 1} + (C_\cV + C_\cW)^{k-1}\|\tilde{m}_0\|_1\ .
$$
Then since $\|\bar n_k\|_1 \leq \|\tilde{m}_k\|_1 + \|W_{k}\|_1$,
$$
\|\bar n_k\|_1 \leq  ( \|\cS\cK\|_{1\to 1} + (C_\cV + C_\cW)^k 2 \|\tilde{m}_0\|_1\ ,
$$
and using   the bound \eqref{FIRSTIT2} on $\|\tilde{n}_0\|_1$ yields \eqref{GOODSOLS1}, form which  the convergence of the sum in  \eqref{GOODSOLS2} readily follows. Then the sum in \eqref{GOODSOLS2}  is self-adjoint, traceless and satisfies 
$\cL_\gamma \bar\rho_\gamma =0$. Since $t_{{\rm mix}}(\cD_P^\sharp) < \infty$ implies that $t_{{\rm mix}}(\cL_\gamma) < \infty$ for sufficiently large $\gamma$, and then it follows from  Lemma~\ref{POSLEM} that this sum is a density matrix in the kernel of $\cL_\gamma$, and hence it can only be $\bar\rho_\gamma$. This proves \eqref{GOODSOLS2}. 
\end{proof}

\begin{thm}\label{KAEXTHM} Suppose that $\cD_P^\sharp$ is ergodic and gapped so that for large $\gamma$, the unique steady state $\bar\rho_\gamma$ of $\cL_\gamma$ has a convergent expansion of the form \eqref{SSTEX0}.
Then $\tr_B \bar n_0 =0$, so that 
$$\tr_B[\bar\rho_\gamma] = \pi_A + {\mathcal O}(\gamma^{-2})\ ,$$
if and only if $[\pi_A, K_A] = 0$ where
\begin{equation}\label{KAEXDEF}
K_A := H_A + \tr_B[ (\one_A\otimes \bar R)H_{AB}]
\end{equation}
where $\bar R$ is the unique steady state of $\cD_P^\sharp$. 
\end{thm}

\begin{proof} By \eqref{SSTEX1B1}, $\bar n_0 = - \cS \cK(\pi_A\otimes \bar R) + \pi_A\otimes X_0$ where $\tr X_0 = 0$, and hence $\tr_B \bar n_0  = \cS  \cK(\pi_A\otimes \bar R)$. Since $\cS = \cS_A\otimes \one$,  $\tr_B \cS = \cS_A\tr_B$.  By \eqref{JFGAPPED1}, for all $X\in \widehat{\cH}_A$, $\cQ_A X = X$ if and only if $\tr X = 0$, and hence $\cS_A$ is one-to-one on traceless operators on $\cH_A$. Since $\tr_B[ \cK(\pi_A\otimes \bar R)]$ is a traceless operator on $\cH_A$, 
 $\tr_B\bar n_0 =0$ if and only if $\tr_B[ \cK(\pi_A\otimes \bar R)] =0$.  Exactly as in the proof of Lemma~\ref{CommLem},
\begin{eqnarray*}
i\tr_B [ \cK(\pi_A\otimes \bar R)] &=& \tr_B[ H (\pi_A\otimes \bar R)] -  \tr_B[  (\pi_A\otimes \bar R)H] \\
&=& \tr_B[ H (\one_A\otimes \bar R)(\pi_A\otimes \one_B)] -  \tr_B[  (\one_A \otimes \bar R)H (\pi_A\otimes \one_B)] \\
&=& \tr_B[ H (\one_A\otimes \bar R)]\pi_A - \pi_A\tr_B[H (\pi_A\otimes \one_B)]
\end{eqnarray*}
where we have used cyclicity of the trace in the second equality. By cyclicity of the partial trace $\tr_B[H (\one_A\otimes \bar R)] = \tr_B[ (\one_A\otimes \bar R)H]$, and this operator is self-adjoint. 
This shows that $\tr_B[\bar n_0] =0$ if and only if $\tr_B[ (\one_A\otimes \bar R)H]$ commutes with $\pi_A$. Then since   $(\one_A\otimes \bar R)H = H_A\otimes \bar R + H_{AB}\otimes \bar R + \one_A\otimes H_B\bar R$, 
$$
\tr_B[ (\one_A\otimes \bar R)H]  = K_A + \tr[ H_B\bar R] \one_B\ ,
$$
and hence $\tr_B[ (\one_A\otimes \bar R)H]$ commutes with $\pi_A$ if and only if $K_A$ commutes with $\pi_A$. 
\end{proof}

\begin{example} We continue with the model introduced Example~\ref{EXAMP1} in the particular case in which $H_P = \sigma_2$. 
We have seen in Example~\ref{EXAMP1} that  independent of the choice of $H_B$, $\cD_P$ is given by \eqref{EX1DP}, which is the same as $\cD_A$ except that it acts on $\cH_B$. 

From \eqref{DASPECEX1} and \eqref{EX1DP}, it follows that 
$\cD_P(\pi_A) =0$, and 
     \begin{equation*}
   \cD_P(\sigma_1) =  -\sigma_1\ , \quad\cD_P \sigma_2 = -\sigma_2\quad  {\rm and}\quad \cD_P(\sigma_3) = -2(\sigma_3)\ ,
   \end{equation*}
    with $t$ denoting $\tanh(\beta/2)$, $\cD_P(\one) = -\tfrac{t}{\gamma}\sigma_3$. It follows that the matrix representation of the 
    projected Lindblad generator   $\cK_P + \gamma^{-1}\cD_P$  in the $\{\one,\sigma_1,\sigma_2,\sigma_3\}$ basis is
  $$
 \left[\begin{array}{cccc} 0 & 0 & 0 & 0\\ 0 & -\frac1\gamma & 0 &2\\ 0 & 0 & -\frac1\gamma & 0\\ -\frac{t}{\gamma} & -2 & 0 & -\frac2\gamma
 \end{array}\right]\ . 
 $$
 For all $\gamma > \frac14$, the eigenvalues of $\cK_P + \frac1\gamma \cD_P$ are
 \begin{equation}\label{PROJSTEADYSTATE5}
0\ ,\quad -\frac1\gamma\quad{\rm and}\quad -\frac{3 \pm i \sqrt{16\gamma^2-1}}{2\gamma}\ .
    \end{equation}
 The eigenvectors are easily calculated and one finds that the steady state $\bar R_\gamma$ is 
 \begin{equation}\label{PROJSTEADYSTATE}
\bar R_\gamma =  \frac12\one -\frac{\gamma t}{4\gamma^2+2}\sigma_1 - \frac{ t}{8\gamma^2+4}\sigma_3 ~, 
 \end{equation}
 which has the expansion
 \begin{equation}\label{PROJSTEADYSTATE2A}
 \frac12\one - \gamma^{-1}\frac{t}{4}\sigma_1 - \gamma^{-2} \frac{t}{8}\sigma_3  + \mathcal{O}(\gamma^{-3})\ .
 \end{equation}

 In Example~\ref{EXAMP1} we have also computed the explicit form of $\cD_P^\sharp$ which is given by \eqref{DPSHARP}. We can also write \eqref{DPSHARP2} as
  \begin{equation}\label{DPSHARP2B}
 \cD_P^\sharp(\one) = 0\ ,\quad \cD_P^\sharp(\sigma_1) = -\sigma_1\ ,\quad  \cD_P^\sharp(\sigma_3) = -\sigma_3 \quad{\rm and}\quad \cD_P^\sharp(\sigma_2) = -\frac32\sigma_2\ .
 \end{equation}
 Therefore, the unique steady state for the Lindblad generator $\cD_P^\sharp$, $\bar R$, is given by
 \begin{equation}\label{DPSHARP30}
 \bar R = \frac12 \one\ ,
 \end{equation}
 and this is the unique steady state for the Lindblad generator $\cD_P^\sharp$, independent of $\beta$. In particular $\cD_P^\sharp$ is ergodic.  Therefore, in this example, the ergodicity assumption of Theorem~~\ref{GOODSOLS} is satisfied. 

Comparing \eqref{PROJSTEADYSTATE5} and \eqref{DPSHARP2B},  note that the real parts of the non-zero eigenvalues of $\gamma \cL_{P,\gamma}$ are $-1$ and $-\tfrac32$, which are exactly the non-zero eigenvalues of $\cD_P^\sharp$. This is consistent with, but does not imply, the relation between the mixing times in Theorem~\ref{TMIXZL3}. Comparing \eqref{PROJSTEADYSTATE} and \eqref{DPSHARP30}, note that $\bar R = \lim_{\gamma\to\infty}\Bar R_\gamma$. 
 
 This example is simple enough that one can exactly compute the unique steady sate of $\cL_\gamma$. In the basis consisting of tensor products of operators in $\{\one, \sigma_1,\sigma_2,\sigma_3\}$ its coefficients are rational functions of $\gamma$ with a common denominator that is $2\gamma^4 + 23 \gamma^2 + 64$, and in which the numerators are polynomials in $\gamma$ of degree between $1$ and $4$. While all $16$ coefficients can be computed exactly, for our purposes it is best to give the expansion of the steady state in inverse powers of $\gamma$. 
 With $t$ denoting $\tanh(\beta/2)$, 
  \begin{eqnarray}\label{SSEXEXACT1A}
  \bar\rho_\gamma &=& \pi_A \otimes \tfrac12\one\nonumber \\
 &+& \gamma^{-1}\left( -\frac{t}{4}(\one\otimes \sigma_1 +\sigma_1\otimes \sigma_2 - \sigma_2\otimes \sigma_1) +\frac{t^2}{4}\sigma_3\otimes \sigma_1 \right)\nonumber\\
 &+& \gamma^{-2}\left( \frac{t}{8}\sigma_1\otimes\one -  \frac{t}{8}\one\otimes \sigma_1-\frac{3t}{2}\one\otimes \sigma_2 - \frac{t^2}{4}\sigma_2\otimes \one  + \frac{t}{8}\one\otimes \sigma_3 + \frac{t}{4}\sigma_3\otimes \one
 \right)\nonumber\\
 &+& \gamma^{-2}\left( -\frac{t}{2}\sigma_1\otimes \sigma_1 -\frac{t}{2}\sigma_2\otimes \sigma_2 -\frac{3t^2}{8}\sigma_3\otimes \sigma_3 -\frac{t}{4}\sigma_2\otimes\sigma_3 + \frac{3t^2}{2}\sigma_3\otimes \sigma_2)
 \right)+ \mathcal{O}(\gamma^{-3})\ .
  \end{eqnarray}
  In particular,
   \begin{equation}\label{PROJSTEADYSTATE3A}
  \tr_B[\bar\rho_\gamma ] = \pi_A  +\gamma^{-2}\left(\frac{t}{4}\sigma_1 -\frac{t^2}{2}\sigma_2 + \frac{t}{2}\sigma_3\right)+ \mathcal{O}(\gamma^{-3})
 \end{equation}
  and
 \begin{equation}\label{PROJSTEADYSTATE4A}
 \tr_A[\bar\rho_\gamma ]  = \frac12 \one  -\gamma^{-1}\frac{t}{2}\sigma_1 +\gamma^{-2}\left( -\frac{t}{4}\sigma_1 -3t\sigma_2 +\frac{t}{4}\sigma_3 \right) + \mathcal{O}(\gamma^{-3})\ .
   \end{equation}
   Note that \eqref{PROJSTEADYSTATE4A} agrees with  \eqref{PROJSTEADYSTATE2A} at first order in $\gamma^{-1}$, but not beyond. The result \eqref{PROJSTEADYSTATE3A} is consistent with Theorem~\ref{KAEXTHM}, as it must be: Since in this example $\bar R$ is a multiple of the identity, and $\tr_B H_{AB} =0$ by definition,
   the operator $K_A$ defined in \eqref{KAEXDEF} is simply $H_A = \sigma_z$, which commutes with $\pi_A$, so that, by Theorem~\ref{KAEXTHM}, $\tr_B[\bar\rho_\gamma] = \pi_A + {\mathcal O}(\gamma^{-2})$.

   Such a direct computation of $\bar\rho_\gamma$ is not feasible in all but the simplest systems. We now show how to recover \eqref{SSEXEXACT1A} 
   using the Hilbert expansion for steady states that is justified by Theorem~\ref{GOODSOLS}.
   
 Again with $t$ denoting $\tanh(\beta/2)$, one  computes
${\displaystyle
\cK\left(\pi_A\otimes \tfrac12 \one\right) = -\frac{t}{4}(\sigma_1\otimes \sigma_2 - \sigma_2\otimes\sigma_1)}$.
  Therefore \eqref{SSTEX1} becomes
  \begin{equation}\label{HEXSS1}
  \cD \bar n_0 = \frac{t}{4}(\sigma_1\otimes \sigma_2 - \sigma_2\otimes\sigma_1) \ .
  \end{equation}
  In accordance with Theorem~\ref{GOODSOLS}, the partial trace over $\cH_A$ of the right hand side of \eqref{HEXSS1} is zero, and hence \eqref{HEXSS1} is solvable, and the general solution is
  $$
  \bar n_0 = \frac{t}{4}(\sigma_1\otimes \sigma_2 - \sigma_2\otimes\sigma_1) + \pi_A\otimes X_0
  $$
  for some traceless and self-adjoint $X_0$.

  We then compute 
    $\tr_A[ \cK \bar n_0] =-t\sigma_3  -i[\sigma_2,X_0] $.  Decompose $X_0 = Y_0+ Z_0$ where $Y_0\in {\rm ran}(\cK_P)$ and $Z_0\in {\rm ker}(\cK_P)$.  
    In this example, the range of $\cK_P$ is  ${\rm span}(\{\sigma_1,\sigma_3\})$ while the kernel is ${\rm span}(\{\one,\sigma_2\})$. Hence the spaces $\cW$ and $\cV$ in the proof of Lemma~\ref{DHARPLEM}
    are $\cV = {\rm span}(\{\sigma_1,\sigma_3\})$ and  $\cW = {\rm span}(\{\sigma_2\})$ and  the unique unique choice for $Y_0\in \cV$ 
     is $Y_0 = -\tfrac{t}{2}\sigma_1$, and hence
    \begin{eqnarray*}
    \bar  n_0 &=& -\frac{t}{4}(\sigma_1\otimes \sigma_2 - \sigma_2\otimes\sigma_1) -\frac{t}{2} \pi_A\otimes \sigma_1 + \pi_A\otimes Z_0\\
    &=& -\frac{t}{4}(\sigma_1\otimes \sigma_2 - \sigma_2\otimes\sigma_1) -\frac{t}{4}\one\otimes \sigma_1 + \frac{t^2}{4}\sigma_3\otimes\sigma_1+ \pi_A\otimes Z_0\\
   &=:& \tilde{m}_0 + \pi_A\otimes Z_0\ .
    \end{eqnarray*}

    To go on to compute the next order corrections,  and to determine $Z_0$, we must solve $\cD \bar  n_1 = -\cK \bar  n_0$ subject to $\tr_A[\cK \bar n_1] =0$.  We compute
    \begin{eqnarray*}
    \cK \tilde{m}_0=   \frac{t}{2}( -\sigma_1\otimes \sigma_1 -\sigma_2\otimes \sigma_2 +  \sigma_3\otimes\one) - \frac{t}{4}\sigma_2\otimes \sigma_3- \frac{t^2}{2}\left(\sigma_3\otimes \sigma_3 + 
    \frac12\sigma_2\otimes \one\right)
    \end{eqnarray*}
    Then $\tr_A[\cK \bar n_0] = 0$, and $\cD \bar n_1 = -\cK \bar n_0$ is solvable, and the general solutions is
    $$
    \bar n_1 = \cS \cK \tilde{m}_0 + \cS \cK(\pi_A \otimes Z_0) + \pi_A \otimes X_1
    $$
    for some self-adjoint and traceless $X_1$ that we decompose as $X_1 = Y_1+Z_1$ with $U_1\in \cV$ and $Z_1\in \cW$. 
    Then the condition $\tr_A[\cK \bar n_1] =0$ becomes
    $$
    \tr_A[\cK \cS \cK \tilde{m}_0] + \tr_A[\cK \cS \cK (\pi_A\otimes Z_0)] + \tr_A[\cK \pi_A\otimes Y_1] = \tr_A[\cK \cS \cK \tilde{m}_0]  - \cD_P Z_0 + \cK_P Y_1\ .
    $$

    \begin{eqnarray*}
    \bar n_1 &=& -\cS\cK \bar n_0 +\pi_A\otimes X_1 \\
    &=& \frac{t}{2}( -\sigma_1\otimes \sigma_1 -\sigma_2\otimes \sigma_2 +  \frac12\sigma_3\otimes\one) - \frac{t}{4}\sigma_2\otimes \sigma_3- \frac{t^2}{4}\left(\sigma_3\otimes \sigma_3 + \sigma_2\otimes \one\right)
    + \pi_A\otimes X_1\ .
    \end{eqnarray*}
    We compute
    $$
    -\cK\cS\cK \tilde{m}_0 = \frac{5t}{4}\sigma_1\otimes \sigma_2 - \frac{7t}{4} \sigma_2\otimes\sigma_1- \frac{3t^2}{2}\sigma_3\otimes \sigma_1  - \frac{3t}{2}\sigma_1\otimes \sigma_3 -\frac{t}{4}\one\otimes \sigma_1 + \frac{t^2}{2}\sigma_1\otimes \one\ .
    $$
    Therefore, $\tr_A[\cK\cS\cK \tilde{m_0}] \in \cV$, and we take $Z_0= 0$ and $Y_1 = t\sigma_3$.  At this point, $\bar n_0$ is completely determined, and the result is 
    $$
    \bar n_0 = -\frac{t}{4}(\sigma_1\otimes \sigma_2 - \sigma_2\otimes\sigma_1) -\frac{t}{4}\one\otimes \sigma_1 + \frac{t^2}{4}\sigma_3\otimes\sigma_1
    $$ 
    which is consistent with \eqref{SSEXEXACT1A}. Continuing the expansion, in the next stage we  find a non-zero value for $Z_1$, and the result for $\bar n_1$  is consistent with \eqref{SSEXEXACT1A}, as it must be. 
\end{example}

\section{Hydrodynamic limits}

 An interesting perspective on the  relation between the evolution equations involving the generators $\cL_\gamma$, $\cK_P$, $\cL_{P,\gamma}$ and $\cD_P^\sharp$ is provided by an analogy with the theory of hydrodynamic limits in classical  kinetic theory. The Boltzmann 
equation for the phase space density $f(v,x,t)$ of a dilute gas has the form 
\begin{equation}\label{Boltzmann}
\frac{\partial}{\partial t} f(v,x,t) = -v\cdot \nabla_x f(v,x,t) + \gamma \cC f(v,x,t) 
\end{equation}
where the first term on the right describes the effects of the free motion of molecules between collisions, while $\cC$ is the  collision kernel, an operator describing the effects of binary collisions between  molecules. The parameter $\gamma$ is then the inverse of the {\em Knudsen number} which gives the ratio of the microscopic times scale (the mean time between collisions) and the macroscopic time scale (the time for a molecule to travel a macroscopic distance). 

In the analogy with the boundary driven quantum systems discussed here, our starting point
\begin{equation}\label{BoltzmannA}
 \frac{{\rm d}}{{\rm d}t}\rho(t) = -i[H,\rho(t)] + \gamma \cD \rho(t)
 \end{equation}
 corresponds to the Boltzmann equation. The space $\cH_A$ accounts for the degrees of freedom corresponding to the the velocity variables in \eqref{Boltzmann}, while 
 $\cH_B$ accounts for the degrees of freedom corresponding to  the spatial variables. The parameter $\gamma$ corresponds to the inverse Knudsen number, and $\cD$ corresponds to the collision operator.

When the Knudsen number is small, and hence $\gamma$ is large, this formal similarity between \eqref{Boltzmann} and \eqref{BoltzmannA}  is brought out more clearly by noting that the phase space densities densities $f(v,x)$ satisfying $\cC f(v,x) =0$ are precisely the {\em local Maxwellian densities}
${\displaystyle 
\frac{\rho(x)}{(2\pi T(x))^{3/2}}e^{- |v - u(x)|^2/2T(x)}}$
where $\rho(x)$, $u(x)$ and $T(x)$ are the {\em hydrodynamic moments}  of $f(v,x)$.  The analog of $\cM$ as in \eqref{SSM} is then the space of all local Maxwellian densities, and when $\gamma$ is large, one expects  solutions of \eqref{Boltzmann}  to quickly become, and then remain, approximately locally Maxwellian. To describe the further evolution after this initial layer, one need only keep track of the evolution of the hydrodynamic moments, which correspond physically to the density, bulk velocity and temperature. Hydrodynamic equations for the evolution of these quantities, such as the Euler equations or the Navier-Stokes equations, had been derived before the Boltzmann equation, not starting from molecular dynamics, but in the framework of continuum mechanics. 

The problem of deriving equations for the hydrodynamic moments from molecular dynamics  was discussed by Hilbert in his follow-up \cite{DH01} to the famous list of problems that he proposed at the International Congress of Mathematicians in 1900.  His original Sixth Problem was, briefly put, to axiomatize mechanics and probability.  In \cite{DH01} he expounded more, and highlighted the specific problem of deriving the hydrodynamic equations of continuous mechanics from the kinetic equations of Maxwell and Boltzmann,  and deducing expression for, e.g., the viscosity, in terms of  the microscopic interaction of molecules. Later, he took up this problem himself  and  proposed an expansion method \cite{DH12} that yields a hierarchy of systems of  equations for the hydrodynamic moments.  The expansion leading to this hierarchy is known as the Hilbert expansion. 

The first system of equations  in the hierarchy is the Euler equations, and it does not contain $\gamma$. The Euler equations are {\em formally} reversible and conservative, and their analog here is the leading order coherent approximation 
${\displaystyle \frac{{\rm d}}{{\rm d}t} R(t) = \cK_P R(t)}$. 

The remaining systems of equation in the hierarchy do contain $\gamma$.  The first of these corresponds to the equation
${\displaystyle \frac{{\rm d}}{{\rm d}t} R(t) = \cL_{P,\gamma} R(t)}$.  This next term in Hilbert's hierarchy is not the Navier-Stokes equations themselves, but rather a precursor to them. The 
Navier-Stokes equations do not contain the Knudsen number, $1/\gamma$.  To obtain the Navier-Stokes equations, one rescales space and time depending on $\gamma$ so that all terms in  
the precursor equations are multiplied by a common power of $\gamma$ that can then be cancelled out, leaving a negative power of $\gamma$ in the remainder terms. One then takes $\gamma\to \infty$. For a more complete discussion, see \cite{BGL91,BGL93}.  

In our quantum setting, one can rescale time, but there is no really direct analog of spatial scaling. To cancel $\gamma$ in this setting, we must rescale time and make a suitable unitary transformation which, 
as in \cite{Da74}, is the passage to the interaction picture. This leads to  ${\displaystyle \frac{{\rm d}}{{\rm d}t} R(t) = \cD_P^\sharp R(t)}$. This equation, not involving $\gamma$, may be regarded as the counterpart  of the  
Navier-Stokes equations in this analogy. 

It is unclear whether some analogous  procedure, possibly involving a restriction on the initial data, may be used to extract an analog of the Burnett equations from the next equation in our quantum hierarchy, \eqref{INTR51}. 
However, given that  \eqref{INTR51} may not describe a CPTP evolution, and given the problematic nature of the Burnett equations, there may be little point in pursuing this.

\appendix

 \section{Bounding $\|e^{t\cD}\cQ\|_{1\to 1}$ in terms of  $\|e^{t\cD_A}\cQ_A\|_{1\to 1}$} 
 
In Appendix~\ref{GENINVBND} we shall  use the spectral gap of $\cD_A$ to bound the rate of decay of $\|e^{t\cD_A}\cQ_A\|_{1\to1} = \|e^{t\cD_A} - |\pi_A\rangle\langle\one_A|\|_{1\to 1}$. However,  we need  control on the rate at which  $\|e^{t\cD}\cQ\|_{1\to 1} = \|e^{t\cD}-\cP\|_{1\to 1}$ tends to zero. Since 
 $e^{t\cD} = e^{t\cD_A}\otimes \one_B$ and $\cP = |\pi_A\rangle\langle\one_A| \otimes \one_B$, this might seem automatic, but it is not.   If $\cT_1$ and $\cT_2$ are two CPTP super-operators on $\cH_A$, then it can be the case that
 $$
 \|(\cT_1-\cT_2)\otimes \one_B\|_{1\to 1} > \|\cT_1-\cT_2\|_{1\to 1}\ .
 $$
 This has been investigated by Kitaev \cite{K97,KSV02}  who introduced the {\em diamond norm} on Hermitian super-operators; ie.e, those for which $\cT(X^*) = \cT(X)^*$ for all $X$. By a well-known variant of Choi's Theorem, every Hermitian  super-operator 
 $\cT$  can be written as the difference of two CPTP maps; see, e.g. \cite{W05}.  Kitaev \cite{K97} proved that while it can be the case that 
 $\|(\cT_1-\cT_2)\otimes \one_B\|_{1\to 1}$ increases with the dimension of $\cH_B$, there is no further increase after the dimension of $\cH_B$ has reached the dimension of $\cH_A$. 
 The diamond norm of $\cT_1-\cT_2$ is then defined by 
 $$
 \|\cT_1-\cT_2\|_\diamond := \|(\cT_1-\cT_2)\otimes \one_B\|_{1\to 1} \quad{\rm where}\quad {\rm dim}(\cH_B) = {\rm dim}(\cH_A)\ .
 $$
 However, this leaves open the question of how much larger $\|\cT_1-\cT_2\|_\diamond$ might be than $\|\cT_1-\cT_2\|_{1\to 1}$. In our finite dimensional setting, all norms are equivalent, 
 so hence if $\lim_{t\to\infty}\|e^{t\cD}\cQ\|_{1\to 1} =0$, then $\lim_{t\to\infty}\|e^{t\cD}\cQ\|_{\diamond} =0$, and evidently $\|e^{t\cD}\cQ\|_{1\to 1} \leq \|e^{t\cD}\cQ\|_{\diamond}$. 
 However, there do not seem to be good general  bounds that can be used to compare the rates.

The following lemma will however allow us to transfer a certain bound on $\|e^{t\cD_A}\cQ_A\|_{1\to1}$ directly to $\|e^{t\cD}\cQ\|_{1\to1}$, and this bound of  the latter  is independent of the dimension of $\cH_B$. 
 
 \begin{lem}\label{SOPOPNOLEM} Let $\cH,\cK$ be finite dimensional Hilbert spaces. Let $X,Y\in \widehat{\cH}$, and define $|Y\rangle\langle X|$ as in \eqref{VCTZ}. Then the super-operator $|Y\rangle\langle X|\otimes \one$ on $\cH\otimes \cK$ satisfies
 \begin{equation*}
 \| |Y\rangle\langle X|\otimes \one\|_{1\to 1} =  \| |Y\rangle\langle X|\|_{1\to 1} = \|Y\|_1\|X\|_\infty\ .
 \end{equation*}
 In other words, for rank-one super-operators $\| |Y\rangle\langle X|$, $\| |Y\rangle\langle X|\|_\diamond =  \| |Y\rangle\langle X|\|_{1\to 1}$.
 \end{lem} 
 
 \begin{proof} Let $Z$ be any operator on $\cH\otimes \cK$ with $\|Z\|_1 =1$. Then $(|Y\rangle\langle X|\otimes \one)Z= Y\otimes \tr_{\cH}[(X^\dagger \otimes \one) Z]$,
 and hence $\|(|Y\rangle\langle X|\otimes \one)Z\|_1 = \|Y\|_1\|\tr_\cH[(X^\dagger\otimes \one) Z]\|_1 \leq  \|Y\|_1\|\|(X^\dagger \otimes \one) Z\|_1$ since the partial trace is a CPTP map and hence a trace norm contraction. By H\"older's inequality for operators, $\|(X^\dagger \otimes \one) Z\|_1 \leq \|X^\dagger \otimes \one\|_\infty \|Z\|_1$, and $\|X^\dagger \otimes \one\|_\infty = \|X\|_\infty$. Altogether,
 $$\|(|Y\rangle\langle X|\otimes \one)Z\|_1 \leq \|Y\|_1\|X\|_\infty =  \| |Y\rangle\langle X|\|_{1\to 1}\ .$$
 This proves that  $\| |Y\rangle\langle X|\otimes \one\|_{1\to 1} \leq  \| |Y\rangle\langle X|\|_{1\to 1}$ and the opposite inequality is trivial.
 \end{proof}

There is a close connection between the issue discussed in this appendix with the norm multiplicitivity problem, and we conclude by briefly discussing the connection.
 A generalization of the super-operator trace norm plays an important role in quantum information theory because of its connections with quantum entropy.
  For $1 \leq p <\infty$ the Schatten $p$-norm on $\widehat{\cH}$ is defined by $\|X\|_p = \left(\tr[(X^*X)^{p/2}]\right)^{1/p}$.  
 For $p=2$ this is simply the Hilbert space norm associated to the Hilbert-Schmidt inner product. For $p=1$, this is the  trace norm. For all $X$, $\|X\|_p$ is 
 monotone non-increasing, and $\lim_{p\to\infty}\|X\|_p = \|X\|_\infty$ where $\|X\|_\infty$ is 
the operator norm of $X$, which is equal to its largest singular value of $X$. 
 
 Let $\cH$  be a finite dimensional Hilbert space, and let $\cT$ be a linear operator on $\widehat{\cH}$.  For $1 \leq p \leq \infty$, define
 $$
 \|\cT\|_{1\to p} = \sup\{ \|\cT(X)\|_p\ : \|X\|_1 = 1\ \}\ .
$$

 Now let $\cH_A$ and $\cH_B$ be two finite dimensional Hilbert spaces,
 and let $\cT_A$ and $\cT_B$ be super-operators on $\cH_A$ and $\cH_B$ respectively. 
 For five years it was an open  conjecture \cite{AHW00} that if both
 $\cT_A$ and $\cT_B$ are CPTP, then for $1 \leq p\leq  \infty$,
 $\|\cT_A\otimes \cT_B\|_{1\to p} = 
 \|\cT_A\|_{1\to p}\|\cT_B\|_{1\to p}$.  It was shown  
 \cite{S04} that this conjecture,
 at least for $p$ sufficiently close to $1$, would have a wide range
 of applications. While the conjecture was eventually
 proved to be false
 for all $p>1$ in \cite{HW08}, there is a positive result
 in the orignal paper \cite{AHW00}, where it is proved
 that if $\cT_B= \one_B$, multiplicitivity does hold for $\cT_A$
 CPTP: 
 \begin{equation}\label{AHWTHM}
\|\cT_A\otimes \one_B\|_{1\to p} = \|\cT_A\|_{1\to p}\ .
 \end{equation}

While $e^{t\cD_A}$ is CPTP, $e^{t\cD_A}\cQ_A$ is not; it is instead  the diference between two CPTP maps, $e^{t\cD_A}$ and 
$\cP_A$. Thus we {\em cannot} apply the bound \eqref{AHWTHM}
to conclude that $\|e^{t\cD}\cQ\|_{1\to1} = \|e^{t\cD_A}\cQ_A\|_{1\to1}$, and based on  Kitaev's results, such an identity is unlikely to be true in general.  However, 
Lemma~\ref{SOPOPNOLEM} will give us what we need in Appendix~\ref{GENINVBND}.

 \section{Dual bases and Jordan bases}\label{DBJBA}

In this section $\cH$ is any Hilbert space of finite dimension $n$ with inner product $\langle \cdot,\cdot\rangle$.  We begin with some elementary facts about dual Jordan bases for which we do not know a convenient reference. 

\begin{lem} Let  $\{\bx_1,\dots,\bx_n\}$ and $\{\by_1,\dots,\by_n\}$ be two sets of vectors in $\cH$ such that for all $j,k$, 
\begin{equation}\label{DualBas1}
\langle \bx_j,\by_k\rangle = \delta_{j,k}\ .
\end{equation}
Then both $\{\bx_1,\dots,\bx_n\}$ and $\{\by_1,\dots,\by_n\}$ are linearly independent, and hence are bases, not necessarily orthogonal, of $\cH$.
\end{lem} 

\begin{proof} Suppose that 
${\displaystyle \sum_{j=1}^n \beta_j\bx_j = 0}$. Then $0 = \left \langle \sum_{j=1}^n \beta_j\bx_j, \by_k\right\rangle = \sum_{j=1}^n  \overline{\beta_j} \delta_{j,k} = \overline{\beta_k}$. Since  this is true for each $k$,
$\{\bx_1,\dots,\bx_n\}$ is linearly independent. Since the dimension of $\cH$ is $n$, $\{\bx_1,\dots,\bx_n\}$ spans $\cH$, and hence is a basis for $\cH$. By symmetry, the same is true of $\{\by_1,\dots,\by_n\}$.
\end{proof}

\begin{lem} Let $\{\by_1,\dots,\by_n\}$ be  any basis of $\cH$, not necessarily orthonormal. Then there exists a uniquely determined set of vectors $\{\bx_1,\dots,\bx_n\}$ such that \eqref{DualBas1} is satisfied. 
\end{lem} 

\begin{proof} Every ${\boldsymbol \psi}\in \cH$
has an expansion
${\displaystyle 
{\boldsymbol \psi} = \sum_{j=1}^n \alpha_j \by_j}$  with uniquely determined coefficients. For each $j=1,\dots,n$, define a linear functional $\Lambda_j$ on $\cH$
by 
$\displaystyle{\Lambda_j\left(\sum_{j=1}^n \alpha_j \by_j \right) = \alpha_j}$.
By the Riesz Representation Theorem, there is a uniquely determined $\bx_j\in \cH$ such that 
\begin{equation}\label{DualBas3}
\langle \bx_j,{\boldsymbol \psi}\rangle = \Lambda_j({\boldsymbol \psi})\ .
\end{equation}
 Taking ${\boldsymbol \psi} = \by_k$ yields  \eqref{DualBas1}. Conversely, given any set $\{\bx_1,\dots,\bx_n\}$ such that 
 \eqref{DualBas1} is satisfied for all $j,k$, then \eqref{DualBas3} is satisfied where $\Lambda_j$ is given by \eqref{DualBas3}. Hence the uniqueness of $\{\bx_1,\dots,\bx_n\}$ follows from the uniqueness in the Reisz Representation Theorem. 
\end{proof}

\begin{defn} Let $\{\by_1,\dots,\by_n\}$ be  any basis of $\cH$. Then the unique basis  $\{\bx_1,\dots,\bx_n\}$  of $\cH$ such that  \eqref{DualBas1} is satisfied for all $j,k$ is the 
{\em dual basis} to  $\{\by_1,\dots,\by_n\}$
\end{defn}

It is clear that $\{\bx_1,\dots,\bx_n\}$ is the dual basis to $\{\by_1,\dots,\by_n\}$ if and only if $\{\by_1,\dots,\by_n\}$ is the dual basis to $\{\bx_1,\dots,\bx_n\}$.

Let $A$ be any operator on $\cH$. Then $A$ may or may not be diagonalizable, but there always exists a direct sum decomposition of $\cH = \bigoplus_{\ell =1}^m \cV_\ell$ into subspaces (not in general mutually orthogonal) that are invariant under $A$, with each $\cV_\ell$ contained in a generalized eigenspace of $A$ for some eigenvalue $\lambda_\ell$ of $A$, and finally such that within $\cV_\ell$ the eigenspace corresponding to $\lambda_\ell$ is one dimensional. Each of these spaces $\cV_\ell$ corresponds to a block in a Jordan normal form matrix representation of $A$. There are two of these normal forms, upper and lower, which we need to distinguish. Let $d_\ell$ denote the dimension of $\cV_\ell$. A basis $\{\by_1,\dots,\by_{d_\ell}\}$ for $\cV_\ell$ is an {\em upper Jordan block basis}  in case
$$
A\by_1 = \lambda_\ell \by_1  \quad{\rm and\ for}\quad j=2,\dots,d_\ell\ , A\by_j = \lambda_\ell\by_j + \by_{j-1}\ .
$$
For  $d_\ell =3$, the matrix representing the action of $A$ on $\cV_\ell$ in an upper Jordan block basis is
$$
\left[\begin{array}{ccc} \lambda & 1 & 0\\ 0 & \lambda & 1\\ 0 & 0 & \lambda\end{array}\right]\ .
$$
A basis $\{\bx_1,\dots,\bx_{d_\ell}\}$ for $\cV_\ell$ is a {\em lower Jordan block basis} in case
$$
A\bx_{d_\ell} = \lambda_\ell \bx_{d_\ell}  \quad{\rm and\ for}\quad j=1,\dots,d_\ell-1\ , A\bx_j = \lambda_\ell\bx_j + \bx_{j+1}\ .
$$
For  $d_\ell =3$, the matrix representing the action of $A$ on $\cV_\ell$ in a lower Jordan block basis is
$$
\left[\begin{array}{ccc} \lambda & 0 & 0\\ 1 & \lambda & 0\\ 0 & 1 & \lambda\end{array}\right]\ .
$$
The distinction between upper an lower Jordan block bases is minor; reversing the order of the basis elements changes from on to the other. However, the distinction maters when considering dual bases.  Concatenating bases for each $\cV_\ell$ yields a basis for $\cH$. A basis for $\cH$ is an {\em upper Jordan basis} in  case it is the concatenation of upper Jordan block bases for each $\cV_\ell$, and is a 
{\em lower Jordan basis} in  case it is the concatenation of lower Jordan block bases for each $\cV_\ell$.

\begin{thm}\label{DUBAJF2} Let $\{\by_1,\dots,\by_n\}$ be an upper Jordan basis for an operator $A$ on $\cH$. Then the dual basis $\{\bx_1,\dots,\bx_n\}$ is a lower Jordan basis for $A^\dagger$. Furthermore, if for some $j$, $A\by_j = \lambda \by_j$, and $\lambda$ is an eigenvalue of $A$ of algebraic multiplicity one, then $A^\dagger \bx_j =\overline{\lambda}\bx_j$. 
\end{thm}

\begin{proof} Consider first the case $m=1$ so that there is only one Jordan block. Then $A$ has only one eigenvalue $\lambda$ and for each $1 \leq j,k \leq n$, 
$\langle \bx_k,A\by_j\rangle = \lambda \delta_{j,k} + \delta_{j-1,k}$. Therefore, 
$$
\langle \by_j,A^\dagger \bx_k\rangle = \overline{\langle \bx_k,A\by_j\rangle} =  \overline{\lambda} \delta_{j,k} + \delta_{j-1,k}  =  \overline{\lambda} \delta_{j,k} + \delta_{j,k+1}\ 
$$
so that 
$A^\dagger \bx_n =\overline{\lambda}  \bx_n$ and for $k=1,\dots,n-1$, $A^\dagger \bx_k = \overline{\lambda}  \bx_k + \bx_{k+1}$.  The general case follows easily by concatenation.

For the final part, $\lambda$ is an eigenvalue of $A$ of algebraic multiplicity one if and only if the generalized eigenspace of $A$ for  $\lambda$ coincides with the eigenspace for $\lambda$, and the latter is one dimensional. 
Therefore, $A\by_j = \lambda \by_j$ where $\lambda$ has algebraic multiplicity one if and only if $\{\by_j\}$ is an upper and lower  Jordan block basis for the single Jordan invariant subspace of $A$  with eigenvalue $\lambda$. It then 
follows from the first part that $\{\bx_j\}$ a Jordan upper and lower basis for the single Jordan invariant subspace of $A^\dagger$  with eigenvalue $\overline{\lambda}$. In particular, $A^\dagger \bx_j = \overline{\lambda}\bx_j$.
\end{proof}

\section{Generalized inverses of Lindblad generators}\label{GENINVBND}

We apply the results of the previous subsection  to the Hilbert space $\widehat{\cH}_A$ consisting of operators on $\cH_A$.  Let $n_A+1$ denote the dimension of $\widehat{\cH}_A$. 
Let $\cD_A$ be a Lindblad generator acting on operators on $\cH_A$. Because $\cD_A$ generates a semigroup that is contractive in the trace norm, all eigenvalues of $\cD_A$ have a non-positive real part, and for any that are purely imaginary, the algebraic and geometric multiplicities coincide.  
  As always, we suppose that $\cD_A$ is ergodic and gapped with  unique steady state $\pi_A$. 

Let $\{Y_0,\dots, Y_{n_A}\}$ be a Jordan basis of $\cH$ for $\cD_A$ in which $Y_0 = \pi_A$.  Let 
$\{X_0,\dots,X_{n_A}\}$ be the dual basis. By Theorem~\ref{DUBAJF2}, $\cD_A^\dagger X_0 = 0$.  Since $\cD_A^\dagger $ generates a semigroup of completely positive unital operators, $\cD_A^\dagger\one_A =0$,
and hence $\one_A$ spans the eigenspace of $\cD_A$ corresponding for the eigenvalue zero.  Since $\langle \one_A, \pi_A\rangle = \tr[\pi_A] = 1$, $X_0 = \one_A$. 

Since $\cD_A$ is  gapped, so  that for some $a > 0$,  all non-zero eigenvalues $\lambda$ of $\cD_A$ satisfy
$\Re(\lambda) \leq -a$,
\begin{equation*}
e^{t\cD_A} = \sum_{j=0}^{n_A}e^{t\lambda_j}|Y_j\rangle \langle X_j|\ ,
\end{equation*} and hence 
\begin{equation}\label{JFGAPPED1}
\cP_A := \lim_{t\to\infty} e^{t\cD_A} = |Y_0\rangle\langle X_0| =  |\pi_A\rangle\langle \one_A|\ .
\end{equation}
Then $\cD_AY_0 =0$ and $\cD^\dagger X_0 =0$, $\cD_A\cP_A =0 = \cP_A\cD_A$. In particular, $[\cD_A,\cP_A] =0$ and hence 
 the complementary projection $\cQ_A = \one_{\widehat{\cH}} - \cP$ also commutes with $\cD_A$.   Evidently, $\cQ_A = \sum_{j=1}^{n_A} |Y_j\rangle\langle X_j|$.  

Since $\cP_A$ and $\cQ_A$ are complementary projections, so are $\cP_A^\dagger$ and $\cQ_A^\dagger$. 
Evidently $\cP_A^\dagger = |\one_A\rangle\langle \pi_A|$ and $\cQ_A^\dagger = \sum_{j=1}^{n_A}|X_j\rangle\langle Y_j|$. 
Since both $\cP_A$ and $\cQ_A$ commute with $\cD_A$, it follows that  both $\cP_A^\dagger$ and $\cQ_A^\dagger$ commute with $\cD_A^\dagger$. 

An operator $Z$ on $\cH_A$  belongs to ${\rm ran}(\cQ_A)$,  the range of $\cQ_A$, if and only if $\cP_A Z =0$, which is the case if and only if $\tr[Z]=\langle\one_A,Z\rangle = 0$. 
Therefore, ${\rm ran}(\cQ_A)$ is the subspace of $\widehat{\cH}_A$ consisting  operators whose trace is zero. Likewise, $Z$ belongs to ${\rm ran}(\cQ_A^\dagger)$ if and only if 
$\tr[\pi_A Z] = \langle \pi_A,Z\rangle = 0$, so that ${\rm ran}(\cQ_A^\dagger)$ 
is the subspace of $\widehat{\cH}_A$ consisting operators whose expected value in the state $\pi_A$ is zero.

Recall that the generalized inverse of $\cD_A$ is the super-operator  $\cS_A$ given by
\begin{equation}\label{JFGAPPED2}
\cS_A = -\int_0^\infty e^{t\cD_A} \cQ{\rm d}t  = -\int_0^\infty \left(\sum_{j=1}^{n_A} |e^{t\cD_A}Y_j\rangle  \langle X_j| \right){\rm d}t\ .
\end{equation}
Then since $\cD_A$ commutes with $\cQ_A$, 
\begin{equation}\label{JFGAPPED3}
\cD_A \cS_A = \cS_A \cD_A = \cQ_A  \ .
\end{equation}

Taking adjoints in \eqref{JFGAPPED3} yields $\cD_A^\dagger \cS_A^\dagger =  \cS_A^\dagger\cD_A^\dagger = \cQ_A^\dagger$ so that $\cS_A^\dagger$ is a generalized inverse to $\cD_A^\dagger$.  
One gets a formula for $\cS_A^\dagger$ from  \eqref{JFGAPPED2} by taking adjoints, but there is another approach yields a useful alternate formula.
Since $\cQ_A$ and $e^{t\cD_A}$ commute, $(e^{t\cD_A}\cQ_A)^\dagger = (\cQ_A e^{t\cD_A})^\dagger = e^{t\cD_A^\dagger}\cQ_A^\dagger$. Therefore, since $\cD_A^\dagger$ has the same spectral gap as $\cD_A$, 
\begin{equation}\label{JFGAPPED21}
\cS_A^\dagger  := \int_0^\infty e^{t\cD_A^\dagger } \cQ^\dagger {\rm d}t  = \int_0^\infty \left(\sum_{j=1}^{n_A} |e^{t\cD_A^\dagger}X_j\rangle  \langle Y_j| \right){\rm d}t  \ .
\end{equation}

\begin{lem}\label{ALTS} Under our assumptions on $\cD_A$ and the dual bases $\{Y_0,\dots,Y_{n_A}\}$ and $\{X_0,\dots,X_{n_A}\}$, the generalized inverse $\cS$ of $\cD_A$ is satisfies
\begin{equation}\label{JFGAPPED4}
\cS_A  = \int_0^\infty \left(\sum_{j=1}^{n_A} |Y_j\rangle  \langle e^{t\cD_A^\dagger}X_j| \right){\rm d}t  \ .
\end{equation}
\end{lem}

\begin{proof} The formula \eqref{JFGAPPED4} follows directly from \eqref{JFGAPPED21} by taking adjoints. 
\end{proof}

\begin{thm}\label{GENINVTNBLM} Let $\cD_A$ be an ergodic Lindblad generator on $\cH_A$ with invariant state $\pi_A$ and a spectral gap $a>0$.  
Let $\cD = \cD_A\otimes \one_B$ be the corresponding Lindblad generator on $\cH_B$. Define the projection $\cP$ by  $\cP(X) := \pi_A\otimes \tr_A X$, and define the complementary projection $\cQ$ by $\cQ = \one - \cP$. Then there is a finite constant $C$ that depends only on $\cD_A$, and in particular is independent of the dimension of $\cH_B$ such that
\begin{equation}\label{GENINVTNBLM1}
\|e^{t\cD}\cQ\|_{1\to 1} \leq Ce^{-ta/2}\ .
\end{equation}
Consequently $\cS  := \int_0^\infty e^{t\cD_A}\cQ_A {\rm d}t$  is well-defined and
$\displaystyle{
\|\cS\|_{1\to 1} \leq \frac{2C}{a} }$.
\end{thm}

\begin{proof}Let $\{Y_0,\dots, Y_{n_A}\}$ be a Jordan basis of $\cH$ for $\cD_A$ in which $Y_0 = \pi_A$.  Let $\lambda_j$ be the eigenvalue of
$\cD_A$ such that $Y_j$ is a generalized eigenvector of $\cD_A$ with this eigenvalue.  Let 
$\{X_0,\dots,X_{n_A}\}$ be the dual basis, so that in particular, $X_0 = \one$.  
For each $j$, there is a natural number $d_j \leq n_A$ such that $(\cD_A -\lambda_j)^{d_j}Y_j =0$ but  $(\cD_A -\lambda_j)^{d_j-1}Y_j \neq 0$. Then
${\displaystyle e^{t\cD_A}Y_j =  e^{t\lambda}e^{t(\cD_A - \lambda)}Y_j = e^{t\lambda}\sum_{k=0}^{d_j-1} \frac{t^k}{k!}Y_{j+k}}$.  It follows that for some finite constant $C$, 
\begin{equation}\label{PROPBND}
\|e^{t\cD_A}Y_j \|_1\|X_j\|_\infty\leq \frac{C}{n_A} e^{-at/2}\ .
\end{equation} 
By the triangle inequality
\begin{equation}\label{PROPBND2}
\|e^{t\cD_A}\cQ_A\|_{1\to 1} \leq C^{-at/2}\ .
\end{equation}

Since
${\displaystyle e^{t\cD}\cQ = \sum_{j=1}^{n_a}  |e^{t\cD_A}Y_j \rangle\langle X_j| \otimes \one}$, it follows from \eqref{PROPBND}, 
the triangle inequality, and Lemma~\ref{SOPOPNOLEM}  that
$$
\|e^{t\cD}\cQ\|_{1\to 1} \leq  \sum_{j=1}^{n_a}  \|e^{t\cD_A}Y_j \|_1\|X_j\|_\infty\ .
$$
Now \eqref{GENINVTNBLM1} follows from this and \eqref{PROPBND}. The second part of the theorem is an immediate consequence of the first. 
\end{proof}

\begin{rem}\label{GAPMIX} It is easy to see that for all $0 < \epsilon < 1$, one may replace 
$Ce^{-ta/2}$ with $C_\epsilon e^{-t(1-\epsilon) a}$ in \eqref{PROPBND},
but then  the constant $C_\epsilon$ will diverge 
as $\epsilon \to 0$ if $\cD_A$ is not diagonalizable.  However, if
$\cD_A$ is diagonalizable, the factor of $\tfrac12$ is superfluous. 

In any case, by \eqref{JFGAPPED1}, $\cQ = \one - |\pi_A\rangle\langle \one|$, and then by \eqref{PROPBND2}, for all density matrices 
$\rho_0$ on $\cH_A$, 
\begin{equation}\label{PROPBND3}
\|e^{t\cD_A}\rho_0 - \pi_A\|_1 \leq Ce^{-ta/2}\ ,
\end{equation}
so that \eqref{PROPBND2} gives a uniform bound on the rate of approach to stationarity. Note however that since $\|e^{t\cD_A}\rho_0 - \pi_A\|_1\leq$ for all choices of $\rho_0$ and all $t$, if $C$ is large, the bound \eqref{PROPBND3} will be trivial for $t \leq \frac{2}{a}\log(2C)$. 
\end{rem}

\section{Volterra integral operators and the comparison of evolution equations}\label{VOLTERRA}

In what follows, we wish to compare various evolution mechanisms for density matrices and for super-operators. A powerful method for doing this involves writing the equations governing the evolution mechanisms as 
Volterra integral equations.  The elementary Theorem~\ref{VOLTIWTHM} stated below will be applied a number of times. We give the short proof for convenience.

In this section $\cX$ denote an arbitrary Banach space with norm $\|\cdot\|$. In the application made here his will always be the space of operators on a finite dimensional Hilbert space $\cH$ 
equipped with the trace norm $\|\cdot \|_1$, or else the space of super-operators on $\cH$ equipped with the super-operator trace norm $\|\cdot \|_{1\to1}$.     

Let $\cT>0$ be given and define
$\sC$  to be the space of continuous functions $A:[0,T]\to \cX$ equipped with the norm 
 ${\displaystyle |\!|\!| A|\!|\!| := \sup_{\tau\in [0,T]}\|A(\tau)\|}$.
 For an operator $\sH$ on $\sC$, define
 \begin{equation*}
 \|\sH\|_{\sC \to \sC} = \sup\{ |\!|\!| \sH A|\!|\!| \ :\ |\!|\!| A|\!|\!| \leq 1\}\ .
 \end{equation*}
 The operator $\sH$ is a {\em Volterra integral operator} on $\sC$ in case it has the form
\begin{equation}\label{VOLTINTOP}
\sH Y(\tau) = \int_0^\tau \cG(\tau,\sigma)Y(\sigma){\rm d}\sigma
 \end{equation}
 where $\cG(\tau,\sigma)$  is a continuous function of $0 \leq \sigma \leq \tau \leq T$ with values in the space of operators on $\widehat{\cX}$.
 Define $\cG(\tau,\sigma)\|_{\cX \to \cX} := \sup\{ \|\cG(\tau,\sigma)\|\ :\ \|X\| \leq 1\}$ and then
 \begin{equation}\label{VOLTINTOP2}
 G := \sup_{0 \leq \sigma \leq \tau \leq T}\|\cG(\tau,\sigma)\|_{\cX \to \cX}\ ,
  \end{equation}
  and then  $\|\sH\|_{\sC \to \sC}  \leq TG$.   
  
  We will work with perturbations of Volterra integral operators. 
  Let let  $\widetilde{\sH}$ be another bounded Volterra integral operator defined in terms of $\widetilde{\cG}(\tau,\sigma)$ as in 
 \eqref{VOLTINTOP}. Then $\sH - \widetilde{\sH}$ is a Volterra integral operator  and
 ${\displaystyle
 (\sH - \widetilde{\sH})Y(\tau) = \int_0^\tau (\cG(\tau,\sigma)- \widetilde{\cG}(\tau,\sigma))Y(\sigma){\rm d}\sigma}$.
 Therefore,
 \begin{equation*}
 \|\sH - \widetilde{\sH}\|_{\sC \to \sC} \leq \sup_{0 \leq \tau \leq T}\left\{ \ \int_0^\tau  \|\cG(\tau,\sigma)- \widetilde{\cG}(\tau,\sigma)\|_{\cX \to\cX}{\rm d}\sigma\ \right\}\ .
 \end{equation*}

 For given   $Z\in \sC$ and Volterra integral operator $\sH$, the equation
  \begin{equation}\label{DAVLIM312}
 Y(\tau) =  \sH Y(\tau) + Z(\tau)\ . 
 \end{equation}
 is a Volterra integral equation of the second kind, and it may be solved by iteration. The following theorem collects the elementary results we need. 
 
 \begin{thm}\label{VOLTIWTHM} For all bounded Volterra integral operators $\sH$ on $\sC$, and all $Z\in \sC$, the equation \eqref{DAVLIM312} has a unique solution $Y\in \sC$ which is given by
 $\displaystyle{
Y(\tau) :=  \sum_{n=0}^\infty \sH^n Z(\tau)}$
 where the sum converges in norm in $\sC$. Moreover, let  $\widetilde{\sH}$ be another bounded Volterra integral operator defined in terms of $\widetilde{\cG}(\tau,\sigma)$ as in 
 \eqref{VOLTINTOP}. Let $C:= \max\{G,\widetilde{G}\}$ where $G$ and $\widetilde{G}$ are given by  \eqref{VOLTINTOP2}. Let  $\widetilde{Z}$ be an element of $\sC$, 
 and define $\widetilde{Y}$ to be the solution of 
 \eqref{DAVLIM312} with $\widetilde{\sH}$ in place of $\sH$, and $\widetilde{Z}$ in place of $Z$. Then
  \begin{equation}\label{VOLTIWTHM2} 
  \|Y -\widetilde{Y}\|_\sC \leq \frac{1}{TC}(e^{2TC}-1)\|\sH - \widetilde{\sH}\|_{\sC \to \sC}\min\{\|Z\|_{\sC},\|\widetilde{Z}\|_{\sC}\}  + e^{2TC}\|Z - \widetilde{Z}\|_\sC\ .
  \end{equation}
 \end{thm}
 
 \begin{proof}
 Iteration of \eqref{DAVLIM312} yields, for each $N\in \N$, 
   \begin{equation}\label{VOLTINTOP5}
Y(\tau) = \sH ^{N+1} Y(\tau)  +  \sum_{n=0}^N \sH ^n Z(\tau)
  \end{equation} 

Note that
${\displaystyle \sH^2Y(\tau) = \int_0^\tau \int_0^{\sigma_2} \cG(\tau,\sigma_2)\cG(\sigma_2,\sigma_1)Y(\sigma_1){\rm d}\sigma_2{\rm d}\sigma_1}$,
and hence
${\displaystyle  \|\sH^2\|_{\sC \to \sC}  \leq \frac{T^2}{2}G^2}$. \hfill\break
Likewise, one finds that for all $n\in \N$, $\|\sH^n\|_{\sC \to \sC}  \leq \frac{T^n}{n!}G^n$. Likewise,
$\|\widetilde{\sH}^n\|_{\sC \to \sC}  \leq \frac{T^n}{n!}\widetilde{G}^n$. Therefore,
 \begin{equation*}
 \max\{\ \|\sH^n\|_{\sC \to \sC}\ ,\  \|\widetilde{\sH}^n\|_{\sC \to \sC} \ \}  \leq \frac{T^n}{n!}C^n\ ,
  \end{equation*}
and ${\displaystyle 
Y(\tau) :=  \sum_{n=0}^\infty \sH^n Z(\tau)}$
  is well-defined since the sum converges $\sC$. One readily checks that it satisfies \eqref{DAVLIM312}, and by \eqref{VOLTINTOP5} it is the unique solution.
  For the second part,
  $$
  Y(\tau)  -\widetilde{Y}(\tau) =  \sum_{n=1}^\infty (\sH^n -\widetilde{\sH}^n) Z(\tau) + \sum_{n=0}^\infty \widetilde{\sH}^n (Z(\tau) - \widetilde{Z}(\tau))\ .
  $$
  By the telescoping sum formula,
${\displaystyle
(\sH^{(\gamma)})^n - (\widetilde{\sH}^{(\gamma)})^n = \sum_{j=1}^{n-1} (\sH^{(\gamma)})^{n-1-j} (\sH^{(\gamma)} - \widetilde{\sH}^{(\gamma)}) (\widetilde{\sH}^{(\gamma)})^j}$
and therefore
\begin{eqnarray*}
\| (\sH^{(\gamma)})^n - (\widetilde{\sH}^{(\gamma)})^n\|_{\sC \to \sC} &\leq&
\sum_{j=1}^{n-1} \|\sH^{(\gamma)}\|_{\sC \to \sC} ^{n-1-j} \|\sH^{(\gamma)} - \widetilde{\sH}^{(\gamma)}\|_{\sC \to \sC}  \|\widetilde{\sH}^{(\gamma)}\|^j_{\sC \to \sC} \\
&\leq& \sum_{j=1}^{n-1}\frac{1}{j!(n-j)!}  \left(T C\right)^{n-1} \|\sH^{(\gamma)} - \widetilde{\sH}^{(\gamma)}\|_{\sC \to \sC} \\
&\leq& \frac{2^n}{n!}  \left(T C\right)^{n-1} \|\sH^{(\gamma)} - \widetilde{\sH}^{(\gamma)}\|_{\sC \to \sC}\ .
\end{eqnarray*}Therefore,
$$
\|Y -\widetilde{Y}\|_\sC \leq \frac{1}{TC}(e^{2TC}-1)\|\sH^{(\gamma)} - \widetilde{\sH}^{(\gamma)}\|_{\sC \to \sC}\|Z\|_{\sC}\|\widetilde{Z}  + e^{2TC}\|Z - \widetilde{Z}\|_\sC\ .
$$
By the symmetry in $Y$ and $\widetilde{Y}$, we have the same bound with $Z$ and $\widetilde{Z}$ swapped, and this proves \eqref{VOLTIWTHM2} .
 \end{proof}
 
 \section{Mixing times}\label{MTAPP}
 
 Let $\cL$ be a Lindblad generator acting on $\widehat{\cH}$, $\cH$ a finite dimensional Hilbert space.  By  \eqref{PROPBND3},  and the definitions \eqref{TVDDEF} and \eqref{MIXTIMEDEF}, if $\cL$ is ergodic and gapped with unique steady state $\pi$, $t_{{\rm mix}}(\cL,\epsilon) < \infty$ for all $0 < \epsilon < \tfrac12$.  The next lemma provides a converse.

\begin{lem}\label{POSLEM} Let $\cL$ be a Lindblad generator acting on $\widehat{\cH}$. Assume only that  $t_{{\rm mix}}(\cL,\epsilon) <\infty$ for some $0 <\epsilon< \tfrac12$.
Then $\cL$ has a unique steady state $\pi$, and  ${\rm ker}(\cL)$ is spanned by $\pi$. Moreover, 
$\cL$ has no purely imaginary eigenvalues so that in our finite dimensional setting  $\cL$ is ergodic and gapped.  In particular, if $\cL\rho =0$ and $\tr[\rho] =1$, then $\rho = \pi$. 
\end{lem} 

\begin{proof}   Every Lindblad generator acting on $\widehat{\cH}$ has at least one steady state, and since $t_{{\rm mix}}(\cL,\epsilon) <\infty$, there is a unique steady state $\pi$. 
Since  for all $X$,  $\cL(X^\dagger) = \cL(X)^\dagger$, if $X\in {\rm ker}(\cL)$ then $X^\dagger \in {\rm ker}(\cL)$. Hence if $\pi$ does not span ${\rm ker}(\cL)$ there would be a self-adjoint $X\in{\rm ker}(\cL)$
that is linearly independent of $\pi$.  Subtracting an appropriate multiple of $\pi$ from $X$ and normalizing, we may assume that $\tr[X] =0$ but $\|X\|_1 = 2$. Let $X= X_+ -X_-$ be the spectral decomposition of 
$X$ into its positive and negative parts. Then $X_+$ and $X_-$ are both density matrices. Since $\cL X =0$, for all $t>0$, 
$$
X = e^{t \cL} X =   e^{t \cL} X_+ - e^{t \cL} X_-\ .
$$
Therefore, $2 = \|X\|_1 = \| e^{t \cL} X_+ - e^{t \cL} X_-\|_1 = 2d_{{\rm TV}}(e^{t \cL} X_+ ,e^{t \cL}X_-)$. By \eqref{MIXTIMEDEF}, for $t> t_{{\rm mix}}(\cL,\epsilon)$, $d_{{\rm TV}}(e^{t \cL} X_+ ,e^{t \cL}X_-) < \epsilon$.  This contradiction proves that 
$\pi$  span ${\rm ker}(\cL)$.

Likewise, suppose that $X\in \widehat{\cH}$ satisfies $\cL X = i\omega X$ for some non-zero real $\omega$ and that $X \neq 0$. Since   $\cL(X^\dagger) = \cL(X)^\dagger$, $\cL X^\dagger  = -i\omega X^\dagger$.
Since $\tr[\cL X] =0$, $\tr[X]=0$.
Define $A = \frac12(X+ X^\dagger)$, and note that $A\neq 0$ since $\omega \neq 0$, but $\tr[A] =0$.  Let $t := 2\pi \omega$. Then $e^{t\cL A} = A$. Rescaling $A$, we may assume $\|A\|_1 =2$. Then  
$2 = \|A\|_1 = \|e^{kt\cL} A\|_1$ and exactly as in
the first part of the proof, $\lim_{k\to\infty}\|e^{kt\cL} A\|_1= 0$. This contradiction precludes the existence of any purely imaginary eigenvalues. 
\end{proof}

We next show that when $t_{{\rm mix}}(\cL,\epsilon) < \infty$, so that there is a unique stationary state $\pi$, $d_{{\rm TV}}(e^{kt_{{\rm mix}}(\cL)})$ decreases to zero exponentially fast in $k$. 
For $t>0$, let $P_t := e^{t\cL}$. In preparation for this, define the super-operator $\Pi$ by $\Pi X :=  \tr[X]\pi $. Then $\Pi$ is a CPTP map. It easy to write down an explicit Kraus representation showing this, but also 
 $\Pi = \lim_{t\to\infty}P_t$, and each $P_t$ is CPTP. A super-operator $\cT$ is said to be {\em Hermitian} in case $\cT(X)^\dagger = \cT(X^\dagger)$ for all 
 $X\in \widehat{\cH}$. Since all CPTP operations are Hermitian (see e.g. \cite[Lemma 5.2]{C25}), and hence $P_t - \Pi$ is Hermitian.

 In this appendix, let $\cX$ denote the space of {\em self-adjoint} operators on $\cH$ equipped  with the trace  norm, and let $\widehat{\cX}$  denote the space of {\em Hermitian} 
 super-operators on $\cH$ equipped  with the norm
 \begin{equation}\label{MIXTIMEAPP3}
 |\!|\!| \cT  |\!|\!| := \sup\{ \|\cT(X)\|_1\ :\ \|X\|_1 = 1 \quad{\rm and}\quad X = X^\dagger\}\ .
 \end{equation}
 This is closely related to the super-operator trace norm $\|\cdot\|_{1\to 1}$, but may be smaller because of the restriction that $X$ must be self-adjoint. 
 
 \begin{lem}\label{MTNORM} For all $t > 0$,
  \begin{equation}\label{MIXTIMEAPP4}
  |\!|\!| P_t - \Pi   |\!|\!| = \sup\{ \| P_t\rho -\pi\|_1 \ :\  \rho\ {\rm a \ density\ matrix\ on}\ \cH\}\ .
  \end{equation}
 \end{lem}
 
 \begin{proof} For any density matrix $\rho$, since $\| P_t\rho -\pi\|_1 = \|(P_i - \Pi)\rho\|_1$ and since  $\rho= \rho^\dagger$ and $\|\rho\|_1 =1$, $\|(P_t - \Pi)\rho\|_1 \leq |\!|\!| P_t - \Pi  |\!|\!|$.  This proves that the left side of 
 \eqref{MIXTIMEAPP4} is at least as large as the right side. 
 
 Next, in our finite dimensional setting the supremum in \eqref{MIXTIMEAPP3} is a maximum. Let $X$ be such that $X= X^\dagger$, $\|X\|_1 = 1$ and
 $ |\!|\!| P_t - \Pi   |\!|\!| =\|(P_t- \Pi)X\|_1$.  Let $X = X_+ - X_-$ be the decomposition of $X$ into its positive and negative parts, and define $0 \leq \lambda\leq 1$ by $\lambda := \tr[X_+]$
 Then $\mu := \lambda^{-1}X_+$ and $\nu := (1-\lambda)^{-1}X_-$ are density matrices and $X = \lambda \mu - (1-\lambda)\nu$. Therefore,
 $$
 \|(P_t- \Pi)X\|_1 \leq \lambda \|(P_t- \Pi)\mu\|_1 + (1-\lambda) \|(P_t- \Pi)\nu\|_1 \leq \sup\{ \| P_t\rho -\pi\|_1 \ :\  \rho\ {\rm a \ density\ matrix\ on}\ \cH\}\ .
 $$
 This proves that the right side of 
 \eqref{MIXTIMEAPP4} is at least as large as the left side. 
 \end{proof}

 \begin{thm}\label{TMIXK} Let $\rho_0,\rho_1$ be any two density matrices on $\cH$. For all $0 < \epsilon < \tfrac12$, and all $t \geq t_{{\rm mix}}(\cL,\epsilon)$, and all positive integers $k$, 
 \begin{equation}\label{TMIXK1}
 d_{{\rm TV}}(P_{kt}\rho_0,P_{kt}\rho_1) \leq 2(2\epsilon)^{k} \ .
 \end{equation}
 \end{thm}
 
  \begin{proof} Let $\rho_0$ and $\rho_1$  be  arbitrary density matrices on $\cH$. For $j=0,1$,  $P_t\rho_j - \pi = (P_t - \Pi)\rho_j$.   Then
 \begin{equation}\label{MIXTIMEAPP5}
 \|P_t \rho_0 - P_t\rho_1\|_1 \leq  \|(P_t -\Pi)\rho_0 \|_1 + \|(P_t -\Pi)\rho_1 \|_1\ .
  \end{equation}
 
 For each $t>0$, since $P_t$ is trace preserving, $\Pi P_t = \Pi$. Since $P_t\pi = \pi$, $P_t\Pi =\Pi$. That is, 
 $P_t \Pi = \Pi P_t = \Pi$. Therefore, 
 by the Binomial Theorem, for any positive integer $k$,
 \begin{equation}\label{MIXTIMEAPP6}
 (P_t - \Pi)^k = \sum_{\ell=0}^k(-1)^\ell \binom{k}{\ell}P_t^{k-\ell}\Pi^\ell = P_{kt} - \Pi +  \Pi \sum_{\ell=0}^k(-1)^\ell \binom{k}{\ell} = P_{kt} - \Pi \ .
 \end{equation}
 Then by \eqref{MIXTIMEAPP4}, \eqref{MIXTIMEAPP5}, \eqref{MIXTIMEAPP6} and the triangle inequality,
 \begin{equation}\label{MIXTIMEAPP7}
 \|P_{kt} \rho_0 - P_{kt}\rho_1\|_1 \leq  \|(P_t -\Pi)^k\rho_0 \|_1 + \|(P_t -\Pi)^k\rho_1 \|_1 \leq 2|\!|\!| P_t - \Pi   |\!|\!|^k \ .
  \end{equation}
 
 Now suppose that $t > t_{{\rm mix}}(\cL,\epsilon)$.  By the definition \eqref{MIXTIMEDEF} of $t_{{\rm mix}}(\cL,\epsilon)$, for any density matrix $\rho$,  
 $$  \epsilon \geq d_{{\rm TV}}(P_t \rho, \pi) = \tfrac12 \|P_t\rho  - P_t\pi\|_1 = \tfrac12 \|P_t\rho  - \pi\|_1 \ .$$
Therefore by Lemma~\ref{MTNORM}, $|\!|\!| P_t - \Pi   |\!|\!| \leq 2\epsilon$. Now \eqref{TMIXK1} follows from \eqref{MIXTIMEAPP7}.
\end{proof}

 The $\epsilon$-mixing time of $\cL$ is often defined to be
 \begin{equation}\label{GAPMIX2}
\min\{ t>0 \ :\ \|P_t\rho_0 -\pi\|_1 < \epsilon \quad {\rm for\ all\ density\ matrices}\ \rho_0\ \}\ .
 \end{equation}
 By \eqref{PROPBND3} in Remark~\ref{GAPMIX}, 
 the quantity in \eqref{GAPMIX2} is finite whenver $\cL$ is ergodic and gapped. 
 Then, since for any two density matrices $\rho_0$, $\rho_1$,
 $$
\|P_t\rho_0 - P_t\rho_1\|_1 \leq  \|P_t\rho_0 - \pi_A\|_1 +
\|P_t\rho_1 - \pi_A\|_1
 $$
 $t_{{\rm mix}}(\cL,\epsilon) < \infty$ whenever $\cL$ is ergodic and gapped. 
 Moreover, by Theorem~\ref{MTNORM} the definition of mixing time used
 here differs from the one defined in terms of \eqref{GAPMIX2} by at
 most a factor of $2$. Our applications require a formulation
 that does not refer to a putative unique invariant measure. Many variants of the mixing time are discussed in the classical literature on Markov chains, and the classical analog of the one considered here occurs in the foundational paper \cite{A83} of Aldous

\section*{Ackowledgements}
E.C. presented a preliminary version of this work at the IPAM workshop {em Dynamics of Density Operators}, April 28- May 2, 2025, and thanks the participants, especially Di Fang and Marius Junge for stimulating discussions, and thanks IPAM for hospitality and a stimulating, productive environment.  D.A.H. was supported in part by NSF QLCI grant OMA-2120757.

\section*{Data availability statement} In this work, we did not generate or analyze any datasets.

\section*{Conflict of interest statement} The authors declare that there are no competing interests with regard to this work.

\bibliographystyle{unsrtnat}

\end{document}